\documentclass[12pt]{article}

\usepackage{amssymb,bbm}
\usepackage{amsfonts}
\usepackage[fleqn]{amsmath}
\usepackage{geometry}
\usepackage{graphicx}
\usepackage{subcaption}
\usepackage{natbib}
\usepackage{multibib}
\newcites{main}{References}
\newcites{app}{References for Appendix}
\usepackage[utf8x]{inputenc}
\usepackage{float}

\usepackage{enumitem}
\usepackage[font=scriptsize,justification=centering]{caption}
\usepackage{setspace}
\usepackage{bigints}
\usepackage[usenames, dvipsnames]{xcolor}
\usepackage{ed}      	
\usepackage{pgf}
\usepackage{tikz}
\usetikzlibrary{patterns}
\usetikzlibrary{decorations.pathreplacing}
\usepackage{hyperref}
\usepackage[noabbrev]{cleveref}
\hypersetup{urlcolor=teal}
\usepackage{url}
\hypersetup{pdftitle = Composite Sorting, pdfauthor=Job Boerma and Aleh Tsyvinski and Ruodu Wang and Zhenyuan Zhang, colorlinks, citecolor=black,filecolor=black,linkcolor=Blue,pdftitle=title,pdfauthor=author,bookmarks, citecolor=Blue,pdftex}
\usepackage{amsthm}
\usepackage{mathrsfs,pgfplots}

\usepgfplotslibrary{fillbetween}
\usepackage{algorithm2e}
\theoremstyle{definition}
\newtheorem{theorem}{Theorem}
\newtheorem*{theorem*}{Theorem}

\newtheorem{corollary}{Corollary}
\newtheorem*{corollary*}{Corollary}

\newtheorem{definition}{Definition}
\newtheorem*{definition*}{Definition}

\newtheorem{lemma}{Lemma}

\newtheorem{proposition}[theorem]{Proposition}

\usetikzlibrary{arrows,decorations.markings,arrows.meta}

\newcommand{\ee}{\varepsilon}

\def\d{\mathrm{d}}

\newcommand{\bone}{ {\mathbbm{1}} }

\newcommand{\bP}{\mathbb{P}}

\usepackage{mathrsfs}

\newcommand{\lp}{\lambda^+}
\newcommand{\lm}{\lambda^-}
\newcommand{\op}{\omega^+}
\newcommand{\om}{\omega^-}

\renewcommand{\ge}{\geqslant}
\renewcommand{\le}{\leqslant}
\renewcommand{\geq}{\geqslant}
\renewcommand{\leq}{\leqslant}
\renewcommand{\epsilon}{\varepsilon}

\tikzset{
    ncbar angle/.initial=90,
    ncbar/.style={
        to path=(\tikztostart)
        -- ($(\tikztostart)!#1!\pgfkeysvalueof{/tikz/ncbar angle}:(\tikztotarget)$)
        -- ($(\tikztotarget)!($(\tikztostart)!#1!\pgfkeysvalueof{/tikz/ncbar angle}:(\tikztotarget)$)!\pgfkeysvalueof{/tikz/ncbar angle}:(\tikztostart)$)
        -- (\tikztotarget)
    },
    ncbar/.default=0.5cm,
}

\tikzset{square left brace/.style={ncbar=0.5cm}}
\tikzset{square right brace/.style={ncbar=-0.5cm}}

\tikzset{round left paren/.style={ncbar=0.5cm,out=120,in=-120}}
\tikzset{round right paren/.style={ncbar=0.5cm,out=60,in=-60}}
\newcommand{\E}{\mathbb{E}}    
\newcommand{\R}{\mathbb{R}}    
\newcommand{\N}{\mathbb{N}}    
\newcommand{\C}{\mathcal{C}} 

\let\OLDthebibliography\thebibliography
\renewcommand\thebibliography[1]{
  \OLDthebibliography{#1}
  \setlength{\parskip}{0pt}
  \setlength{\itemsep}{0.2pt plus 0.3ex}
}

\begin{document}



\title{Composite Sorting\thanks{%
\baselineskip12.5pt We are particularly grateful to Tom Sargent and John Stachurski for many comments and for developing concise exposition and numerical replication (\href{python-advanced.quantecon.org/match_transport.html}{https://python-advanced.quantecon.org/}) as a part of the QuantEcon project \citep{Sargent:2025}. We thank Manuel Amador, Andy Atkeson, Hector Chade, Roberto Corrao, Alfred Galichon, Pieter Gautier, Chad Jones, Philipp Kircher, Rasmus Lentz, Ilse Lindenlaub, Paolo Martellini, Guido Menzio, Giuseppe Moscarini, Andrea Ottolini, Tommaso Porzio, Jean-Marc Robin, Fedor Sandomirskiy, Edouard Schaal, Lones Smith, Stefan Steinerberger, and Alexander Zimin for useful comments.}}
\author{ Job Boerma \\ 
{\small { \hspace{2 cm} University of Wisconsin-Madison  \hspace{2 cm} } }\\ \vspace{-0.5 cm}
\and  Aleh Tsyvinski \\
{\small { \hspace{4 cm}  Yale University \hspace{4 cm} } }\\ \vspace{-0.5 cm}
\and  Ruodu Wang \\
{\small { \hspace{4 cm} University of Waterloo \hspace{4 cm} } }\\ \vspace{-0.5 cm}
\and  Zhenyuan Zhang \\
{\small { \hspace{4 cm}  Stanford University \hspace{4 cm} } }\\ \vspace{0.2 cm}
}
\date{\vspace{-0.5 cm} May 2025 \vspace{0.0 cm}}
\maketitle

\vspace{-0.0cm}

\begin{abstract}
\fontsize{11.5pt}{18.0pt} \selectfont This paper introduces an assignment model with concave costs of skill gaps, which arise generally when firms mitigate
costs of mismatch as in \citet{Stigler:1939} and \citet{Laffont:1986,Laffont:1991}. Concave costs of skill gaps imply that the output function is neither supermodular nor submodular. We thus introduce a tractable model that interpolates between the polar canonical cases of supermodularity and submodularity.  We characterize sorting, wages, and comparative statics and show these substantively differ
from traditional assignment models. Under composite sorting: (1) distinct
worker types work in the same occupation, and (2) worker types are
simultaneously part of both positive and negative sorting. Quantitatively, our model can generate and help explain earnings dispersion between and within occupations. 

\vspace{0.45cm}
\noindent \textbf{JEL-Codes}: J01, D31, C78  
\end{abstract}

\renewcommand{\thefootnote}{\fnsymbol{footnote}} \renewcommand{%
\thefootnote}{\arabic{footnote}} 
\thispagestyle{empty} \setcounter{page}{0}

\fontsize{11.0pt}{21.0pt} \selectfont


\newpage

\section{Introduction}

Sorting models have a prominent position in economics, dating to classic contributions by \citet{Koopmans:1957} and \citet{Becker:1973}; see \cite{Chade:2017} and \cite{Eeckhout:2018} for recent comprehensive reviews. The central insight of this literature is assortative sorting $-$ supermodular output yields positive sorting while submodular output yields negative sorting. Moreover, every worker type is paired with a unique job type, except possibly due to discretization. Existing literature restricts attention to explicitly assortative cases or identifies conditions under which assortative predictions hold (see \citet{Eeckhout:2018b}, \citet{Chade:2018,Chade:2020b,Chade:2017b}, and \citet{Calvo:2024} for recent prominent examples).

An important open question is how to characterize sorting models outside of the polar cases where the output function is supermodular or submodular. In this paper, we introduce a new class of sorting models where the output function is concave in the skill gap between workers and jobs. We show that concave costs of mismatch arise generally when firms mitigate variable costs of mismatch with fixed investment as in \citet{Stigler:1939} and  \citet{Laffont:1986,Laffont:1991}. We formulate our assignment model within the framework of optimal transport theory \citep{Monge:1781,Kantorovich:1942,Villani:2009} and show that this class of economies gives rise to equilibrium sorting, wages, and comparative statics that contrast sharply with $-$ and are significantly richer than $-$ those derived from assortative sorting models. 

In the traditional assortative setting, sorting depends exclusively on the output function and not on the distributions of workers and jobs. In our setting with concave costs of skill gaps, both technology as well as the worker and job distributions determine optimal sorting. This joint dependence poses the main challenge in characterizing optimal sorting, wages, and comparative statics and results in a new sorting pattern that we call \textit{composite sorting}. 

Composite sorting has two main new features that sharply contrast with assortative sorting. First, a given worker type can simultaneously be part of both positive and negative sorting. Second, distinct worker types sort into the same occupation $-$ providing a new mechanism for earnings variation within occupations. In sum, composite sorting introduces an intermediate case between positive and negative sorting that we show is analytically tractable. 


We start by describing two necessary conditions for optimal sorting. First, an optimal assignment maximizes the number of perfect pairs, which are pairs without skill gaps. Since output costs of skill gaps are concave, it is preferable to have one pair with a small and one pair with a large skill gap, as opposed to two pairs with medium skill gaps. A perfect pair together with a pair with a large skill gap exemplifies this. Second, pairs of workers and jobs should not intersect. Visualize two pairs as arcs. When the arcs intersect, two medium-sized skill gaps occur. With concave costs of skill gaps, it is preferable to have one large and one small skill gap $-$ pairings that do not intersect. This non-intersecting condition implies that the sorting problem can be decomposed into layers of independent problems. The layers are derived from a measure of underqualification that evaluates the worker skills compared to the job requirements up to a given skill level. A layer contains an equal number of workers and jobs determined by a particular slice of this measure of underqualification. The overall optimal assignment combines the perfect pairs with the optimal assignments on each layer.


Our main result for the primal optimal assignment problem (\Cref{t:numberpair}) establishes a significant reduction in complexity when the distributions of workers and jobs are mixtures of normal distributions.\footnote{Normal mixture distributions are weakly dense in the set of all distributions.} This result enables our quantitative analysis with a large number of worker and job types. \Cref{t:numberpair} shows that if the worker distribution is a mixture of $n$ normal distributions and the job distribution is a mixture of $m$ normal distributions, then each layer contains at most $n+m-1$ pairs. This result builds on the theory of variation diminishing transformations and P\'{o}lya frequency functions \citep{Schoenberg:1930,Schoenberg:1950,Karlin:1968}. These technical tools relate to those used in recent work on stochastic dominance \citep{Pomatto:2020}, on matching and information disclosure \citep{Chade:2024}, and on multinomial stochastic choice rules \citep{Sandomirskiy:2023,Sandomirskiy:2024}. While the direct computation of the assignment problem is infeasible for a non-trivial number of types,\footnote{See, for example, the Hungarian algorithm in \citet{Burkard:2012}.} \Cref{t:numberpair} delivers a significant reduction in complexity for the primal assignment problem that enables our quantitative analysis with many worker and job types.

Our second main contribution is the full characterization of equilibrium wages and firm values – the solution to the dual assignment problem (\Cref{thm:dualalg}, \Cref{prop:hiereff}, and \Cref{p:dual}). A striking outcome is that wages and firm values exhibit a \textit{regional hierarchical structure}, fundamentally different from the classical sorting models. 

We first show that the dual solution can be segmented by skill regions, within which the equilibrium wages and firm values are determined by regional conditions. In other words, for any given skill group, the relative earnings of workers in that group depend only on the output and the assignment within that group. Relative earnings in one skill region are independent of the workers, jobs, and output in other skill regions. Second, we establish that the hierarchical structure aggregates the regional wages to wages for larger groups, preserving relative wages in smaller regions. This hierarchical assembly ensures that global dual feasibility holds when we stitch together regional solutions. The main technical challenge in the construction of the dual is to reconcile local wage determination with global constraints, and our construction shows explicitly how to achieve this. Theorem \ref{prop:hiereff} shows that our hierarchical solution implies a sharp reduction in complexity due to a limited number of distinct regional structures. This result makes the dual problem numerically tractable. 


Our third main contribution is characterizing comparative statics. In traditional assignment models, sorting is determined solely by the output function and is thus invariant to technological change as long as the output function remains supermodular or submodular. In sharp contrast, composite sorting depends on both the output function as well as the distributions of workers and jobs. This joint dependence poses new challenges for comparative statics, which leads us to developing a different approach leveraging the characterization of equilibrium sorting with concave costs of skill gaps.

On the characterization of comparative statics, we obtain two significant new results. First, sorting becomes more positive, by which we mean larger in concordance order, when the cost of skill gaps is less concave (\Cref{p:cs}). In order to show this, we provide a new characterization of the classical cyclical monotonicity specific to concave costs of skill gaps. Second, we prove the existence of a threshold level of concavity beyond which sorting is positive in each layer  $-$ yet still not an overall positive sorting $-$ which we call a \textit{layered positive sorting} (\Cref{t:tolinear}). 



Finally, we quantitatively illustrate our framework. One of the distinctive implications of our model among assignment models is that equilibrium features earnings dispersion within occupations and, hence, we apply our model to evaluate earnings dispersion within and across occupations. The quantitative model analyzes the implications of concave costs of skill gaps for sorting and earnings dispersion within occupations in isolation. We also use our framework to quantify the determinants of changes in earnings in the United States between 1980 and 2005. 



\vspace{0.35cm}
\noindent \textbf{Related Literature}. The understanding of sorting when output is neither submodular nor supermodular is limited; see \citet{Chade:2017} and \cite{Eeckhout:2018} for recent reviews. 

Closest to this paper, \citet{Sargent:2025} develops a textbook exposition and computational replication of our framework. Their lecture also presents detailed Python code for computing the equilibrium assignment, wages, and firm values.\footnote{See the QuantEcon lecture at \href{python-advanced.quantecon.org/match_transport.html}{https://python-advanced.quantecon.org/}.} 


Two recent papers are also particularly relevant to our paper, as they characterize optimal transport models with non-convex costs. 

First, an influential paper by \citet{Fajgelbaum:2020} studies a different but related equilibrium transport problem where non-convexity plays an important role. They consider optimal transportation on a network where the planner can invest in infrastructure to mitigate the cost of transporting goods along specific edges. Mitigation of costs via investment, as in our work, determines the concavity of the cost function. When the planning problem is convex, \citet{Fajgelbaum:2020} characterize the equilibrium using convex duality. In contrast, when the planning problem is non-convex, they argue that the problem becomes significantly more complex. For the case when there is a unique commodity produced in a single location, they prove that the optimal transport network is a tree. Moreover, for the non-convex costs they are also able to numerically solve several more general cases and show how the optimal network concentrates transport flows along major routes. Our problem is different as the cost is concave in the skill gap, as opposed to being concave in flows as in \citet{Fajgelbaum:2020}. For our setting with concave costs of skill gaps, we provide a complete analytical characterization of the primal solution, the dual solution, and comparative statics and use this characterization to compute the equilibrium. 


Second, the most technically related paper is \citet{Echenique:2024} who study stability in non-transferable utility matching markets with aligned preferences, compared to transferable utility markets in our paper. Their main result establishes a connection between stability and a concave optimal transport problem, which leverages the non-crossing property to derive a non-combinatorial characterization of their primal problem. In addition, \citet{Perez:2022,Perez:2025} study a falsification-proof mechanism with a cost function that may be convex or concave in the distance between the natural score and the falsified score. They show that this mechanism design problem can be represented as an optimal transport problem and study its dual representation and comparative statics.

%
%
%
%

This paper is part of a growing literature that builds on optimal
transport theory to solve economic problems (see \citet{Villani:2009}, \citet{Galichon:2018}, and \citet{Sargent:2024} for comprehensive overviews). One
example is recent work on multimarginal assignment problems \citep{Chade:2018,Eeckhout:2018b,BTZ:2021} in which
multiple agents work together in a team as in \citet{Kremer:1993}.\footnote{\citet{Kremer:1996} studies a role assignment model where a single population of workers is split to work into teams as managers and assistants. \citet{Anderson:2022a} provides the most comprehensive analysis for this class of models and shows that positive clustering is optimal. \citet{Porzio:2017} uses optimal transport theory to study a model with a technology decision that scales output produced by the team, similarly to productive capital.} There is mathematical work on optimal transport with concave distance
costs started by \citet{Gangbo:1996} and \citet{McCann:1999}
as well as literature on algorithmic sorting problems with distance
costs \citep{Werman:1986,Aggarwal:1995,Delon:2012b,Ottolini:2023}.   
Our first contribution to this literature is to derive a new characterization of the primal problem for normal mixture distribution that significantly reduces the complexity of the problem and facilitates our quantitative analysis. Our second contribution is to provide a full construction of the dual solution in a model of concave cost of skill gaps. The regional hierarchical structure of the dual solution is central to the analysis of equilibrium wages. Third, our approach to comparative statics, leveraging the characterization of the optimum, differs from \citet{Anderson:2022} as their conditions for more positive sorting are not satisfied in our setting. Instead, our analysis of the primal problem leverages the use of the variation diminishing property that is also important in recent work of \citet{Pomatto:2020}, \citet{Sandomirskiy:2023,Sandomirskiy:2024}, and \citet{Chade:2024}.

\section{Model}

In this section we develop a sorting model in which investment mitigates the cost of mismatch between workers and jobs and results in concave costs of skill gaps. This is a setting where the output function is neither supermodular nor submodular. Importantly, we thus introduce an environment that interpolates between the canonical cases of supermodularity and submodularity that yield assortative sorting. 

\subsection{Environment}\label{s:environment}

The economy is populated by risk-neutral workers and jobs. The workers differ in skills indexed by $x$. The set of worker skills $X$ is a finite number $n$ of types $x_{1}<x_{2}<\dots<x_{n}$. Workers are distributed according to the cumulative distribution function $F(x)$.

Jobs differ in difficulty indexed by $z$. The set of occupations $Z$ is a finite number $m$ of occupation types $z_{1}<z_{2}<\dots<z_{m}$. Jobs
are distributed according to the cumulative distribution function $G(z)$.\footnote{For our main results, it does not matter whether the distributions of workers and jobs are discrete or continuous. In order to avoid presenting non-essential technical details, we present all results in the simplest setting.}




Firms produce a single good. Production requires one worker for each job. The surplus generated by a worker with skill $x$ in an occupation with complexity $z$ is:
\begin{equation}
s(x,z)= \alpha(x) + \theta(z) -\gamma_{p}\max ( z-x,0 ) -\gamma_{u}\max ( x-z,0 ),\label{e:firmtechnology}
\end{equation}
with $\gamma_{p}, \gamma_u \geq0$. There are four terms in this technology specification. The first term $\alpha(x)$ with $\alpha'(x)>0$ reflects that a more skilled worker contributes more to production, independent of the job. The second term $\theta(z)$ with $\theta'(z)>0$ reflects that a more difficult job $z$ produces more output independent of the worker that fulfills the job. The third and the fourth terms reflect the costs of skill gaps, which is the difference between worker skill $x$ and job complexity $z$. The third term $\gamma_{p}\max(z-x,0)$ reflects that a worker with a skill $x$ that is lower than the job complexity $z$ causes a loss of output. It is costly to have workers perform tasks for which they have limited talent. The fourth term $\gamma_{u}\max(x-z,0)$ reflects that workers with skills $x$ that exceed the job complexity $z$ are overqualified and need to be compensated for their utility cost (as in \citet{Rosen:1986}).


\vspace{0.35cm}
\noindent \textbf{Concave Mismatch Costs}. A firm can reduce mismatch costs by making investments. The main insight of \citet{Stigler:1939} and \citet{Laffont:1986,Laffont:1991} is that fixed investment results in an effective output function with a concave costs of skill gaps. 


We model a firm making fixed investments to reduce variable costs of skill gaps. Consider the case where a worker is underqualified, $x<z$. Firms choose the variable
cost of production mismatch $\gamma_{p}$, which comes at an associated
fixed cost $\Psi_{p}(\gamma_{p}) = \frac{1}{\eta_{p}}\gamma_{p}^{-\eta_{p}}$, where $\eta_{p}$ is strictly positive.\footnote{General convex cost functions are considered in Technical Appendix \ref{s:generalprod}.} By decreasing variable costs
$\gamma_{p}$, the firm increases its fixed costs, or $\Psi_{p}'<0$,
where $\Psi_{p}''>0$. The effective output of worker $x$ in occupation $z > x$ is: 
\begin{equation}
y(x,z)= \hspace{0.15 cm}\max_{\gamma_{p} \geq 0}\hspace{0.25 cm}  \alpha(x) + \theta(z) -\gamma_{p}(z-x)-\frac{1}{\eta_{p}}\gamma_{p}^{-\eta_{p}} .\label{e:technologychoice}
\end{equation}
Investment increases in the difference between the worker skill and the job complexity $\gamma_{p}=(z-x)^{-\frac{1}{1+\eta_{p}}}$. Firms choose a low variable cost of production mismatch if the worker is less qualified, that is, when the skill gap $(z-x)$ is large.


The effective output of worker $x$ in occupation $z$ is, using the optimal investment decision, given by: 
\begin{equation}
y(x,z)= \alpha(x) + \theta(z) -\frac{1}{\zeta_p}(z-x)^{\zeta_{p}}. 
\end{equation}
for underqualified workers $x<z$, with $\zeta_{p}=\frac{\eta_{p}}{1+\eta_{p}}  \in(0,1)$. The cost is concave in the distance between worker skill and job difficulty. In sum, when the marginal costs of mitigating mismatch are increasing, a production function with linear mismatch costs and investment choice results in an effective output function with concave costs of skill gaps.\footnote{In line with our framework, \citet{Brynjolfsson:2025} and \citet{Noy:2023} find that the introduction of generative artificial intelligence tools increases the output of low-skill workers with minimal impact on high skill workers.}

Firms can similarly reduce the extent to which overqualification penalizes worker's utility by providing amenities. We model amenity choices analogous to investment choices. The effective output of overqualified workers is $y(x,z)=\alpha(x) + \theta(z) - \frac{1}{\zeta_{u}}(x-z)^{\zeta_{u}}$, with $\zeta_{u}=\frac{\eta_{u}}{1+\eta_{u}}\in(0,1)$.

The effective output is: 
\begin{equation}
y(x,z)=\alpha(x) + \theta(z)-\begin{cases}
\frac{1}{\zeta_{p}}(z-x)^{\zeta_{p}}\hspace{4.6cm}\text{ if }z\geq x\\
\frac{1}{\zeta_{u}}(x-z)^{\zeta_{u}}\hspace{4.59cm}\text{ if }z<x,
\end{cases}\label{e:y}
\end{equation}
where $\zeta_p,\zeta_u \in(0,1)$. We use the effective output (\ref{e:y}) to define the effective costs of skill gaps between worker $x$ and job $z$ as: 
\begin{equation}\label{eq:cxz}
c(x,z) = \alpha(x) + \theta(z) - y(x,z) = \begin{cases}
\frac{1}{\zeta_{p}}(z-x)^{\zeta_{p}}\hspace{3cm}\text{ if }z\geq x\\
\frac{1}{\zeta_{u}}(x-z)^{\zeta_{u}}\hspace{2.99cm}\text{ if }z<x.
\end{cases}
\end{equation} 
The cost of skill gaps is the maximal output of worker $x$ and job $z$ minus effective output $y(x,z)$. Thus, the cost function is concave in the skill gap, the discrepancy between worker $x$ and job $z$. 

\vspace{0.05cm}
\begin{definition} 
An \underline{assignment} pairs workers and jobs.  Given a worker distribution $F$ and a job distribution $G$,
the set of feasible assignment functions is $\Pi =\Pi(F,G)$, which
is the set of probability measures on the product space $X\times Z$
such that the marginal distributions of $\pi$ onto $X$ and $Z$
are respectively $F$ and $G$. We denote the support of assignment $\pi$ by $\Gamma_\pi=\{(x,z):\pi(\{(x,z)\})>0\}$. 
\end{definition}

\vspace{0.1 cm}
\noindent \textbf{Discussion}. The key feature of our environment is that the effective output (\ref{e:y}) is neither supermodular nor submodular. The cross-derivatives of the output function being negative for both $z>x$ and $z<x$ directly rules out supermodularity. Moreover, the output function is not submodular. Consider two workers and two jobs, each with skills $a$ and $b$ where $b \neq a$. Submodularity requires the combined surplus of pairs $(a,b)$ and $(b,a)$ to exceed the combined surplus of pairs $(a,a)$ and $(b,b)$. However, pairs $(a,b)$ and $(b,a)$ have skill gaps and consequently lower output than the positive sorting $(a,a)$ and $(b,b)$ that gives no skill gaps. Thus, the output function is not submodular either.\footnote{While our analysis relies on the concavity of mismatch costs, it does not require strict concavity. For example, a technology where underqualified workers linearly induce mismatch losses up to some maximum is concave, not strictly concave. Furthermore, zero cost of skill gaps is a concave function. This allows us to extend our analysis to the case where workers do not incur disutility from being overqualified, or $\gamma_u =0$ in the technology (\ref{e:firmtechnology}). Moreover, we can extend our production technology to allow for additional fixed costs of mismatch that are incurred when a pair is not perfect, or $x \neq z$. We analyze the uniqueness of sorting with strictly concave costs of skill gaps in Technical Appendix \ref{s:uniquenesssorting}.} 





\subsection{Planning Problem} \label{s:planning} 

We solve two planning problems to characterize an equilibrium.\footnote{The equilibrium definition is standard and is presented in Technical \Cref{s:equilibriumdefn} for completeness.} We first solve a primal planning problem to characterize an equilibrium assignment. The primal planning problem is to choose an assignment to maximize aggregate output: 
\begin{equation}
\max_{\pi\in\Pi}\int y(x,z)\,\text{d}\pi\label{pp}
\end{equation}
and is equivalent, in terms of
choosing an optimal assignment, to a planning problem that minimizes
the costs of mismatch: 
\begin{equation}
\min_{\pi\in\Pi}\;\int c(x,z)\,\text{d}\pi,\label{ppmin}
\end{equation}where $c(x,z)$ represents the concave cost of skill gaps \eqref{eq:cxz}. The key difference from classical assignment problems is that the cost function \eqref{eq:cxz} is neither supermodular nor submodular.


\vspace{0.35cm}
\noindent \textbf{Dual Problem}. In order to obtain equilibrium wages $w$ and the firm value function $v$, we solve a
dual problem. The dual problem
is to choose functions $w$ and $v$ that solve: 
\begin{equation}
\min \;\int w(x)\,\text{d}F+\int v(z)\,\text{d}G,\label{e:pp_dual}
\end{equation}
subject to the constraint $w(x)+v(z)\geq y(x,z)$ for any $(x,z)\in X\times Z$. The Monge-Kantorovich duality states that the values of (\ref{pp}) and (\ref{e:pp_dual}) are the same: $\max \int y(x,z)\,\text{d}\pi= \min \int w(x) \,\text{d}F+\int v(z)\, \text{d}G$.

\vspace{0.3 cm}
\noindent We use the following relation between the planning problem and the dual problem.

\begin{lemma}\label{lemma:dual}
Suppose that assignment $\pi \in \Pi$ and functions $w$ and $v$ are such that $w(x)+v(z)\geq y(x,z)$ for any $(x,z)$ and that $w(x)+v(z)= y(x,z)$ for any $(x,z)\in\Gamma_\pi$. Then the assignment $\pi$ is an optimal assignment and $(w,v)$ is an optimal dual pair.\footnote{Suppose the assignment $\pi$ and the functions $(w,v)$ satisfy the assumptions in Lemma \ref{lemma:dual}. Then it holds that $\int y(x,z)\,\text{d}\pi=\int w(x) \,\text{d}F+\int v(z) \,\text{d}G$. By linear duality, the maximum for the primal problem is attained by $\pi$ and the minimum for the dual problem is attained by $(w,v)$, as required.}    
\end{lemma}

\section{Composite Sorting}


This section introduces composite sorting. First, we establish that, in contrast to the classic assortative models, optimal sorting depends not only on the production function but also on the distributions of workers and jobs. The joint dependence on both the production function and the distributions introduces significant challenges in characterizing sorting, wages, and comparative statics. Second, we provide a stylized example in which composite sorting arises to develop intuition for \Cref{s:principles}.

\vspace{0.35cm}
\noindent \textbf{Assortative Sorting}. In the classic assortative setting, sorting is either positive or negative depending on the production function and does not depend on the distributions of workers and jobs. We first make an important observation that optimal sorting in our environment with concave costs of skill gaps, rather than being determined by the production function alone as in the classic assortative problems, also depends on the distributions of workers and jobs. 

This shows the main challenge in the environment with concave costs of skill gaps $-$ both technology as well as the worker and job distributions determine optimal sorting. In order to make this point, we show that for the same concave costs of skill gaps, the optimal sorting pattern differs due to differences in the distributions of workers and jobs. Specifically, we show that our sorting problem, with the output function that is neither supermodular nor submodular, can feature positive and negative sorting for different distributions of workers and jobs. 

Consider a problem with two workers and two jobs. Worker skills are given by $x_{L}$ and $x_{H}$ and job difficulties are given by $z_{L}$ and $z_{H}$ satisfying $x_{L}<z_{L}<x_{H}<z_{H}$. Let the distance between $x_{i}$ and $z_{j}$ be given by $d_{ij}:=|x_{i}-z_{j}|$.


First, positive sorting can be optimal. Consider the following configuration of distances: $d_{LL}^{\zeta}+d_{HH}^{\zeta} \leq d_{LH}^{\zeta}+d_{HL}^{\zeta}$  with the cost function $c(x,z) = |x-z|^\zeta$ with $\zeta \in (0,1)$.  The low-skill worker $x_{L}$ and the low-complexity job $z_{L}$ as well
as the high-skill worker $x_{H}$ and the high-complexity job $z_{H}$
are close to each other, while the skill gap between the
low-complexity job $z_{L}$ and high-skill worker $x_{H}$ is large. It is natural to pair the low-skill worker with the low-complexity job and to pair the high-skill worker with the high-complexity job to minimize the costs of skill gaps and, hence, sorting is positive. That is, when worker-job skill groups are far apart, it is optimal to sort within those groups. 


Second, negative sorting can be optimal. Consider the opposite
configuration of distances where $d_{LL}^{\zeta}+d_{HH}^{\zeta}>d_{LH}^{\zeta}+d_{HL}^{\zeta}$.
In this case, the high-skill worker and the low-complexity job are close
to each other, while the distance between the low-skill worker and
the low-complexity job as well as the distance between the high-skill worker
and the high-complexity job is large. Since the cost of skill gaps is concave in the distance between the worker's skill and the job's complexity, it is optimal to pair the high-skill worker with the low-complexity job $-$ having one small skill gap and one large
skill gap is better than having two medium-sized skill gaps. 

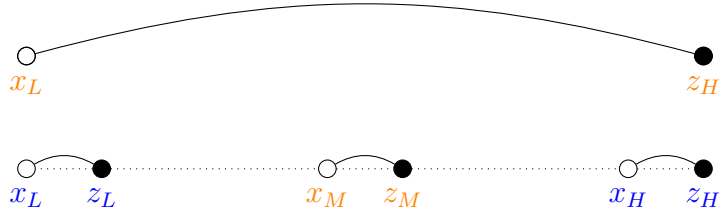
\begin{figure}[!t]
\begin{centering}
\begin{tikzpicture}

\node[circle,draw, minimum size=0.1pt,scale=0.6,label=below:{}] (A) at  (0,0){};
\node[circle,fill=black,draw, minimum size=0.1cm,scale=0.6,label=below:{$\textcolor{blue}{z_L}$}] (B) at  (1.,-1.5) {} ;

\node[circle,draw, minimum size=0.1pt,scale=0.6,label=below:{$\textcolor{blue}{x_L}$}] (C) at  (0,-1.5){};
\node[circle,fill=black,draw, minimum size=0.1cm,scale=0.6,label=below:{$\textcolor{orange}{z_M}$}] (B3) at  (5.,-1.5) {} ;
\node[circle,draw, minimum size=0.1pt,scale=0.6,label=below:{$\textcolor{orange}{x_M}$}] (C3) at (4.,-1.5) {};
\node[circle,fill=black,draw, minimum size=0.1cm,scale=0.6,label=below:{$\textcolor{blue}{z_H}$}] (B2) at  (9,-1.5) {} ;
\node[circle,draw, minimum size=0.1pt,scale=0.6,label=below:{$\textcolor{blue}{x_H}$}] (C2) at  (8,-1.5){};
\node[circle,fill=black,draw, minimum size=0.1cm,scale=0.6,label=below:{}] (D) at  (9,0) {} ;

\path[-,every node/.style={font=\sffamily\small}] (C) edge[bend left=30] node [left] {} (B);
\path[-,every node/.style={font=\sffamily\small}] (C3) edge[bend left=30] node [left] {} (B3);
\path[-,every node/.style={font=\sffamily\small}] (C2) edge[bend left=30] node [left] {} (B2);
\path[-] (D) edge[bend right=15] node [yshift=0.2cm,above] {} (A);

\node[circle,draw, minimum size=0.1pt,scale=0.6,label=below:{$\textcolor{orange}{x_L}$}] (F) at  (0,0){};
\node[circle,fill=black,draw, minimum size=0.1cm,scale=0.6,label=below:{$\textcolor{orange}{z_H}$}] (H) at  (9,0){};

\draw[dotted] (B) -- (C);\draw[dotted] (B) -- (C3);\draw[dotted] (C3) -- (B3);\draw[dotted] (B3) -- (C2);\draw[dotted] (C2) -- (B2);

        \end{tikzpicture} 
\par\end{centering}
\caption{An Example of Composite Sorting}
\label{f:examplecs} {\scriptsize{}{}\vspace{0.2cm}
 Figure \ref{f:examplecs} illustrates composite sorting with four
workers and four jobs. First, distinct worker types are paired with
identical occupations. A low-skill worker $x_{L}$ and a high-skill
worker $x_{H}$ both work in the identical high-complexity occupation
$z_{H}$, while the medium-skill worker $x_{M}$ does not work in
this occupation. Second, a worker type is simultaneously part of both
positive and negative sorting. A low-skill worker $x_{L}$ is paired
positively with a low-complexity job $z_{L}$ (part of positive sorting
$(x_{L},z_{L})$ and $(x_{H},z_{H})$ in blue) and the same worker
type is paired negatively to a distant high-complexity job $z_{H}$
(part of negative sorting $(x_{L},z_{H})$ and $(x_{M},z_{M})$ in
orange).} 
\end{figure}

\vspace{0.35cm}
\noindent \textbf{Composite Sorting}. We next show that our environment gives rise to a new sorting pattern, which we call composite sorting, that is significantly richer than assortative sorting. Composite sorting has two characteristic features: distinct worker types work in the same occupation, giving rise to earnings dispersion within occupations, and a given worker type can be simultaneously part of positive and negative sorting. In order to introduce composite sorting, we first consider an assignment problem between three workers and three jobs in the bottom half of \Cref{f:examplecs}. Since the skill groups are far apart, it is optimal to positively sort within groups.

The top half of \Cref{f:examplecs} introduces an additional low-skill worker $x_L$ and an additional high-complexity task $z_H$. One could break the medium-skill worker-job pair $(x_M,z_M)$ such that the added low-skill worker is assigned to the medium-complexity job forming a pair $(x_L,z_M)$ and the added high-complexity job is assigned to the medium-skill worker forming a pair $(x_M,z_H)$. This gives two medium-sized skill gaps. Instead, pairing the added low-skill worker to the added high-complexity job forming $(x_L,z_H)$, while preserving the medium-skill pair $(x_M,z_M)$, results in one small skill gap and one large skill gap, which is preferred by the concavity of the mismatch cost. The added low-skill worker $x_L$ and high-complexity job $z_H$ are thus optimally paired as indicated by the arc in Figure \ref{f:examplecs}.

The optimal assignment features composite sorting. First, distinct worker types are paired to identical jobs. In \Cref{f:examplecs}, a low-skill worker $x_L$ and a high-skill worker $x_H$ both work in the identical high-complexity occupation $z_H$, while the medium-skill worker $x_M$ does not work in this occupation. Second, a worker type is simultaneously part of positive as well as negative sorting. In \Cref{f:examplecs}, a low-skill worker $x_{L}$ is paired positively with the low-complexity job $z_{L}$ (part of the positive sorting $(x_{L},z_{L})$ and $(x_{H},z_{H})$ in blue) and the same worker type is also paired negatively to the high-complexity job $z_{H}$ (part of the negative sorting $(x_{L},z_{H})$ and $(x_{M},z_{M})$ in orange). This example shows that concavity of the mismatch function and the distributions of workers and jobs jointly determine the optimal assignment and this gives rise to a new sorting pattern $-$ composite sorting.





\section{Characterizing Optimal Sorting} \label{s:principles}

This section characterizes optimal sorting. \Cref{s:necessarycond} describes necessary conditions for optimal sorting. \Cref{t:numberpair} in \Cref{s:nm} proves a significant reduction of complexity of the assignment problem. The formal statements and the proofs are in \Cref{a:necconos}.

\subsection{Necessary Conditions} \label{s:necessarycond}

We start by describing two necessary conditions for optimality: (1) maximal number of perfect pairs, and (2) no intersecting pairs.\footnote{For ease of exposition, we analyze the case of finitely many skill levels. The same intuition naturally carries over to continuous distributions.} 



\vspace{0.35cm}
\noindent An optimal assignment maximizes
the number of pairs that are perfectly sorted, i.e., the number of pairs with no skill gap between workers and jobs, or $x=z$. When the costs of skill gaps are concave, it is preferred to have a pair with a small skill gap and a pair with a significant skill gap rather than to have two pairs with medium skill gaps. A perfect pair is an example of this since it has no skill gap.\footnote{We remark that strictly convex costs of skill gaps instead implies that positive sorting is optimal, which generally conflicts with maximal perfect pairing.} Maximal perfect pairing shows that workers and jobs that are part of the common component of the worker and job distributions are positively sorted. In analyzing the sorting problem between remaining workers and jobs we can thus consider assignments between worker and job distributions for which the common components are removed. For brevity, we label the remaining worker distribution $F$ and the remaining job distribution $G$.

The second feature of an optimal assignment is that pairings between workers and jobs do not intersect. We first describe intersecting and non-intersecting pairs. Consider two pairs and visualize their pairings by arcs. We refer to pairs as intersecting when the arcs corresponding to the pairs intersect. When the arcs do not intersect, they are non-intersecting pairs.\footnote{More formally, arcs $(x,z)$
and $(x',z')$ do not intersect if and only if the intervals $(x,z)$
and $(x',z')$ are either disjoint or one interval is a subset of
the other interval. When referring to an interval $(x,z)$, we do not require that the
worker skills and job complexity are ordered: we mean
the set of numbers between $z$ and $x$ on the real line.} 



For any two pairs $(x,z)$ and $(x',z')$ in an optimal sorting, their arcs do not intersect. Specifically, consider two
unique configurations with intersecting pairs.\footnote{There are six distinct orderings of workers $x$'s (white circles)
and jobs $z$'s (black dots) to consider, which can be represented
as: 
\includegraphics[width=1.5cm,height=0.6cm]{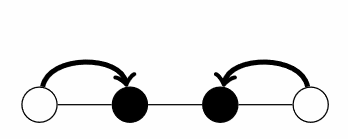}
\includegraphics[width=1.5cm,height=0.6cm ]{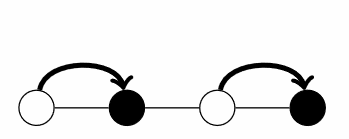}
\includegraphics[width=1.5cm,height=0.6cm ]{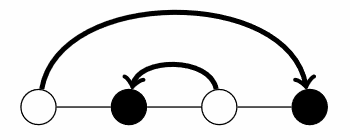}
\includegraphics[width=1.5cm,height=0.6cm ]{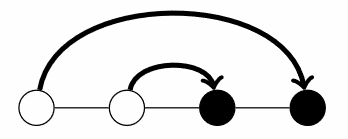}
\includegraphics[width=1.5cm,height=0.6cm ]{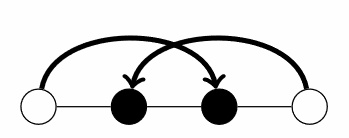}
\includegraphics[width=1.5cm,height=0.6cm ]{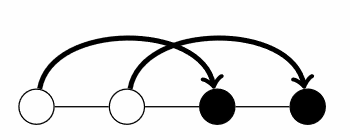}.
The first four configurations do not contain intersections. The
final two configurations are discussed in the main text.} The first configuration is $x<z'<z<x'$. Since the mismatch cost
increases in the skill gap, an improvement is
to instead pair the closer points $(x,z')$ and $(x',z)$, as it reduces
the mismatch cost for each worker and hence the total mismatch cost.
The second configuration is $x<x'<z<z'$. In this case, we
change the pairing to one large skill gap $(x,z')$
and one small skill gap $(x',z)$. By concavity, the cost is
smaller than two medium-size skill gaps.

A direct consequence of the no-intersection principle is that an assignment problem can be  decomposed into layers. In order to establish layering, first observe that if worker $x$ is paired with job $z$, then there is an identical number of workers and jobs in between skill levels $x$ and $z$. Suppose the number of workers and the number of jobs in between skill levels $x$ and $z$ are not the same. Then there is a worker $\hat{x}$ that cannot be paired with a job inside the skill interval $(x,z)$. As a result, worker $\hat{x}$ has to be paired with a job outside the interval $(x,z)$, which would lead to the intersection of pairs $(x,z)$ and $(\hat{x},\hat{z})$, contradicting no intersecting pairs. We conclude that there is the same number of workers and jobs between optimally paired worker $x$ and job $z$, $F(z)-F(x)=G(z)-G(x)$, or alternatively $F(z)-G(z)=F(x)-G(x)$. 

We use the observation that the number of workers and jobs between an optimally paired worker $x$ and job $z$ is identical to decompose the overall sorting problem into sorting problems for different layers of the measure of underqualification. The measure of underqualification $H=F-G$ defines the extent to which workers up to skill $s$ outnumber the jobs requiring skills up to $s$. Since the number of workers and jobs in between optimally paired workers and jobs is identical, only workers and jobs within the same layer of the measure of underqualification can be paired in an optimal assignment. An optimal assignment between workers and jobs is thus the sum of optimal assignments for each layer of the measure of underqualification $H$. This observation decomposes the original problem into independent problems for each layer, which we formalize in \Cref{l:layer}.\footnote{In each layer, there is an alternating configuration of workers and jobs $-$ every worker skill level is followed by a job difficulty level, possibly except for the last one.}


In order to understand the statement of layering (Lemma \ref{l:layer}), we define the measures of the workers and the jobs in each layer $\ell$, which we denote by $F_{\ell}$ and $G_{\ell}$ respectively. In order to obtain the layers of the measure of underqualification, we identify the skill levels where underqualification increases and the skill levels where underqualification decreases. The underqualification
measure $H$ takes a finite number of values in $a_{0}<a_{1}<\dots<a_{L}$ since there is a finite number of skill types. Underqualification
increases from $a_{\ell}$ to $a_{\ell+1}$ at a skill level $s$
if $H(s_{-})\leq a_{\ell}<a_{\ell+1}\leq H(s)$, where $s_{-}$ represents
the limit from the left. Analogously, the measure of underqualification $H$
decreases from $a_{\ell+1}$ to $a_{\ell}$ at a skill level $s$
if $H(s_{-})\geq a_{\ell+1}>a_{\ell}\geq H(s)$. The set of skill
levels where underqualification increases is denoted by $S_{\ell}^{\uparrow}:=\{\text{\ensuremath{s}: \ensuremath{H} increases from \ensuremath{a_{\ell-1}} to \ensuremath{a_{\ell}} at skill level \ensuremath{s}}\}$
for all $1\leq\ell\leq L$. Similarly, the set of skill levels where
underqualification decreases is denoted by $S_{\ell}^{\downarrow}:=\{\text{\ensuremath{s}: \ensuremath{H} decreases from \ensuremath{a_{\ell}} to \ensuremath{a_{\ell-1}} at skill level \ensuremath{s}}\}$.
We then define the discrete measures of workers and jobs for all layers $1\leq\ell\leq L$ by: 
\begin{equation*}
F_{\ell}:=(a_{\ell}-a_{\ell-1})\sum_{s \in S_{\ell}^{\uparrow}}\delta_{s}\hspace{2.7cm}\text{ and }\hspace{2.7cm}G_{\ell}:=(a_{\ell}-a_{\ell-1})\sum_{s \in S_{\ell}^{\downarrow}}\delta_{s} \;,
\end{equation*}
where $\delta_{s}$ is the Dirac measure at a particular skill $s$. A layer $\ell$ contains all the points where the measure of underqualification is valued in $(a_{\ell-1},a_\ell)$, which thus has measures of workers and jobs given by $F_\ell$ and $G_\ell$. It follows that
$F=\sum F_{\ell}$ and $G=\sum G_{\ell}$.  

\begin{lemma}{\textit{Layering}.}\label{l:layer} Let $\pi_{\ell}$
be an optimal assignment between the worker distribution $F_{\ell}$ and the job distribution $G_{\ell}$
for the layer $\ell\in\{1,\dots,L\}$. Then, an optimal assignment
between workers $F$ and jobs $G$ is their sum, $\pi:=\sum_{\ell}\pi_{\ell}$.
\end{lemma}

\noindent A formal statement of the layering feature and a proof are presented in Technical Appendix \ref{p:layer}.



\vspace{0.4 cm}
\noindent \textbf{Discussion}. The impossibility of crossing first appeared in \citet{Monge:1781}, as discussed by \citet{Villani:2009}. Non-crossing arcs are also a central feature of algorithmic sorting problems with distance costs \citep{Aggarwal:1995,Werman:1986} and of the literature on optimal transportation with concave distance costs \citep{Gangbo:1996,McCann:1999}. In the economics literature, \citet{Echenique:2024} also leverages the property of non-crossing pairs to characterize their primal problem. Perfect pairing is a consequence of the non-crossing pairs and is referred to as “mass stays in place if it can” \citep{Gangbo:1996,Villani:2003}. The observation that the impossibility of crossing arcs implies that the assignment problem can be decomposed into a series of independent problems, or layers, is first made in \citet{Aggarwal:1995} and used in \citet{Delon:2012b}.\footnote{A central result in optimal transport is Brenier's theorem that relates the optimal map to the gradient of a convex function \citep{Brenier:1991,Villani:2003}. When Brenier's theorem holds, additional properties of the solution can be derived for the concave costs of skill gaps \citep{McCann:1999,Pegon:2015} and general increasing costs of skill gaps \citep{Gangbo:1996,Villani:2003,Santambrogio:2015,Clark:2023}. In our setting, Brenier's theorem does not apply.} 

\subsection{Characterization} \label{s:nm}

In this section, we provide a characterization of the optimal assignment that results in significant reduction of complexity that applies to an important class of normal mixture distributions. Specifically, we build on the theory of variation diminishing transformations to provide a sharp upper bound to the number of pairs per layer. The key object in this section is the excess skill supply function that captures the extent to which workers outnumber jobs at a particular skill level. 

In this section we consider continuous distributions for worker skills and jobs for ease of exposition. Throughout the paper, we consider discretizations of continuous distributions on $N_s$ ordered skill levels $\{ s_i \}$. The discretized cumulative distribution function equals the original distribution function at each skill level $\{s_i\}$, or $\hat{F}(s_i) = F(s_i)$, except for the final skill where both discretized cumulative distribution functions are equal to one.\footnote{Formally, the discretized cumulative distribution function for workers is described by $\hat{F}(x) = F(x_{n})$ if $x_{n} \leq x < x_{n+1}$ for all $1 \leq n \leq N-1$, $\hat{F}(x) = 0$ if $x<x_1$ and $\hat{F}(x) =1$ when $x \geq x_N$. The discretized distribution $\hat{F}$ and $\hat{G}$  are increasing and right-continuous with left limits.}


\begin{definition}
Let the density functions for workers and for jobs be respectively denoted by $f$ and $g$. The excess skill supply function $h$ is given by the difference $h = f - g$. 
\end{definition}

\noindent When workers are in excess supply, the sign of the excess skill supply function is positive. When jobs are in excess supply, the sign of the excess skill supply function is negative.


\vspace{0.4 cm}
\noindent A necessary condition to obtain another crossing of a given layer is that the measure of underqualification $H$ changes direction. The direction of the measure of underqualification is given by the excess skill supply function $h$. The measure of underqualification increases when workers are in excess supply and decreases when jobs are in excess supply. This shows that the maximum number of points in each layer is less than or equal to the maximum number of directions of the measure of underqualification $H$, or signs of the excess skill supply function. Equivalently, the maximum number of points in each layer is less than or equal to the number of sign changes of the excess skill supply function plus one, since the first point in a layer does not require a change of direction.

The main result of this section (\Cref{t:numberpair}) shows a significant reduction of complexity in the primal problem when the distributions of workers and jobs are normal mixtures distributions. This result enables our quantitative analysis with a large number of worker and job types. Moreover, normal mixture distributions are weakly dense in the set of all distributions: for any probability distribution $F$ on the real line, there exist distributions $\{F^{(n)}\}_{n\in\N}$ that are normal mixtures converging weakly to $F$ as $n \to \infty$.
 
\begin{theorem}{\textit{Characterization with Normal Mixture Distributions}.} \label{t:numberpair}
Let the worker distribution be a mixture of $n$ normal distributions and the job distribution be a mixture of $m$ normal distributions. Then each layer in the measure of underqualification consists of at most $n+m-1$ pairs. The same conclusion holds for the discretized distributions $\hat{F}$ and $\hat{G}$ for any collection of ordered skill levels $\{ s_i \}_{i=1}^{N_s}$.
\end{theorem}


\noindent We prove \Cref{t:numberpair} by establishing that the excess skill supply function $h$ changes sign at most $2(n+m-1)$ times. In turn, this implies there are at most $2(n+m-1) + 1$ points on each layer. Since the number of points in each layer is even, it follows that there are at most $n + m -1$ pairs in each layer. While we present \Cref{t:numberpair} for continuous distribution functions for clarity of exposition, the result holds for both the continuous and for the discretized distributions.

The proof is presented in Appendix \ref{p:nm} and builds on the result that for a density of a signed measure, its convolution with any normal density does not increase its number of sign changes (\Cref{lemma:total positivity2}). The proof uses Schoenberg’s theory of variation diminishing transformations and P{\'o}lya frequency functions \citep{Schoenberg:1930,Schoenberg:1950}.\footnote{The literature on total positivity building on this work \citep{Karlin:1968} has been used in information economics. For example, see \citet{Jewitt:1987}, \citet{Athey:2002}, \citet{Choi:2017}, \citet{Wilson:2019}, \citet{Chade:2020} and \citet{Chade:2024}.}  The most involved part of the proof is \Cref{lemma:total positivity2}, which shows that the variation diminishing property holds for convolutions of signed measures containing point masses with Gaussian noise. A different but related technique is used by \citet{Pomatto:2020} whose proof to their Theorem 1 uses that the convolution with a specific normal density reduces the number of zeros of the density of a signed measure in the context of stochastic dominance with independent noise.  Similar technical tools were also recently used by \citet{Sandomirskiy:2023,Sandomirskiy:2024} for establishing the origins of the multinomial logit stochastic choice rule.

An important implication of \Cref{t:numberpair} is a significant reduction in complexity of solving the sorting problem. Specifically, we compare the efficiency of solving a discretized analog of the normal mixture model to the efficiency of solving this model using the layering structure. This complexity bound makes feasible quantitative analysis in \Cref{sec:quantitative} with a large number of worker and job types.

\setcounter{corollary}{0}

\begin{corollary}\label{corol}
 Suppose that worker and job distributions satisfy the conditions in Theorem \ref{t:numberpair}.   Let $\hat{F},\hat{G}$ be discretizations on ordered skill levels $\{ s_i \}_{i=1}^{N_s}$. 
Then the time complexity of computing the optimal sorting is $O((m+n)^3 N_s)$. 
\end{corollary}

\noindent By Theorem \ref{t:numberpair}, the maximum number of pairs per layer is at most $m+n-1$ for the discretized measure of underqualification $\hat{H}$. Using the Hungarian algorithm (see, for example, \citet{Burkard:2012}), the complexity of solving the assignment problem in a single layer is $O((m+n)^3)$. Since the number of layers is bounded above by the number $N_s$ of skill levels, the complexity after decomposing into layers is $O((m+n)^3 N_s)$. In contrast, directly solving the assignment problem with $N_s$ skill types has complexity $O(N_s^3)$ and is not feasible even for a relatively modest numbers of job and worker types. Note that both $m$ and $n$ are typically quite small for most univariate datasets \citep{McLachlan:2000}. 

\vspace{0.4 cm}
\noindent \textbf{Discussion}. Due to Theorem \ref{t:numberpair}, optimal sorting can be computed for a large number of worker and job types using, for example, a standard linear program for each layer. In Technical Appendix \ref{s:withinlayer}, we describe the details of characterizing the optimal assignment within each layer. \citet{Sargent:2025} present a concise description of the recursive algorithm that we use to compute the optimal assignment.

Finally, while we analyzed the setting where the worker and job distributions are given by normal mixtures, the same results hold when the worker and job distributions follow lognormal mixture distributions. This follows directly from an exponential transformation of the skills, which preserves the number of points per layer. This observation is used in our quantitative analysis in \Cref{sec:quantitative}, where the distributions of workers and jobs are described by discretized lognormal mixture distributions.


\section{Wages and Firm Values} \label{s:wagesandfirmvalues}

This section derives a full characterization of the dual solution to the sorting problem with concave costs of skill gaps (\ref{e:pp_dual}) and determines equilibrium wages and firm values. The main result of our dual construction is that equilibrium wages and job values form a \textit{regional hierarchical structure}, fundamentally different from classical sorting models. The main technical challenge that we need to overcome in our construction is to ensure consistency at every scale from regional to global scales via aggregation of the regional relative wages. 


The first part of this section characterizes the dual solution for mismatched workers and jobs. We construct the mismatch penalty and show that it has a regional hierarchical structure where wages are determined independently within each region of skills without having to consider other regions. We then use these mismatch penalties in the second part of the section to construct wages and firm values for all workers and jobs.\footnote{Another literature, which follows \citet{Garicano:2000} and \citet{Garicano:2006}, solves hierarchical
assignment models with heterogeneous workers. Knowledge is cumulative, so that more skilled workers know how to solve a problem when less skilled workers
do. Production is supermodular in worker skill, and equilibrium sorting $-$ the primal solution $-$ is positive and hierarchical. We establish a hierarchical structure in equilibrium wages and firm values $-$ the dual solution in \Cref{thm:dualalg}.}


\vspace{0.35 cm}
\noindent \textbf{Mismatched Workers and Jobs}. Let $S = I \cup J$ denote the set of all skill levels, where $I$ and $J$ are disjoint sets of worker skills and job complexities after the removal of perfect pairs. Suppose that an optimal assignment $\pi$ consists of $n$ worker-job pairs $\{(x_i,z_i)\}^n_{i=1}$. Our first goal is to construct mismatch penalties, or a shadow cost of mismatch, $\phi: S \to \R$ such that for each worker $x$ and every occupation $z$,
\begin{equation}
\phi(x)-\phi(z)\leq c(x,z), 		\label{eq:dual}
\end{equation}
which holds with equality if the assignment $\pi$ pairs worker $x$ to job $z$.\footnote{Setting $\psi(z)=-\phi(z)$ for all mismatched jobs $z\in Z$, we equivalently construct the worker mismatch penalty $\phi:I\to\R$ and
the firm mismatch penalty $\psi:J\to\R$ such that for each worker $x$
and job $z$, $\phi(x)+\psi(z)\leq c(x,z)$, where the equality holds
if the optimal assignment $\pi$ pairs worker type $x$ with occupation $z$.}

 

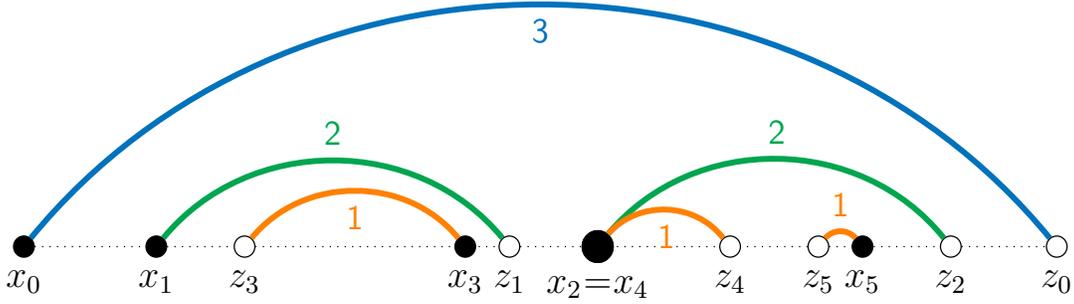
\begin{figure}[!t]
     \begin{centering}
    \resizebox{14.5cm}{5.15cm}{\begin{tikzpicture}
\node[circle,fill=black,draw, minimum size=0.1pt,scale=0.6,label=below:{$x_0$}] (A) at  (1.5,0){};
\node[circle,fill=black,draw, minimum size=0.1cm,scale=0.6,label=below:{$x_1$}] (B) at  (3,0) {} ;
\node[circle,draw, minimum size=0.1pt,scale=0.6,label=below:{$z_3$}] (C) at  (4,0){};
\node[circle,fill=black,draw, minimum size=0.1cm,scale=0.6,label=below:{$x_3$}] (B3) at  (6.5,0) {} ;
\node[circle,draw, minimum size=0.1pt,scale=0.6,label=below:{$z_1$}] (C3) at  (7,0){};

\node[circle,draw,fill=black, minimum size=0.1pt,scale=0.9,label=below:{}] (E) at  (8,-0){};
\node[circle,fill=black,draw, minimum size=0.1cm,scale=0.9,label=below:{ $x_2\hspace{-0.1cm}=\hspace{-0.1cm}x_4$}] (F) at  (8,0) {} ;
\node[circle,draw, minimum size=0.1pt,scale=0.6,label=below:{$z_4$}] (G) at  (9.5,0){};
\node[circle,draw, minimum size=0.1cm,scale=0.6,label=below:{$z_5$}] (H) at  (10.5,0) {} ;
\node[circle,fill=black,draw, minimum size=0.1cm,scale=0.6,label=below:{$x_5$}] (H2) at  (11,0) {} ;
\node[circle,draw, minimum size=0.1pt,scale=0.6,label=below:{$z_2$}] (I) at  (12,0){};

\node[circle,draw, minimum size=0.1cm,scale=0.6,label=below:{$z_0$}] (L) at  (13.2,0) {};

\draw[dotted] (A) -- (B);\draw[dotted] (B) -- (C);

\draw[dotted] (C) -- (C3); 


\draw[dotted] (H) -- (G);\draw[dotted] (C3) -- (G);
\draw[dotted] (H) -- (I);\draw[dotted] (I) -- (L);
\path[line width=0.65mm, RoyalBlue,-,every node/.style={font=\sffamily\small}] (L) edge[bend right=50] node [below] {3} (A);
\path[line width=0.65mm, Green,-,every node/.style={font=\sffamily\small}] (C3) edge[bend right=50] node [above] {2} (B);
\path[line width=0.65mm, orange,-,every node/.style={font=\sffamily\small}] (C) edge[bend left=50] node [below] {1} (B3);
\path[line width=0.65mm, Green,-,every node/.style={font=\sffamily\small}] (I) edge[bend right=50] node [above] {2} (E);
\path[line width=0.65mm, orange,-,every node/.style={font=\sffamily\small}] (G) edge[bend right=50] node [below] {1} (F);

\path[line width=0.65mm, orange,-,every node/.style={font=\sffamily\small}] (H2) edge[bend right=50] node [above] {1} (H);

        \end{tikzpicture}}
\par\end{centering}
\vspace{-0.25cm}

    \caption{Regional Hierarchical Structure}
\label{f:localhierarchical} {\scriptsize{}{}\vspace{0.2cm} Figure \ref{f:localhierarchical} illustrates the regional hierarchical structure for the dual solution given the optimal sorting given by the arcs. A lower skill region has skills in the interval between $x_{1}$ and $z_{1}$ while a higher skill region has skills in the interval between $x_{2}$ and $z_{2}$. The hierarchical structure implies that the relative shadow cost of mismatch for all skills within either the first or the second region is determined within the region. Wages are determined hierarchically within regions. In the lower skill region, wages are first determined for the innermost pair $(x_{3},z_{3})$, and then for the outer pair $(x_{1},z_{1})$. The numbers on the arcs indicate the sequence by which we move from low-level pairs to high-level pairs.}
    
\end{figure}

Before providing a formal description of the regional hierarchical structure, we illustrate the construction for the dual solution given the optimal sorting captured by the arcs in Figure \ref{f:localhierarchical}. A lower skill group has skills in the interval between $x_{1}$
and $z_{1}$. A higher skill group has skills in the interval between
$x_{2}$ and $z_{2}$. The relative shadow cost of mismatch for all
skills within either the first or the second group is determined within each group. Wage determination is thus \textit{regional}, meaning that wages are determined within a group independent of other groups. Wage determination is also \textit{hierarchical} within groups: at each stage, wages depend only on information from the skill group
nested within the progressively larger group. In
the low-skill group, the wage is first determined for the innermost
pair $(x_{3},z_{3})$ that contains no nested skill groups, and then for the outer
pair $(x_{1},z_{1})$. The hierarchical structure shows that the construction of wages moves sequentially from low-level pairs to high-level pairs, as indicated by the numbers on the arcs in Figure \ref{f:localhierarchical}.
For the high-skill
group, the relative wage for the pair $(x_{2},z_{2})$ is thus constructed from
the relative wages for both pairs $(x_{4},z_{4})$ and $(x_{5},z_{5})$. Finally, wages for the
outermost pair $(x_{0},z_{0})$ are constructed using the relative
wages for the first and second skill groups. 

We now describe our hierarchical characterization of the mismatch penalties for composite sorting. The hierarchical structure means that penalties within each region, by which we mean the interval between two points in a pair, are constructed from the mismatch penalties in its subregions and do not depend on other regions. Formally, a region is determined by a pairing $(x_0,z_0)$. The subregions of $(x_0,z_0)$ are given by the subpairs $\{(x_i,z_i)\}^p_{i=1}$, which are non-nested pairs inside the skill interval $(x_0,z_0)$. The hierarchical structure leads to the following separation property of the global mismatch penalty $\phi$: for any two points $s$ and $s'$ in a region,
$\phi(s)-\phi(s')$ can be computed only based on points within this region, and thus it is independent of points outside this region. We describe this structure in detail below.  

The hierarchical construction in each region is as follows. Consider a region $(x_0,z_0)$ that contains $p$ subpairs $\{(x_i,z_i)\}^p_{i=1}$. Each subpair comes with an associated mismatch penalty function $\phi_i$, which satisfies equation (\ref{eq:dual}) for all workers and jobs inside the interval $(x_i,z_i)$. We note that level shifts in the mismatch penalties do not affect the relative mismatch penalty within the region. We can then construct level shifts to obtain a mismatch penalty for the entire region. 

Specifically, given a mismatch penalty function for each subregion, we construct a single mismatch function for the region. First, we ensure that mismatch penalty functions for the subpairs are consistent with each other. We do this by shifting each of the penalty functions in level. Specifically, we construct a mismatch penalty for the region on the domains of the subpairs:
\begin{align}
\phi(s)= \phi_i(s)+ \sum_{k=i+1}^p \beta _k +\phi_p(x_p)-\phi_i(x_i) \label{eq:phidef}
\end{align}
for $i = 1,\dots,p$ and for skills $s\in X_{[x_i,z_i]}$ and $s \in Z_{[x_i,z_i]}$. In equation (\ref{eq:phidef}), level shifts are captured by $\sum \beta _k$, with the convention that the sum equals zero when $i=p$. Without loss of generality, we normalize the level shifts so that the mismatch penalty in the $p$-th subregion is unchanged. 

In addition to specifying the mismatch penalty for all workers and jobs inside the region, we specify the mismatch penalty on the boundaries at $x_0$ and $z_0$. For occupation $z_0$, equation (\ref{eq:dual}) requires $\phi(z_0) \geq \phi(x) - c(x,z_0)$ for all workers in the region. Therefore, we set
\begin{equation}
\phi(z_0) = \max\limits_x \hspace{0.08 cm} \phi(x) - c(x,z_0)  \hspace{0.07 cm} . \label{e:phiz0}
\end{equation}
Moreover, in order to ensure that equation (\ref{eq:dual}) holds with equality for the pair $(x_0,z_0)$, we specify the mismatch penalty for worker $x_0$ as: 
\begin{equation}
\phi(x_0) = \phi(z_0) + c(x_0,z_0) \hspace{0.07 cm}  . \label{e:phix0}
\end{equation}
Given these definitions, we need to ensure for worker $x_0$ that $\phi(x_0) - \phi(z) \leq c(x_0,z)$ for all occupations $z$ in the region. Using the specification of the mismatch penalty at the boundaries \eqref{e:phiz0} and \eqref{e:phix0}, for all workers $x$ and occupations $z$ in the region it has to be that:
\begin{equation}
\phi(x) - \phi(z) \leq c(x_0,z) + c(x,z_0) - c(x_0,z_0)  \hspace{0.07 cm}. \label{e:phix1}
\end{equation}


In order to obtain a valid mismatch penalty of the form (\ref{eq:phidef}) for the region, the level shifts have to satisfy conditions such that (\ref{eq:dual}) and (\ref{e:phix1}) hold. Consider a worker type $x_m$ and an occupation $z_n$, where $n$ and $m>n$ index subpairs. In order to satisfy condition (\ref{eq:dual}), we require $\phi(x_m) - \phi(z_n) \leq c(x_m,z_n)$ and to satisfy condition (\ref{e:phix1}) we require $\phi(x_m) - \phi(z_n) \leq c(x_0,z_n) + c(x_m,z_0) - c(x_0,z_0)$. Using the proposed mismatch penalty (\ref{eq:phidef}), this requires $\sum\limits_{k=n+1}^m \beta _k \geq \max(c_{00}-c_{0n}-c_{m0} + c_{nn}, c_{nn}-c_{mn})$, where $c_{ij}=c(x_i,z_j)$ for $i,j\in\{1,\dots,p\}$. Consider next a worker type $x_n$ and an occupation $z_m$. Similarly, it is necessary that both $\phi(x_n) - \phi(z_m) \leq c(x_n,z_m)$ and $\phi(x_n) - \phi(z_m) \leq c(x_0,z_m) + c(x_n,z_0) - c(x_0,z_0)$, which using the proposed mismatch penalty (\ref{eq:phidef}) requires $\sum\limits_{k=n+1}^m\beta_k\leq \min(c_{0m}+c_{n0}-c_{00} -c_{mm}  ,c_{nm} -c_{mm} )$. In sum, the level shifters of the mismatch penalty $(\beta_2,\dots,\beta_p)$ are necessarily a solution to the system of inequalities: 
\begin{align}  
\max(c_{00}-c_{0n}-c_{m0} + c_{nn}, c_{nn}-c_{mn})  \leq \sum_{k=n+1}^m\beta_k\leq \min(c_{0m}+c_{n0} -c_{00} -c_{mm}  ,c_{nm} -c_{mm} )  \label{eq:ineqs}
\end{align} 
for all $1\le n<m\le p$. 

It is important to note that at each stage, we exploit the concavity of the cost function to establish the existence of a solution to the system of inequalities (\ref{eq:ineqs}) using Lemma \ref{lemma:dualalgorithm}. Thus, our findings on the hierarchical structure for the dual problem are specific to concave costs of skill gaps. Regional groups generally do not arise for convex costs. For instance, when all worker skills are below each job complexity, the optimal assignment for a convex cost is positive sorting, which allows for no local regions (except for the whole set) since any two pairs intersect. In sharp contrast, for a concave cost, every pair forms a local region independent of how workers and jobs are located.

The idea of the hierarchical mechanism is to start from the pairs with no subpairs, and then pairs with all subpairs analyzed in previous steps until all pairs are exhausted.\footnote{When there are no further subpairs, the final step is to choose level shifts to ensure consistency among the existing pairs in the same manner we discussed above.} We present a condensed statement of the theorem here and include the complete formal description of the regional hierarchical mechanism in Appendix \ref{a:dualproof}, and the proof in Appendix \ref{pthm:dualalg}.

\begin{theorem}{\textit{Regional Hierarchical Mechanism for Mismatch Penalties}.} \label{thm:dualalg} Given an optimal assignment,
the regional hierarchical mechanism constructs an optimal dual pair $(\phi,\psi)$
where $\psi=-\phi$. 

Within each skill group, relative wages are
determined regionally: for any two points $s$ and $s'$ in a skill group,
$\phi(s)-\phi(s')$ depends only on the pairs within the group.
\end{theorem}

\noindent Our mechanism solves the problem in the order from bottom to top. At each step, the values of $\phi$ in the hidden arcs need not be computed again but only adjusted with constant level shifts. 

Theorem \ref{prop:hiereff} establishes that our hierarchical solution results in a sharp reduction of complexity when the distributions of workers and jobs are given by discretization of mixtures of normal distributions. 

\begin{theorem}{\textit{Complexity of the Regional Hierarchical Mechanism for Mismatch Penalties}.} \label{prop:hiereff}
    Let the worker distribution $F$ be a discretization of a mixture of $n$ normal distributions and let the job distribution $G$ be a discretization of a mixture of $m$ normal distributions with ordered skill levels $\{ s_i \}_{i=1}^{N_s}$. 
Then the complexity of the dual potentials is $O((m+n)^4 N_s)$.
\end{theorem}

\noindent The proof of \Cref{prop:hiereff} is in Appendix \ref{a:hiereff}. We first note that directly solving the dual problem as a linear program with $N_s$ skill types has time complexity $O(N_s^4)$, which is not feasible even for a relatively modest numbers of job and worker types. Instead, Theorem \ref{prop:hiereff} establishes how we construct the mismatch penalty for a large number of worker and job types by drastically reducing the complexity of the problem. We use this result to compute the wages and firm values in the quantitative analysis in Section \ref{sec:quantitative}. In Proposition \ref{rem:runtime} of Technical Appendix \ref{s:efficiency} we show further efficiency properties of our hierarchical construction for empirical measures. 

We next use the mismatch penalty functions to define worker earnings and firm values for the non-overlapping segments of the worker distribution and
the job distribution. Let wages $w(x)=\alpha(x)-\phi(x)$
and firm values $v(z)=\theta(z)-\psi(z)$, where we recall from the technology (\ref{e:firmtechnology}) that $\alpha$ reflects the worker contribution to production independent of the occupation, and $\theta$ reflects the value of the job independent of the worker that fulfills the job. The first observation is
that assignment $\pi$ that solves the mismatch cost minimization
problem (\ref{ppmin}) also solves the output maximization
problem (\ref{pp}). Moreover, $w(x)+v(z)\geq y(x,z)$ holds for all
$(x,z)$ with equality if worker $x$ is assigned to job $z$ under the optimal assignment, where $y(x,z)=\alpha(x)+\theta(z)-c(x,z)$.
By Lemma \ref{lemma:dual}, it thus follows that $(w,v)$ is a dual optimizer
for the output maximization problem. In sum, given the mismatch
penalty $(\phi,\psi)$ for the minimization problem without overlapping
parts, the dual pair $(w,v)$ for the maximization problem without
overlapping parts is obtained.

\vspace{0.35cm}
\noindent \textbf{Adding Perfectly Paired Workers and Jobs}. Up to this point, we determined worker wages and firm values in the output maximization problem when there is no overlap between the distributions
of workers and jobs. These wage and value functions are used to construct worker wages and firm values where there is overlap in the distributions.



We start with only mismatched workers and jobs and denote the wages constructed above by $\tilde{w}$.\footnote{Recall that the sets of mismatched workers and jobs are $I$ and $J$ and that $S = I \cup J$ is the set of all skill levels.} We add perfectly matched firms and determine what income each firm could generate given mismatched workers and wages $\tilde{w}$. The first auxiliary firm problem is to choose an employee among only mismatched workers $x\in I$. Formally,
a firm with job $z\in S$ solves: $\tilde{v}(z):=\max\limits _{x\in I} \big( y(x,z)-\tilde{w}(x) \big)$. We refer to $\tilde{v}$ as firm mismatch compensation, that is, profits firms can attain given a mismatched worker with wage $\tilde{w}$.

We next introduce perfectly paired workers and present both mismatched and perfectly paired workers with firm mismatch compensation $\tilde{v}$. We determine what wage income both the imperfectly and perfectly paired workers would generate given the compensation required by firms. The auxiliary decision problem of a worker $x\in S$ is to choose any job, including the perfectly paired jobs, to solve: 
\begin{align}
\hat{w}(x):=\max_{z\in S} \; \hspace{0.05 cm} y(x,z)-\tilde{v}(z) .\label{e:hatphi}
\end{align} 
As a result, we obtain wages $\hat{w}$ for both mismatched and perfectly paired workers.



Finally, we determine what profits $\hat{v}$ firms would generate given all workers and their required compensation $\hat{w}$. 
We set up a second auxiliary firm problem, which is the problem
of a mismatched job $z\in J$ choosing an employee among all
workers (perfectly paired and mismatched) subject to wage
schedule $\hat{w}$: 
\begin{align}
\hat{v}(z):=\max_{x\in S}\; \hspace{0.05 cm} y(x,z)-\hat{w}(x) .\label{e:psiz2}
\end{align}
We refer to $\hat{v}$ as mismatched firm compensation, since it represents
the profits of firm type $z\in J$.

Equilibrium wages are formulated using auxiliary wages for mismatched
workers, $w(x)=\hat{w}(x)$ for all $x\in I$, as well as mismatched
firm compensation, $v(z)=\hat{v}(z)$ for all $z\in J$. Equilibrium wages equal $w(x)=\alpha(x)+\theta(x)-v(x)$
for all $x\in J$ and equilibrium firm values $v$ are $v(z)=\alpha(z)+\theta(z)-w(z)$ for every job $z\in I$. Theorem \ref{p:dual}
shows that the wage function $w$ and the firm value function $v$
indeed solve the dual problem for the full assignment problem.\footnote{We can extend the domain of wages $w$ and values $v$ to $K=X\setminus(I\cup J)$
by setting $w(x)=\max\limits_{z\in I\cup J} \big( y(x,z)-v(z) \big)$ for $x\in K$ and
$v(z)=\alpha(z)+\theta(z)-w(z)$ for $z\in K$.}

\begin{theorem}{\textit{Dual Solution}.} \label{p:dual} The constructed functions $(w,v)$ are a dual solution for the sorting problem, that is, $w(x)+v(z)\geq y(x,z)$, which holds with equality if the assignment $\pi$ pairs worker $x$ to job $z$. \end{theorem}

\noindent The proof, as well as a formal analysis of the above mechanism, is in Technical Appendix \ref{s:pdual}.


\section{Comparative Statics} \label{s:comparativestatics}

In this section, we prove two results on how sorting varies with changes in the output function. We show that the analysis of comparative statics with concave costs of skill gaps is significantly richer and more complicated than for the canonical assignment models. In constrast, with convex costs optimal sorting is always assortative and thus does not vary with the extent of convexity. In this section we develop a new approach to comparative statics leveraging the characterization of optimal sorting with concave costs of skill gaps.

Theorem \ref{p:cs} shows that optimal sorting becomes more positive, by which we mean larger in concordance order, as the cost function becomes less concave in the skill gap. Theorem \ref{t:tolinear} shows that there exists a threshold in concavity of costs $\zeta_p$ and $\zeta_u$ beyond which the optimal assignment in each layer is positive, which we call \textit{layered positive sorting}.


\vspace{0.35cm}
\noindent \textbf{Positive Sorting}. For any two assignments $\pi$ and $\hat{\pi}$ between a fixed pair of distributions of workers and jobs, we say that assignment $\pi$ is smaller in \textit{concordance order} than $\hat{\pi}$, which we denote by $\pi \preceq \hat{\pi}$, if for any coordinate $(x,z)$, less mass is concentrated in both the top-right and bottom-left quadrants under assignment $\pi$ than under $\hat{\pi}$. Intuitively, a more positive sorting corresponds to an assignment larger in concordance order, and this equivalence was made precise by \citet{Tchen:1980}.\footnote{When assignment $\hat{\pi}$ is larger in concordance order, other measures of statistical association, such as  the rank correlation, the correlation coefficient, and Kendall's tau coefficient are also larger for $\hat{\pi}$ than for assignment $\pi$ \citep{Joe:1997}.}

A possible approach to analyze comparative statics would be to apply recent results of \citet{Anderson:2022} who provide sufficient conditions under which sorting is larger in concordance order as the output function changes. In Appendix \ref{a:cs}, we show, however, that their conditions are not satisfied in our economy with concave costs of skill gaps, which leads us to develop a different approach to comparative statics that relies on the characterization of optimal sorting with concave costs of skill gaps. The key difficulty is to identify a local cyclical monotonicity condition specifically for concave costs.



Theorem \ref{p:cs} shows that optimal sorting is more positive when costs of skill gaps becomes less concave. 

\begin{theorem}{\textit{Comparative Statics of Composite Sorting}.} \label{p:cs}
Suppose that the cost of skill gaps $c(x,z)$ is an increasing concave function of skill gaps and that $\kappa$ is some increasing convex function so that $\hat{c}= \kappa(c)$ is also an increasing concave function of skill gaps.\footnote{That is, the cost function $c(x,z)$ is an increasing and concave function of skill gap $(x-z)$ on the region $x \geq z$ and a (potentially different) increasing and concave function of skill gap  $(z-x)$ on the region $x\leq z$.} If $\pi$ is an optimal assignment with costs $c$, then there exists an optimal assignment $\hat{\pi}$ with the less concave cost of skill gaps $\hat{c}$ such that $\pi \preceq \hat{\pi}$.
\end{theorem}

\begin{figure}[!t]
\begin{centering}
\resizebox{16.2cm}{7.50cm}{ \begin{tikzpicture}
\node[circle,draw, minimum size=0.1pt,scale=0.6,label=below:{$\tilde{x}_{1}$}] (A) at  (0,0){};
\node[circle,fill=black,draw, minimum size=0.1cm,scale=0.6,label=below:{$\vphantom{\tilde{x}} z_1$}] (B) at  (1,0) {} ;
\node[circle,draw, minimum size=0.1pt,scale=0.6,label=below:{$ \vphantom{\tilde{x}} x_1$}] (C) at  (2,0){};
\node (2C) at (2.8,0){};
\node (c1) at (3.5,0) {\Large$\dots$};
\node (2D) at (4.2,0){};
\node[circle,fill=black,draw, minimum size=0.1cm,scale=0.6] (D) at  (5,0) {} ;
\node[circle,draw, minimum size=0.1pt,scale=0.6] (E) at  (6,0){};
\node[circle,fill=black,draw, minimum size=0.1cm,scale=0.6,label=below:{$\tilde{z}_{1}$}] (7) at  (7,0) {} ;

\node (E2) at (7.8,0){};
\node (c) at (8.5,0) {\Large$\dots$};
\node (F2) at (9.2,0){};

\node[circle,draw, minimum size=0.1pt,scale=0.6,label=below:{$\tilde{x}_{q}$}] (A2) at  (10,0){};
\node[circle,fill=black,draw, minimum size=0.1cm,scale=0.6] (B2) at  (11,0) {} ;
\node[circle,draw, minimum size=0.1pt,scale=0.6] (C2) at  (12,0){};
\node (2C2) at (12.8,0){};
\node (c2) at (13.5,0) {\Large$\dots$};
\node (2D2) at (14.2,0){};
\node[circle,fill=black,draw, minimum size=0.1cm,scale=0.6,label=below:{$\vphantom{\tilde{x}} z_p$}] (D2) at  (15,0) {} ;
\node[circle,draw, minimum size=0.1pt,scale=0.6,label=below:{$ \vphantom{\tilde{x}} x_p$}] (E3) at  (16,0){};
\node[circle,fill=black,draw, minimum size=0.1cm,scale=0.6,label=below:{$\tilde{z}_{q}$}] (72) at  (17,0) {} ;

\draw[dotted] (A) -- (B);
\draw[dotted] (2C) -- (C);\draw[dotted] (D) -- (2D);
\draw[dotted] (B) -- (C);
\draw[dotted] (E2) -- (E);\draw[dotted] (2D) -- (E);
\draw[dotted] (A2) -- (F2);\draw[dotted] (A2) -- (C2);
\draw[dotted] (C2) -- (2C2);\draw[dotted] (E3) -- (2D2);
\draw[dotted] (E3) -- (72);

\path[-,every node/.style={font=\sffamily\small}] (A) edge[bend left=40] node [left] {} (7);
\path[-,every node/.style={font=\sffamily\small}] (C) edge[bend right=40] node [left] {} (B);
\path[-,every node/.style={font=\sffamily\small}] (D) edge[bend left=40] node [left] {} (E);
\path[-,every node/.style={font=\sffamily\small}] (D2) edge[bend left=40] node [left] {} (E3);
\path[-,every node/.style={font=\sffamily\small}] (A2) edge[bend left=40] node [left] {} (72);
\path[-,every node/.style={font=\sffamily\small}] (C2) edge[bend right=40] node [left] {} (B2);

\node[circle,draw, minimum size=0.1pt,scale=0.6,label=below:{$x_{0}$}] (A) at  (0,4){};
\node[circle,fill=black,draw, minimum size=0.1cm,scale=0.6,label=below:{$z_{1}$}] (B) at  (1,4) {} ;
\node[circle,draw, minimum size=0.1pt,scale=0.6,label=below:{$x_1$}] (C) at  (2,4){};
\node[circle,fill=black,draw, minimum size=0.1cm,scale=0.6,label=below:{$z_{2}$}] (D) at  (3,4) {} ;
\node[circle,draw, minimum size=0.1pt,scale=0.6,label=below:{$x_2$}] (E) at  (4,4){};
\node (E2) at (8,4){};
\node (F2) at (11,4){};
\node[circle,fill=black,draw, minimum size=0.1cm,scale=0.6,label=below:{$z_p$}] (F) at  (15,4) {} ;
\node[circle,draw, minimum size=0.1pt,scale=0.6,label=below:{$x_p$}] (G) at  (16,4){};
\node[circle,fill=black,draw, minimum size=0.1cm,scale=0.6,label=below:{$z_0$}] (H) at  (17,4) {} ;
\node (c) at (9.5,4) {\Large$\dots$};
\draw[dotted] (A) -- (B);

\draw[dotted] (B) -- (C);\draw[dotted] (D) -- (C);
\draw[dotted] (E2) -- (E);\draw[dotted] (D) -- (E);
\draw[dotted] (F) -- (F2);
\draw[dotted] (H) -- (G);\draw[dotted] (F) -- (G);
\path[-,every node/.style={font=\sffamily\small}] (A) edge[bend left=30] node [left] {} (H);
\path[-,every node/.style={font=\sffamily\small}] (C) edge[bend right=40] node [left] {} (B);
\path[-,every node/.style={font=\sffamily\small}] (D) edge[bend left=40] node [left] {} (E);
\path[-,every node/.style={font=\sffamily\small}] (F) edge[bend left=40] node [left] {} (G);

 \draw[line width=1.05pt, double distance=0pt, -{Classical TikZ Rightarrow[length=3mm]}] (8.5,3.2) -- (8.5,1.6);

        \end{tikzpicture}} 
\par\end{centering}
\caption{Sorting When Costs of Skill Gaps Become Less Concave}
\label{f:cc'} {\scriptsize{}{}\vspace{0.2cm}
Figure \ref{f:cc'} shows an optimal assignment that is no longer optimal when the costs of skill gaps become less concave. In this case, there exists a pair $(x_0,z_0)$ with positively sorted subpairs $\{(x_i,z_i)\}_{i=1}^p$, which we display in the top panel, such that the assignment with workers and jobs $\{(x_i,z_i)\}_{i=0}^p$ can be improved in a more positive fashion, as shown in the bottom panel.} 
\end{figure}


\noindent We prove Theorem \ref{p:cs} in \Cref{a:cs}. Since the distributions of workers and jobs remain unchanged across different costs, the sorting problem is split into identical layers irrespective of the cost function. Hence, \Cref{p:cs} means that optimal sorting within each layer becomes more positive as the cost of skill gaps becomes less concave. We now outline the two main steps of the proof.

First, for our setting with concave costs, we establish a new characterization of the classical cyclical monotonicity specific to concave costs of skill gaps. If an optimal assignment is no longer optimal when the costs are less concave with the mismatch costs $\hat{c}$, there exists a pair $(x_0,z_0)$ with positively sorted subpairs $\{(x_i,z_i)\}_{i=1}^p$, as in the top panel of Figure \ref{f:cc'} (by Lemma \ref{lemma:cyclical}), such that the local assignment problem with workers and jobs $\{(x_i,z_i)\}_{i=0}^p$ can be improved with more positive sorting shown in the bottom panel of Figure \ref{f:cc'}. In order to prove this, suppose that in an optimal assignment with less concave costs, the worker $x_0 = \tilde{x}_1$ is instead optimally paired to job $z_k = \tilde{z}_1$ for some $k$.  Second, we show in Lemma \ref{l:positive} that since positive sorting is optimal on the interval $(x_0,z_k)$ with a more concave cost of skill gaps, positive sorting is also optimal on $(x_0,z_k)$ with a less concave cost of skill gaps. We continue this procedure to the right, that is, we start with worker $x_{k+1} = \tilde{x}_2$ and repeat the argument, and obtain the structure in the bottom panel of Figure \ref{f:cc'}.


In order to see that the optimal assignment becomes more positive, or larger in concordance order, we make two observations. First, note that all successively positively sorted pairs in the top panel of Figure \ref{f:cc'}, such as $(x_1,z_1)$, $(x_2,z_2)$ and $(x_p,z_p)$, are also formed in the bottom panel of Figure \ref{f:cc'}. Hence, they do not affect the concordance order. Second, we observe that the bottom panel sorts the remaining workers and jobs positively, which has the largest concordance order among all assignments. Since the top panel does not sort the remaining workers and jobs positively, it follows that the assignment for the bottom panel is larger in concordance order. Hence, all improvements make the assignment more positive.

\vspace{0.35cm}
\noindent \textbf{Threshold for Layered Positive Assignment}.  The previous result shows that a more concave cost function yields more negative sorting, and a less concave cost function yields more positive sorting. We next derive a threshold for concavity beyond which optimal sorting is the most positive assignment for our economy. The most positive assignment for our economy is given by positive sorting in each layer, which we call \textit{layered positive sorting}. It is important to note that positive sorting in each layer does not imply positive sorting overall.

We consider the assignment problem when the power indices $\zeta_p$ and $\zeta_u$ for the costs of skill gaps (\ref{eq:cxz}) are close to one, that is, when the cost of skill gaps is almost linear in the distance between the worker skill and the job.\footnote{We remark that this is the only result that uses the specific form for the costs of skill gaps (\ref{eq:cxz}). All other results only use concavity of the costs of skill gaps.} First, we maximize the number of perfect pairs. Second, we decompose the assignment problem into layers $0\leq\ell\leq L$. Third, when $\zeta_p$ and $\zeta_u$ exceed the threshold $\bar{\zeta}$, the optimal assignment within each layer is simple. Specifically, we show in Theorem \ref{t:tolinear} below that the optimal sorting within each layer is positive sorting, which we denote by $\pi_{\ell}^{+}$. The solution to the full assignment problem is given by the combination of the positive assignments within each layer. We refer to this assignment as the layered positive assignment denoted by $\pi^{+}=\sum\pi_{\ell}^{+}$.

\begin{theorem}{\textit{Layered Positive Sorting}.} \label{t:tolinear} For any discrete worker and job distributions, there exists $\bar{\zeta}<1$ such
that for any $\zeta_p, \zeta_u \in[\bar{\zeta},1]$, the layered positive assignment
$\pi^{+}$ is optimal. \end{theorem}


\noindent The proof is in Appendix \ref{App:linearcost}. The implication is that for mismatch power indices above the threshold $\bar{\zeta}$, the solution can be directly constructed by evaluating the measure of underqualification, and by assigning positively within each layer.\footnote{It is useful to contrast our result with \cite{Juillet:2020} who shows that the layered positive assignment is the limit of some optimal assignments as $\zeta\to1^{-}$. Our result proves the existence of a threshold beyond which the layered positive assignment is optimal for our environment and is applicable away from the limit.}

\section{Quantitative Results} \label{sec:quantitative}


In this section, we develop a quantitative illustration of the model. The distinctive feature of our model among assignment models is that we obtain earnings dispersion within occupations and, hence, we apply our model to evaluate earnings dispersion within and across occupations. The quantitative model isolates the implications of concave costs of skill gaps for sorting and earnings dispersion within occupations. We compare our results to settings with supermodular and submodular output functions, which have commonly been used in the assignment literature.


\begin{figure}[!t]
\begin{center}
\begin{subfigure}{0.47\textwidth}
\includegraphics[width=\textwidth,height=0.25\textheight]{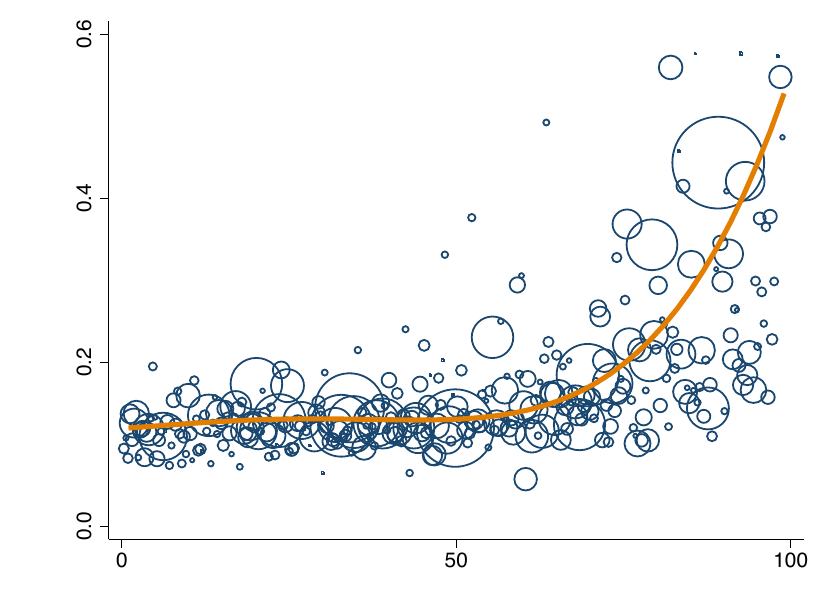} \caption{1980}
\end{subfigure}
\begin{subfigure}{0.47\textwidth}
\includegraphics[width=\textwidth,height=0.25\textheight]{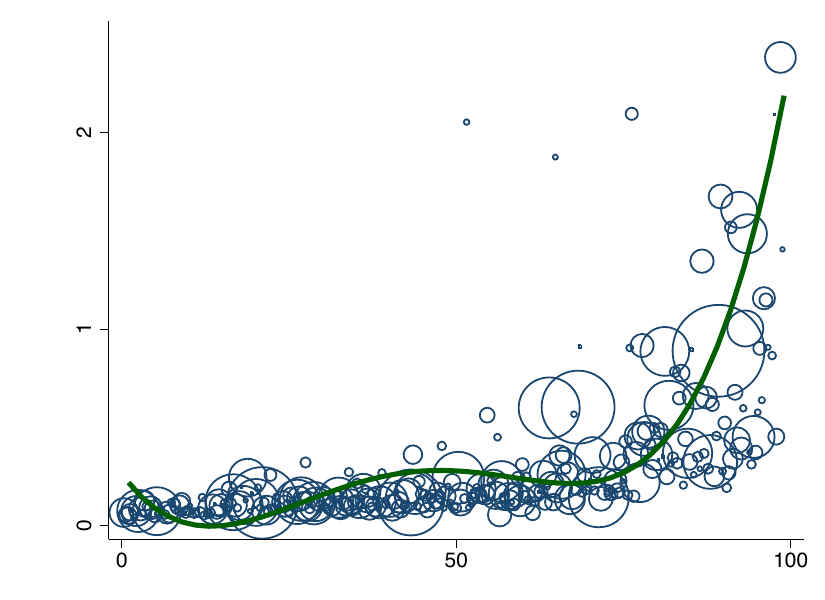} \caption{2005}
\end{subfigure}
\par\end{center}
\vspace{-0.75cm}

\caption{Earnings Dispersion by Occupation}
{\scriptsize{}{}\vspace{0.2cm}
 \Cref{fig:data} displays earnings dispersion within occupation.
On the horizontal axis, we rank occupations by the average earnings earned in each occupation. For every occupation, we calculate the dispersion in earnings within that occupation. The size of the circle indicates the share of employment within the occupation. The data pattern is summarized by the solid lines. \label{fig:data}}

 \end{figure}


\vspace{0.4 cm}
\noindent Our data sources are the Census IPUMS sample for 1980 and the American Community Survey for 2005. We consider individuals between ages 16 and 64 who worked during the previous year. Our measure of earnings is annual wage and salary income before taxes. Earnings are normalized by average earnings in the economy. The data contains information on the individual's occupation and thus can be used to calculate the earnings distribution within every occupation.\footnote{In order to ensure consistent definitions of occupations across the different years, we follow \citet{Autor:2013} by distinguishing 323 different non-farm occupations.}


Figure \ref{fig:data} shows earnings dispersion within occupations in 1980 and 2005. On the horizontal axis, we rank occupations by average earnings in each occupation. For every occupation, we calculate the variance of earnings within that occupation, where the circle size indicates the share of total employment within the occupation. The salient data patterns are captured by the solid lines. First, earnings dispersion within occupations is sizable, with average earnings variation within occupations equal to 0.18 in 1980 and equal to 0.39 in 2005. Second, earnings dispersion within occupations is relatively constant at the bottom two-thirds of occupations but increases for the top third of occupations.


The total variation in earnings can be decomposed into two terms. The first term is the variance of mean earnings across occupations, or the between-occupation variance. The second term is the average of within-occupation dispersion of worker earnings weighted by employment.  The total variation in earnings in 1980 equals 0.22. Of this total variation, 0.18 is accounted for by earnings dispersion within occupations, while 0.04 is accounted for by wage dispersion across occupations. From 1980 to 2005 the variation in earnings increased by 0.31 to 0.53. A third of this increase is attributed to increased variation between occupations, while two-thirds is attributed to increased variation within occupations. Over the same time period, the rank correlation between workers and jobs increased from 0.42 to 0.53.



%
%
%
%
%
%
%
%

\vspace{0.4 cm}
\noindent We evaluate the ability of our model to generate dispersion in earnings, its decomposition between and within occupations, as well as the rank correlation between workers and occupations. We then use the model to decompose increased earnings dispersion in the United States between 1980 and 2005 into supply and demand side factors. 


We parameterize the economy separately for 1980 and 2005. The worker distribution $F$ is a discretized lognormal distribution, $\log x \sim \mathcal{N}(\mu_x,\sigma^2_x)$. We set average worker skills to one in levels, or $\mu_x = -\sigma_x^2/2$, and set $\alpha(x) = x$. The distribution of jobs $G$ is a discretization of a mixture of two lognormal distributions with mean $\mu_{i}$ and variance $\sigma^2_{i}$ for each distribution $i \in \{ 1, 2\}$. The mixing weight on the first distribution is $p$.\footnote{The mismatch penalties exclusively depend on the distributions of worker skills and job complexity, not on the innate productive value of workers $\alpha$ and jobs $\theta$ (see Section \ref{s:wagesandfirmvalues}). In order to study earnings and earnings dispersion in occupations we thus do not need to specify $\theta$.}

We use the cross-sectional earnings distribution and earnings dispersion within occupations to inform the worker and job distributions. Specifically, we choose model parameters to minimize the squared loss between the following model and data statistics: earnings at each percentile (\Cref{f:model_earnings}), earnings dispersion within occupation at each occupation rank (\Cref{f:model_assignment}), and the rank correlation between workers and jobs. The model parameters are summarized in \Cref{t:modelparam}.\footnote{The parameters imply that there are more low and high-complexity jobs and fewer jobs with medium complexity in 2005. In line with the labor market polarization literature, we thus find that the distribution of jobs has become more polarized \citep{Acemoglu:1999,Autor:2013}.} In order to generate increased dispersion in earnings, the model variance of worker skill $\sigma^2_x$ increases by 0.16 from 1980 to 2005. 


\begin{figure}[!t]
\begin{center}
\begin{subfigure}{0.47\textwidth}
\includegraphics[width=\textwidth,height=0.25\textheight]{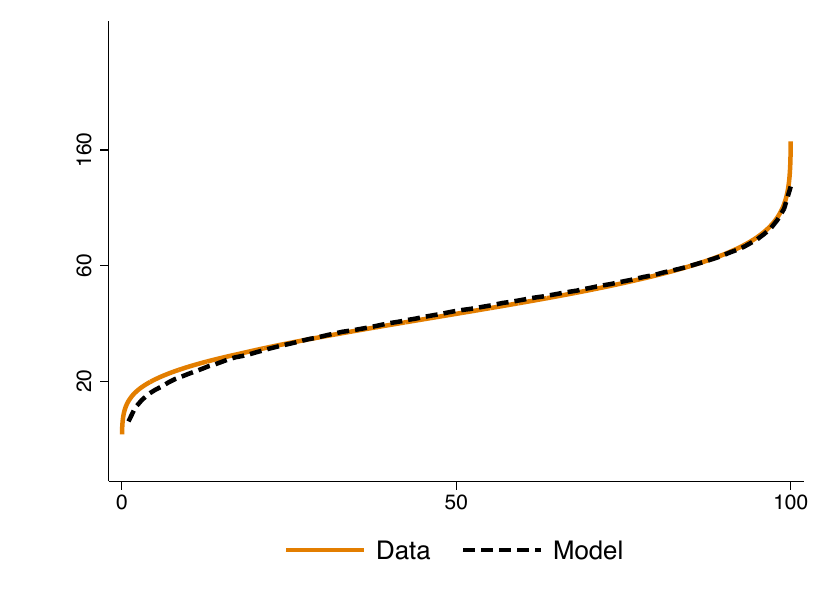} \caption{1980}
\end{subfigure}
\begin{subfigure}{0.47\textwidth}
\includegraphics[width=\textwidth,height=0.25\textheight]{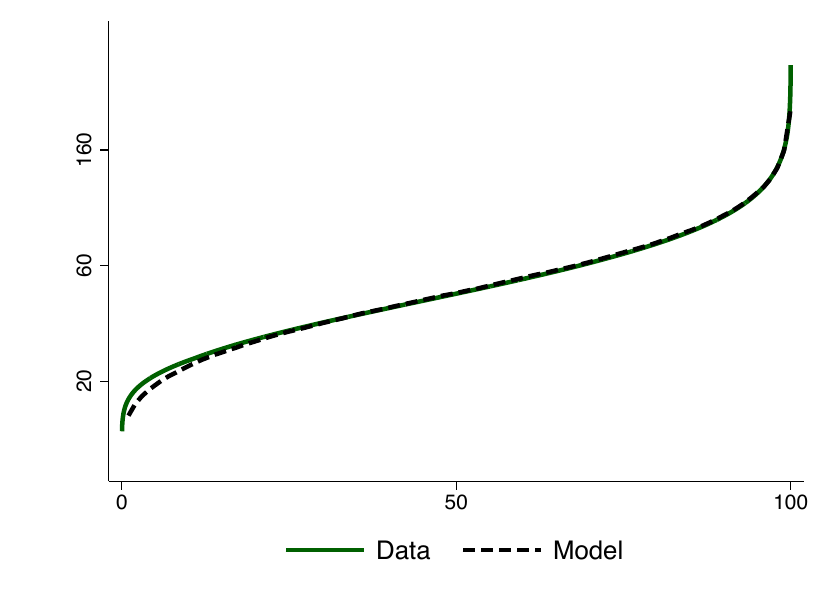} \caption{2005}
\end{subfigure}
\par\end{center}

\vspace{-0.75 cm}

\caption{Data and Model Distributions of Individual Earnings}
 {\scriptsize \vspace{0.2 cm} \Cref{f:model_earnings} compares the empirical earnings distribution to the model earnings distribution. The empirical distributions are represented by solid colored lines, while the model distributions are presented by black dashed lines. The left panel shows the empirical and model distribution for 1980; the right panel for 2005.\label{f:model_earnings} } 

\end{figure}

\begin{figure}[!t]
\begin{center}
\begin{subfigure}{0.47\textwidth}
\includegraphics[width=\textwidth,height=0.25\textheight]{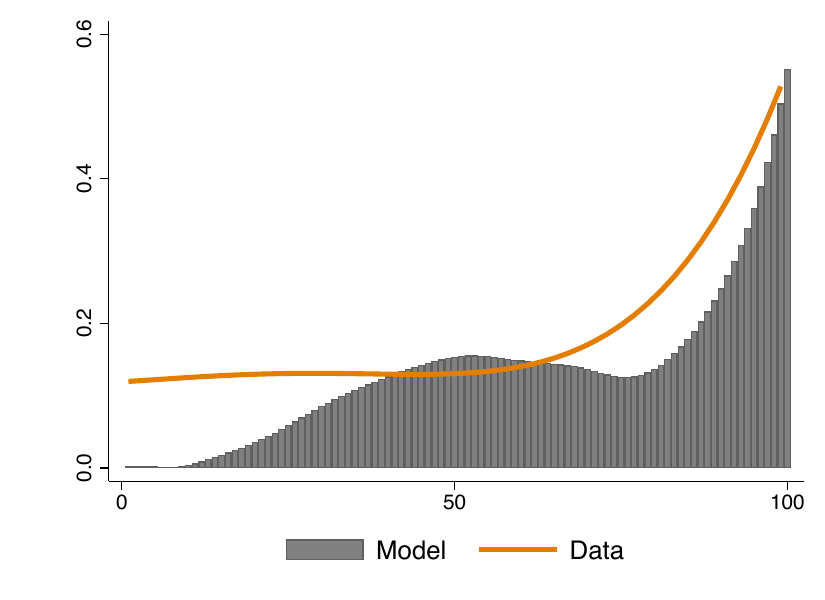} \caption{1980}
\end{subfigure}
\begin{subfigure}{0.47\textwidth}
\includegraphics[width=\textwidth,height=0.25\textheight]{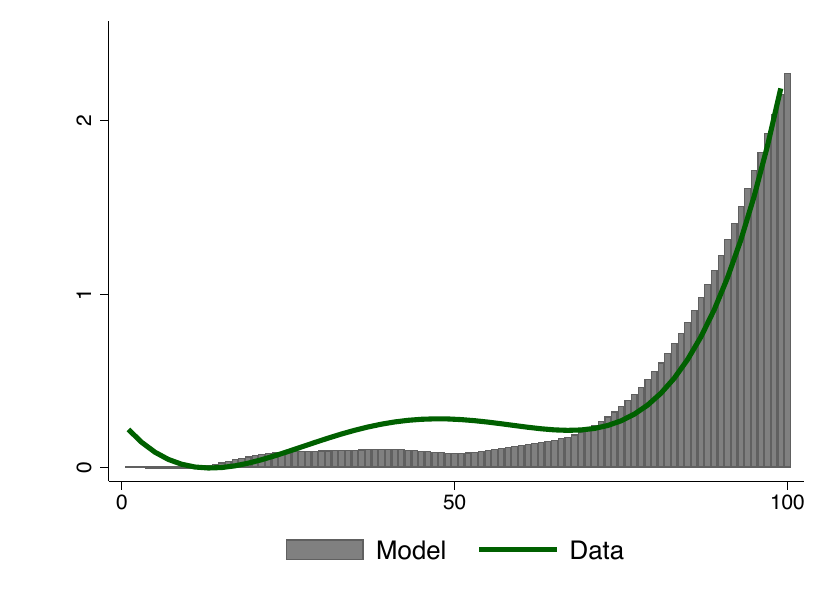} \caption{2005}
\end{subfigure}
\par\end{center}

\vspace{-0.75 cm}

\caption{Data and Model Earnings Dispersion Within Occupation}
{\scriptsize{}{}\vspace{0.2cm}
 \Cref{f:model_assignment} displays earnings dispersion within occupations for the model and for the data. On the horizontal axis, we rank occupations by the average earnings earned in each occupation. The solid lines are a fractional polynomial fit that captures the salient data patterns (\Cref{fig:data}). The bars represent the variation in earnings across model occupations.\label{f:model_assignment} }

\end{figure}


Figure \ref{f:model_earnings} shows that our model generates the observed dispersion in earnings, as well as the changes in the distribution of earnings over time. \Cref{f:model_assignment} shows earnings dispersion within occupations for the model and for the data. The model accounts for three-quarters of the absolute deviation in the middle of the distribution (20th to 80th percentile) and generates the significant increase in within-occupation earnings dispersion at the top of the distribution (80th to 100th percentile). 


\begin{table}[!t]
\def\arraystretch{1.35}
\begin{center}
\caption{Model and Data Earnings Decomposition} \label{t:decomposition}
\begin{tabular}{l|ccc|ccc}
\hline  \hline
\hspace{2.5 cm} & \multicolumn{3}{c|}{Data}     & \multicolumn{3}{c}{Model}     \\
Moment  \hspace{0.32 cm} & \hspace{0.32 cm} 1980 \hspace{0.32 cm}  & \hspace{0.32 cm} 2005 \hspace{0.32 cm} & \hspace{0.32 cm} change \hspace{0.32 cm} & \hspace{0.32 cm} 1980 \hspace{0.32 cm} & \hspace{0.32 cm} 2005 \hspace{0.32 cm} & \hspace{0.32 cm} change \hspace{0.32 cm} \\ \hline
Total & 0.22     & 0.53 & 0.31 & 0.22     & 0.53 & 0.31 \\
Between & 0.04     & 0.14 & 0.10 & 0.09     & 0.17 & 0.08 \\
Within & 0.18     & 0.39 & 0.21 & 0.13     & 0.36 & 0.23 \\ \hline
Rank correlation \hspace{0.17 cm} & 0.42 & 0.53 & 0.11 & 0.45 & 0.49 & 0.04 \\
 \hline \hline
\end{tabular}
\end{center}

\vspace{-0.3cm}

{\scriptsize \vspace{.1 cm} \Cref{t:decomposition} compares the empirical and model decomposition of earnings dispersion. The left panel shows the empirical decomposition of earnings variation in dispersion between occupations and dispersion within occupations, while the right panel shows the model analog.}
\end{table}

\Cref{t:decomposition} displays the empirical and structural decomposition of earnings variation in 1980 and 2005. Our stylized model can capture the salient patterns in the decomposition of earnings variation and its changes over time. The model generates an increase in within-occupation earnings variation of 0.23 points to 0.36 in 2005, which captures the observed increase in within-occupation variation of 0.21 points to 0.39 in the data. The model also generates an increase in the rank correlation, from 0.45 to 0.49.


The key feature of our model is its ability to generate earnings dispersion within the same occupation. The theoretical analysis above shows that there is dispersion of skills in the same occupation and thus there is corresponding dispersion of earnings within occupations. In contrast, a classic sorting model that delivers either positive or negative sorting cannot generate variation in skill levels within a particular occupation, and as a consequence, does not generate earnings dispersion within that occupation. Any model that pairs only one worker type to each job does not generate any earnings dispersion within occupations.

\begin{table}[!t]
\def\arraystretch{1.35}
\begin{center}
\caption{Earnings Relative to Peers} \label{t:wages_peers}
\begin{tabular}{c|cc|cc}
\hline  \hline
\hspace{1.5 cm} & \multicolumn{2}{c|}{Data}     & \multicolumn{2}{c}{Model}     \\
 \hspace{0.68 cm}  Percentile  \hspace{0.68 cm} & \hspace{0.68 cm} 1980 \hspace{0.68 cm}  & \hspace{0.68 cm} 2005 \hspace{0.68 cm}  & \hspace{0.68 cm} 1980 \hspace{0.68 cm} & \hspace{0.68 cm} 2005 \hspace{0.68 cm} \\ \hline
25 & 0.74     & 0.69   & 0.83    & 0.59 \\
50 & 0.93     & 0.86   & 0.84   & 0.84  \\
75 & 1.15     & 1.09   & 1.03    & 1.03 \\ 
90 & 1.38     & 1.33   & 1.24    & 1.43 \\ 
 \hline \hline
\end{tabular}
\end{center}

\vspace{-0.3cm}

{\scriptsize \vspace{.1 cm} \Cref{t:wages_peers} compares the model and data in terms of mean coworker earnings. The left panel shows mean coworker earnings at selected percentiles of the individual earnings distribution. The right panel shows their model analog.}
\end{table}

In addition to analyzing the variance decomposition of earnings, we evaluate the non-targeted earnings of workers relative to the earnings of workers in the same occupation at different percentiles in the earnings distribution. The relative earnings of a worker compared to their peers equals one when sorting is assortative, or, more generally, when sorting is one-to-one. \Cref{t:wages_peers} displays the earnings of workers relative to their peers in the model and the data. The workers at the median of the distribution earn 7 percent less than their peers in 1980 and 14 percent less in 2005. In the model, this workers earns 16 percent less than their peers in both 1980 and 2005. Model and data align qualitatively in both years.\footnote{We provide detailed intuition for the quantitative results by analyzing the equilibrium for 1980 in Technical \Cref{s:intquant}.}



\vspace{0.4 cm}
\noindent The changes in observed earnings patterns between 1980 and 2005 can be driven by changes in the supply or the demand side of the labor market, or by their combination. We use our framework to evaluate the drivers of changing earnings patterns. In order to do this, we analyze counterfactual changes in earnings by only changing the distribution of workers and by only changing the distribution of jobs.

\begin{table}[!t]
\def\arraystretch{1.35}%
\begin{center}
\caption{Counterfactual Decomposition of Earnings Variation} \label{t:decomposition_structural}
\begin{tabular}{l|ccc|cc|cc}
\hline  \hline
\hspace{1.5 cm} & \multicolumn{3}{c|}{Model}     & \multicolumn{2}{c|}{Job Effect} & \multicolumn{2}{c}{Worker Effect} \\
Moment  \hspace{0.22 cm} & \hspace{0.22 cm} 1980 \hspace{0.22 cm}  & \hspace{0.22 cm} 2005 \hspace{0.22 cm} & \hspace{0.22 cm} change \hspace{0.22 cm} & \hspace{0.22 cm} 2005 \hspace{0.22 cm} & \hspace{0.22 cm} change \hspace{0.22 cm} & \hspace{0.22 cm} 2005 \hspace{0.22 cm} & \hspace{0.22 cm} change \hspace{0.22 cm} \\ \hline
Between & 0.41     & 0.32 & $-0.09$  & 0.45 &  \phantom{$-$}0.04  & 0.25 & $-0.16$  \\
Within & 0.59    & 0.68 &  \phantom{$-$}0.09      & 0.55 & $-0.04$  & 0.75 & \phantom{$-$}0.16  \\  \hline
Correlation \hspace{0.07 cm} & 0.45 & 0.49 & \phantom{$-$}0.04 & 0.53 & \phantom{$-$}0.08 & 0.40 & $-0.05$  \\
 \hline \hline
\end{tabular}
\end{center}

\vspace{-0.3cm}

{\scriptsize \vspace{.1 cm} \Cref{t:decomposition_structural} compares the baseline and counterfactual model decomposition of earnings dispersion. The left panel shows the baseline model decomposition of earnings variation as in \Cref{t:decomposition}, while the middle and right panels show counterfactual decompositions. For the job effect counterfactual, we evaluate the model with the worker distribution for 1980 and the job distribution for 2005, while the worker effect counterfactual evaluates the model using the job distribution for 1980 and the worker distribution for 2005.}
\end{table}

\Cref{t:decomposition_structural} shows a decomposition of changes in earnings dispersion from 1980 to 2005. The left panel repeats the baseline model decomposition of \Cref{t:decomposition}, while the middle panel and the right panel present counterfactual results. The job effect counterfactual evaluates the model with the distribution of workers in 1980 and the distribution of jobs in 2005. The middle panel shows that by only changing the distribution of jobs, the share of within-occupation earnings dispersion would have decreased by 4 percentage points, while the rank correlation would have increased by 0.08. The worker effect counterfactual similarly evaluates the model using the job distribution of 1980 and the worker distribution of 2005. The right panel shows that the share of the within-occupation earnings dispersion would have increased by 16 percentage points, while the rank correlation would have decreased by 0.05. Our counterfactual analysis thus shows that both the changes in the worker and job distributions between 1980 and 2005 are important in generating the observed changes in the composition of earnings dispersion.

\section{Conclusion}

We characterize the optimal assignment, wages and comparative statics for 
an assignment problem with heterogeneous workers and
jobs and concave costs of skill gaps, a technology that is neither supermodular
nor submodular. We show that concavity generally arises when firms
make investments to mitigate the cost of skill gaps as in \citet{Stigler:1939} and \citet{Laffont:1986,Laffont:1991}. Our analysis introduces
composite sorting where multiple worker types are sorted to the same occupation
and worker types are simultaneously part of both positive and negative
sorting. Our first main result is to show that composite sorting has a particularly tractable structure when the distributions of workers and jobs follow normal mixture distributions. This empirically relevant case provides a significant reduction in complexity that facilitates our quantitative analysis with a large number of skill types. Our second set of results is on constructing the dual solution of the problem to determine wages and firm values. We show that wages have a striking regional hierarchical structure with relative earnings determined within skill groups and aggregated to determine earnings at different scales. Third, leveraging the characterization of optimal sorting, we derive comparative statics for the optimal assignment and show that sorting is more positive when the cost of
skill gaps is less concave. Moreover, we prove the existence of a threshold level of concavity beyond which sorting is positive in each layer $-$ the layered positive sorting. Our quantitative model can generate and help explain earnings dispersion within occupations, as well its changes over time. In sum, composite sorting provides a tractable assignment framework in between the polar cases of supermodularity and submodularity and delivers strikingly different results from the canonical assortative sorting models. Our results on the assignment problem, the dual solution, and comparative statics are new to both the economics and optimal transport literature.

\clearpage

%
%


%
%
%
%
  
{ {
\bibliographystyle{econometrica}	
\baselineskip15.0pt
\bibliography{bib}
} }

\clearpage

\pagebreak

\renewcommand{\theequation}{A.\arabic{equation}} \setcounter{equation}{0}
\renewcommand{\thefigure}{A.\arabic{figure}}\setcounter{figure}{0}
\renewcommand{\thetable}{A.\arabic{table}}\setcounter{table}{0}
\setcounter{page}{1}
\newpage \appendix
\newpage

\begin{center}
{\Large Composite Sorting \\}
\bigskip
{\Large Online Appendix \\}
\bigskip
{\large  Job Boerma, Aleh Tsyvinski, Ruodu Wang, and Zhenyuan Zhang \\}
\bigskip
{\large May 2025}
\end{center}
\vspace{0.8cm}


\section{Proofs}

In this appendix, we formally prove the results in the main text.

\subsection{Necessary Conditions for Optimal Sorting} \label{a:necconos}

In this appendix, we provide the proofs of the results in Section \ref{s:necessarycond}.

\subsubsection{Maximal Perfect Pairs} \label{p:common}

In this appendix, we formally state and prove the result on maximal perfect pairs.

\begin{lemma}{\textit{Maximal Perfect Pairs}.}\label{l:common}
Let $F\wedge G$ denote the common component of the worker distribution
$F$ and the job distribution $G$. Any optimal assignment $\pi$
between workers and jobs consists of perfect pairings on the support of $F\wedge G$
and an optimal assignment between workers $F-F\wedge G$ and jobs
$G-F\wedge G$. \end{lemma}

\begin{proof}We show that a perfect pair is made when feasible. By contradiction, suppose an optimal assignment contains pairings $(x,z)$ and $(x',z')$ when $x'=z$.\footnote{If either $x = x'$ or $z = z'$, a perfect pairing is naturally made since $x' = z$. We thus restrict our attention to the cases where $x \neq x'$ and $z \neq z'$.} 


By symmetry, it suffices to consider two cases. Consider first the case $x\leq z'<z=x'$. Since the cost of skill gaps $\bar{c}$ is strictly increasing, $c(x,z') + c(x',z) = c(x,z') < c(x,z) \leq c(x,z)+c(x',z')$ where the equality follows since $c(x',z)=0$. Thus, the cost of mismatch when making the perfect pairing is strictly lower than under the optimal configuration, which is a contradiction. 

Second, we consider the case where $x<x'=z<z'$. In this case, the cost of mismatch $\bar{c}$ is given by $c(x,z')+c(x',z) = \bar{c}(z'-x)$ since $c(x',z)=0$. To arrive at a contradiction, choose some weight $\lambda \in (0,1)$ to scale the maximum distance such that $z-x = (1-\lambda)(z' - x)$. Since the total distance is given by $(z - x) + (z' - x') = (z' - x)$, we also have $z'-x' = \lambda (z' - x)$. Since the cost of skill gaps is strictly concave, we use strict concavity and add the two previous equations to obtain $\bar{c}(z-x) + \bar{c}(z'-x') > \bar{c}(z'-x)$.\footnote{We introduce the notation $\bar{c}$ to denote the cost of mismatch  \eqref{eq:cxz} when the worker is underqualified, or $z> x$. Similarly, we use the notation $\underline{c}$ to denote the cost of mismatch when the worker is overqualified, or $z<x$.} The output loss can be strictly reduced by assigning worker $x$ to job $z'$ and by perfectly assigning worker $x'$ to job $z$, which is a contradiction.
\end{proof}

\subsubsection{No Intersecting Pairs} \label{p:nocross}

In this appendix, we formally state and prove the result on no intersecting pairs.

 \begin{lemma}{\textit{No Intersecting Pairs}}. \label{l:nocross}
Let $\pi$ be an optimal assignment. For any two pairs $(x,z)$ and
$(x',z')$ in the support $\Gamma_{\pi}$, their arcs do not intersect.
\end{lemma}

\begin{proof} To establish the result, we show that if two pairings $(x,z)$ and $(x',z')$ under an optimal assignment intersect, then the support of the assignment is not optimal.

By symmetry, it suffices to consider two cases. First, consider the case $x < z'< z <x'$. Since the cost function is increasing, $c(x,z') + c(x',z) = \bar{c}(z'-x) + \underline{c}(x' - z)< \bar{c}(z-x) + \underline{c}(x'-z') = c(x,z) + c(x',z')$. The output loss due to skill gaps is strictly reduced by assigning worker $x$ to job $z'$ and worker $x'$ to job $z$, which is a contradiction.



Second, consider the case $x < x' < z < z'$. In this case, the cost of skill gaps is $c(x,z')+c(x',z) = \bar{c}(z'-x) + \bar{c}(z - x')$. To arrive at a contradiction, choose some weight $\lambda \in (0,1)$ to average the minimum and maximum distance such that:
\begin{equation*}
z-x = (1-\lambda)(z' - x) + \lambda (z-x').
\end{equation*}
Since $(z - x) + (z' - x') = (z' - x) + (z - x')$, we moreover write: 
\begin{equation*}
z'-x' = \lambda (z' - x) + (1-\lambda) (z-x').
\end{equation*}
Since the cost of skill gaps is strictly concave, we can use strict concavity and add the two previous equations to obtain $\bar{c}(z-x) + \bar{c}(z'-x') > \bar{c}(z'-x) + \bar{c}(z-x')$. The output loss due to skill gaps can be strictly reduced by assigning worker $x$ to job $z'$ and worker $x'$ to job $z$, which is a contradiction.\end{proof}

\subsection{Proof of \Cref{t:numberpair}} \label{p:nm}

To prove \Cref{t:numberpair} we prove a general result on the number of sign changes for a linear combination of normal densities in \Cref{t:sign}. We start this section by formally defining sign changes. We then state \Cref{t:sign} and use it to prove \Cref{t:numberpair}. Finally, we prove \Cref{t:sign}.

We first define sign changes. For a given sequence $(x_1,\dots,x_k)$, we define $S^-(x_1,\dots,x_k)$ as the number of sign changes of the sequence, that is, the number of $j=1,\dots,k-1$ such that $x_jx_{j+1}<0$. For a function $f$ on an interval $I$, the number of sign changes is defined as $S^-(f)=\sup \hspace{0.01 cm} S^-(f(t_1),\dots,f(t_k))$, where the supremum is over all sets of $t_1 < t_2 < \dots < t_k$ such that $t_i \in I$ and $k$ is finite. Using this definition, we state \Cref{t:sign}.

\begin{theorem}\label{t:sign}
    Let $n\in\N$, $c_i \in\R$, $u_i \in \R$ and $s_i \in \R_+$. Define
    \begin{align}
        \phi(x)=\sum_{i=1}^n \frac{c_i}{\sqrt{\pi s_i}}e^{-\frac{(x-u_i)^2}{s_i}}.\label{eq:phiform}
    \end{align}
    Then $S^-(\phi)\leq 2(n-1)$.
\end{theorem}

\begin{proof}[Proof of Theorem \ref{t:numberpair}]
Let the normal mixture density functions of the worker distribution and the job distributions be given. The difference between the two density functions $h = f - g$ is of the form (\ref{eq:phiform}). By Theorem \ref{t:sign}, $S^-(h)\leq 2(n+m-1)$. It follows that any layer $\bar{h}$ in the measure of underqualification $H$ has at most $2(n+m-1)+1$ points. Since the number of points on every layer is even, it follows that there are at most $n+m-1$ pairs on every layer.
\end{proof}

\noindent It remains to prove \Cref{t:sign}. In order to do so, let $\phi_s(x)=\frac{1}{\sqrt{\pi s}} \exp \big( -\frac{x^2}{s} \big)$ for $s > 0$ and, in addition, let $\mathcal{S}$ be the Schwartz space, the space of all functions whose derivatives are rapidly decreasing:
\begin{equation*}
\mathcal{S}= \big\{f\in C^\infty(\R):\sup_{x} \big|x^nf^{(m)}(x) \big|<\infty,\,\forall \; n,m\in\{0,1,2,\dots\}\big\}
\end{equation*}
and denote its dual by $\mathcal{S}'$, which is the set of all tempered distributions and includes, for instance, $\mathcal{S}$ and Dirac delta functions. For densities $f\in\mathcal{S}$ and $g\in\mathcal{S}'$, denote by $f\ast g$ the convolution of $f$ and $g$. It follows that $\phi_s\ast\phi_{s'}=\phi_{s+s'}$. To prove \Cref{t:sign}, we use \Cref{lemma:total positivity2}. In turn, we use \Cref{lemma:total positivity} from \citet{Karlin:1968} to prove \Cref{lemma:total positivity2}.\footnote{See Theorem 3.1 in \citet{Karlin:1968}.}


\begin{lemma} \label{lemma:total positivity}
    It holds for all $s> 0$ and continuous real functions $\psi$ that $S^-(\psi \ast \phi_s) \leq S^-(\psi)$.
\end{lemma}

\begin{lemma}\label{lemma:total positivity2}
    Let $\tilde{\psi} \in \mathcal{S}'$ be the sum of a Schwartz function $\psi\in\mathcal{S}$ and a linear combination of finitely many Dirac delta functions. Then for all $s>0$, $S^-(\tilde{\psi} \ast \phi_s)\leq S^-(\tilde{\psi})$.
\end{lemma}

\begin{proof}
Without loss of generality, we assume $\tilde{\psi}=\delta_0+\psi$ where $\psi\in\mathcal{S}$ is a Schwartz function and $\delta_0\in\mathcal{S}'$ is the Dirac delta function at $0$. The general case follows the same proof with heavier notation. 
   For a continuous real function $\psi$ with $S^-(\psi)<\infty$, there exists $\delta>0$ such that for all continuous real functions $\hat{\psi}$ satisfying $\|\hat{\psi}-\psi\|_\infty<\delta$, we have $S^-(\hat{\psi})\geq S^-(\psi)$.  

   
   
Suppose for contradiction that for some $s>0$ we instead have $S^-(\tilde{\psi} \ast \phi_s)> S^-(\tilde{\psi})$. For $\varepsilon>0$, define the piecewise linear wedge function $\psi_\varepsilon$:
\begin{align*}
    \psi_\varepsilon(x)=\begin{cases}
        \frac{1}{\varepsilon}-\frac{x}{\varepsilon^2} & \hspace{1 cm} \text{ if } x \in [0,\varepsilon]\\
        \frac{1}{\varepsilon}+\frac{x}{\varepsilon^2} &\hspace{1 cm} \text{ if } x \in [-\varepsilon,0]\\
        0& \hspace{1 cm} \text{ elsewhere}.
    \end{cases}
\end{align*}  
Since $\phi_s$ is uniformly continuous, $\psi_\varepsilon \ast \phi_s\to\phi_s$ in $L^\infty$ as $\varepsilon\to 0$.\footnote{See Theorem 1.2.19 in \citet{Grafakos:2014}.} Therefore, for $\varepsilon$ small enough, it holds that $S^-(\tilde{\psi}\ast \phi_s)\leq S^-((\psi+\psi_\varepsilon)\ast \phi_s)$ by the previous paragraph.
    By Lemma \ref{lemma:total positivity}, $S^-(\tilde{\psi}\ast \phi_s)\leq S^-(\psi+\psi_\varepsilon)$. 
    
    On the other hand, we claim that for $\varepsilon$ small enough it holds $S^-(\psi+\psi_\varepsilon)=S^-(\psi+\delta_0)$. Indeed, since $S^-(\psi)<\infty$, there are three cases:
    \begin{enumerate}[noitemsep]
        \item If $\psi\geq 0$ in a neighborhood $N$ of $0$, $\psi+\psi_\varepsilon$ and $\psi+\delta_0$ both have no sign change in $N$, and have the same number of sign changes outside $N$ if $\delta$ is chosen small enough.
        \item If $\psi\leq 0$ in a neighborhood $N$ of $0$, $\psi+\delta_0$ have two sign changes in $N$. Since $\psi\in\mathcal{S}$ implies that both $|\psi|$ and $|\psi'|$ are uniformly bounded, $\psi+\psi_\varepsilon$ also has two sign changes in $N$ for $\varepsilon$ small enough. They have the same number of sign changes outside the neighborhood $N$ if $\varepsilon$ is small enough.
        \item if $\psi\leq 0$ at $0^-$ and $\psi\geq 0$ at $0^+$ (or vice versa), $\psi+\psi_\varepsilon$ and $\psi+\delta_0$ both have one sign change in a neighborhood $N$ of $0$ for $\varepsilon$ small, and the same number of sign changes  outside $N$ if $\varepsilon$ is chosen small enough.
    \end{enumerate}
    As a result, we can write down the following chain of inequalities
\begin{equation*} 
S^-(\tilde{\psi})<S^-(\tilde{\psi} \ast \phi_s)\leq S^-(\psi+\psi_\delta)=S^-(\psi+\delta_0)=S^-(\tilde{\psi}) ,
\end{equation*}
which gives a contradiction.
\end{proof}

\begin{proof}[Proof of Theorem \ref{t:sign}] We prove the result by induction. The base is $n=1$, when there is a single normal distribution and hence no sign change, or $S^-(\phi)=0$. Consider $n\in\N$ and $\phi$ of the form in equation \eqref{eq:phiform}. Let $t$ be the minimal variance across the normal distributions, or $t=\min\limits_{1\leq i\leq n} s_i$. Without loss of generality, let there be $k \geq 1$ different normal distributions with the lowest variance, indexed such that $\{i:s_i=t\}=\{n-k+1,\dots,n\}$. 

Observe that convolution with a delta function is equivalent to a lateral shift. It follows that $\phi=\phi_t\ast\tilde{\psi}$, where:
\begin{equation}
\tilde{\psi}(x)=\sum_{i=1}^{n-k} \frac{c_i}{\sqrt{\pi (s_i-t)}}e^{-\frac{(x-u_i)^2}{s_i-t}}+\sum_{i=n-k+1}^nc_i\delta_{u_i}=: \tilde{\phi} (x) + \sum_{i=n-k+1}^n c_i\delta_{u_i}\in\mathcal{S}' \label{e:tildepsi}
\end{equation}
Since by definition $\tilde{\phi}$ is a mixture of normal distributions, $\tilde{\phi}\in\mathcal{S}$.
By Lemma \ref{lemma:total positivity2}, it thus follows that: 
\begin{equation}
S^-(\phi) = S^-(\phi_t \ast \tilde{\psi} ) \leq S^-(\tilde{\psi})\leq 2k+S^-(\tilde{\phi}). \label{e:sminphi}
\end{equation}
where the first inequality follows by Lemma \ref{lemma:total positivity2}. The second inequality follows since each delta function adds at most two more sign changes, and there are $k$ delta functions in (\ref{e:tildepsi}). 

We note that $\phi$ was of the form in equation \eqref{eq:phiform} and similarly that $\tilde{\phi}$ is also of the form \eqref{eq:phiform} with $n-k$ remaining normal distributions. 
    Applying the induction hypothesis, we have $S^-(\tilde{\phi})\leq 2(n-k-1)$. Therefore, we conclude $S^-(\phi)\leq 2(n-1)$, as desired. 
\end{proof}

\subsection{Local Hierarchical Algorithm} \label{a:dualproof}

We propose a new algorithm specifically tailored to the model of composite sorting, which has two distinct merits.  First, it is typically more efficient than existing generic algorithms, as shown in Technical Appendix \ref{s:efficiency}. Second, and more importantly, this new algorithm reveals
a hierarchical structure of the dual potential functions, 
highlighting an implication of the absence of intersecting pairs to dual optimizers.

Recall that in our setting, we consider a problem with a finite number of workers with skill levels in 
$X$ and a finite number of jobs with difficulty levels in  $Z$, where $X$ and $Z$ are disjoint sets. We denote by $S=X\cup Z$ the set of all skill levels. Moreover, we recall that we construct the dual solution given an optimal sorting $\pi$.

Our algorithm relies on recursive computations of $\phi$ constrained on smaller subsets of $S$. 
To explain such a recursive procedure, we introduce the notion of subpairs.  
A pair $(x,z)$ is called a subpair of the pair $(x_{0},z_{0})$ if
$(x,z)$ is a non-nested pair inside the interval $[x_{0},z_{0}]$ that is not equal to $(x_0,z_0)$.

We process each pair  $(x_0,z_0)\in \Gamma_\pi$
sequentially to get a local dual optimizer on ${[x_0,z_0]}$, that is, a function
$\phi_{[x_{0},z_{0}]}$ such that for any $(x,z)\in X_{[x_{0},z_{0}]}\times Z_{[x_{0},z_{0}]}$,
$\phi_{[x_{0},z_{0}]}(x)-\phi_{[x_{0},z_{0}]}(z)\leq c(x,z)$ and it
holds with equality when $(x,z)\in\Gamma_\pi$. {We observe that this property is preserved if $\phi_{[x_{0},z_{0}]}$
is shifted by any constant $a\in \mathbb R$.} 


Below is a recursive construction of the wage penalties, where  $\phi$, $\phi_i$, $x_i$ and $z_i$ are local variables that vary across each iteration, and $\phi_{[x,z]}$ for $(x,z)\in \Gamma_\pi$ are global variables that are the output of the algorithm.  
\begin{enumerate}[noitemsep]
    \item Pick any pair $(x_0,z_0)\in \Gamma_\pi$ that has not been processed such that all subpairs of $(x_0,z_0)$ have been processed. Let $(x_1,z_1),\dots,(x_p,z_p)$ be the subpairs of $(x_0,z_0)$  ordered in a way that:
\begin{align*}
    |x_1-x_0|=\min_{i\in\{1,\dots,p\}}|x_i-x_0|\quad \text{and}\quad |z_p-z_0|=\min_{i\in\{1,\dots,p\}}|z_i-z_0|,
\end{align*}
     with potential functions $\phi_i:=\phi_{[x_i,z_i]}$ 
    for $ i=1,\dots, p$.  
    \item If $p=0$,  then let $\phi (z_0)=0$ and $\phi (x_0)=c(x_0,z_0)$.
    \item If $p\ge 1$, then continue with the following sub-steps.
    \begin{enumerate}
\item       If $p>1$,   let $(\beta_2,\dots,\beta_p)\in \R^{p-1}$ 
    be a solution to the inequality system \begin{align} \tag{\ref{eq:ineqs}}  \max(c_{00}-c_{0n}-c_{m0}, -c_{mn}) + c_{nn} \leq \sum_{k=n+1}^m\beta_k\leq \min(c_{0m}+c_{n0}&-c_{00} ,c_{nm}) -c_{mm} 
    \end{align} 
    for all $1\le n<m\le p$,
    where $c_{ij}=c(x_i,z_j)$ for $i,j\in\{1,\dots,p\}$. We show the existence of such $(\beta_2,\dots,\beta_p)$ in Lemma \ref{lemma:dualalgorithm} below.
    \item For $i=1,\dots,p $,
    let \begin{align}
        \phi(s)= \phi_i(s)+ \sum_{k=i+1}^p \beta _k +\phi_p(x_p)-\phi_i(x_i) \tag{\ref{eq:phidef}}
    \end{align}
    for $s\in X_{ [x_i,z_i]}$ or $s\in Z_{[x_i,z_i]}$.  The above sum $\sum\limits_{k=i+1}^p \beta _k$ is $0$ if $i=p$.
    \item Define $\phi(x_0)$ and $\phi(z_0)$ according to: 
    \begin{align}
    \phi(z_0) & =\begin{cases}
        \displaystyle\max_{i\in \{1,\dots,p\} } (\phi(x_i)-c_{i0}))&\text{ if }x_1\neq x_0,~ z_0 = z_p \\
        \displaystyle \left[ \max_{i\in \{1,\dots,p\} } (\phi(x_i)-c_{i0}), \min_{i\in \{1,\dots,p\} }(\phi(z_i)+c_{0i})-c_{00} \right] &\text{ if }x_1\neq x_0,~ z_0\neq z_p\\
        \displaystyle\min_{i\in \{1,\dots,p\} }(\phi(z_i)+c_{0i})-c_{00}&\text{ elsewhere};
    \end{cases} \tag{\ref{e:phiz0}} \\ 
    \phi(x_0) & =\phi(z_0)+c(x_0,z_0), \tag{\ref{e:phix0}} 
\end{align}
where the second case of definition (\ref{e:phiz0}) means that we pick an arbitrary value inside this interval.
\end{enumerate}
\item Let $\phi_{[x_0,z_0]}$ be equal to $\phi$.
    \item Return to step 1 with the next pair to process, or terminate if all pairs have been processed. 
\end{enumerate}

We note from step 3(b) that 
for $s,s'\in X_{[x_i,z_i]}\cup Z_{[x_i,z_i]}$,
we have $\phi(s)-\phi(s')=\phi_i(s)-\phi_i(s')$. This means that after each iteration, the value of $\phi(s)-\phi(s')$ does not change, and therefore it depends only on points in the region $[x_i,z_i]$. 

Our regional hierarchical mechanism allows for a class of dual solutions. For instance, there is freedom in the choice of the solution to the system of inequalities \eqref{eq:ineqs}, as well as the choices of $\phi(z_0)$ in the second case of \eqref{e:phiz0}. We remark that all dual solutions are obtained by choosing specific feasible solutions allowed by our regional hierarchical construction.\footnote{To observe this, first recall that for any dual solution $\phi_0$ and for any primal solution $\pi_0$, $\phi_0(x)-\phi_0(z) \leq c(x,z)$ holds with equality if the assignment $\pi_0$ pairs worker $x$ with job $z$. At each step of the regional mechanism, the only constraints on the choices of the dual solution arise from inequalities $\phi(x)-\phi(z)\leq c(x,z)$ and the equalities $\phi(x)-\phi(z)=c(x,z)$ for $(x,z)\in\Gamma_{\pi_0}$. The values of $\phi_0$ satisfy these constraints at each step, and therefore constitute as a valid dual construction. Note that we specify the relative wages in our algorithm, and hence a global constant shift suffices to recover the function $\phi_0$.} That is, there is no dual solution that falls outside this class. In our numerical analysis with large number of worker and job types, the class of dual solutions is such that the differences between the dual solutions are economically insignificant.  

The order of processing the pairs does not affect the output of the algorithm because each $\phi_{[x_0,z_0]}$ only depends on the local dual optimizers of its subpairs, which are all processed before this pair. A default order is to always choose the unprocessed pair $(x_0,z_0)$ with the smallest $x_0$ satisfying the condition in step 1. On the other hand, the choice of $(\beta_2,\dots,\beta_p)$ does affect the output of the algorithm. As a default,  $(\beta_2,\dots,\beta_p)$
can be chosen as the solution of \eqref{eq:ineqs} that is the smallest in dictionary order.\footnote{Note that such a smallest solution always exists since $(\beta_2,\dots,\beta_p)$ satisfying \eqref{eq:ineqs} lies in a compact region.}
 In this way, we obtain a unique output of the algorithm. Nevertheless, in the next result, we will show that a dual potential $\phi$ is obtained from the algorithm with arbitrary choices of $(\beta_2,\dots,\beta_p)$ satisfying (\ref{eq:ineqs}) in each iteration.

\subsection{ Proof of Theorem \ref{thm:dualalg}} \label{pthm:dualalg}

We prove Theorem \ref{thm:dualalg} in two parts. First, we prove there exists a solution $(\beta_2,\dots,\beta_p)$ to \eqref{eq:ineqs}, in Lemma \ref{lemma:dualalgorithm}. Second, we prove that the function $\phi$ defined in \eqref{eq:phidef}-\eqref{e:phix0} is indeed a local dual optimizer on $S_{[x_0,z_0]}$.


\begin{lemma}\label{lemma:dualalgorithm}
       Suppose $(x_1,z_1),\dots,(x_p,z_p)$ are ordered subpairs of pair $(x_0,z_0)$ in the optimal assignment $\pi$. Then the system of inequalities, where for all $1\leq n<m\leq p$:
    \begin{align}
        \tag{\ref{eq:ineqs}} \max(c_{00}+c_{nn}-c_{0n}-c_{m0},c_{nn}-c_{mn})\leq \sum_{k=n+1}^m\beta_k\leq \min(c_{0m}+c_{n0}&-c_{00}-c_{mm},c_{nm}-c_{mm})
    \end{align}
    admits a solution $(\beta_2,\dots,\beta_p)$.
\end{lemma}

\begin{proof}[Proof of Lemma \ref{lemma:dualalgorithm}] 
We use Farkas' Lemma to prove this existence result. We state Farkas' Lemma for completeness.

\begin{lemma} \label{Lemma:farkas}
    Let $A$ be a $d_1\times d_2$ matrix, $b\in \R^{d_1}$, and let $x=(x_1,\dots,x_{d_2})^\top$ be a set of real-valued variables. Then the system $Ax\geq b$ allows a set of solutions if and only if for any $y\in [0,\infty)^{d_1}$ such that $y^\top A=0$, it holds $y^\top b\leq 0$. 
\end{lemma}

\vspace{0.25 cm}
\noindent We aim to show that equation \eqref{eq:ineqs} admits a solution $(\beta_2,\dots,\beta_p)$. We observe that we can think of \eqref{eq:ineqs} equivalently as the following set of inequalities: 
\begin{align*}
& \sum_{k=n+1}^m\beta_k \geq c_{00}+c_{nn}-c_{0n}-c_{m0} \\
& \sum_{k=n+1}^m\beta_k \geq c_{nn}-c_{mn} \\
- & \sum_{k=n+1}^m\beta_k \geq c_{mm} - c_{nm} \\
- & \sum_{k=n+1}^m\beta_k \geq c_{00}+c_{mm}-c_{0m}-c_{n0}  
\end{align*}
for all $1 \leq n < m  \leq p$. All inequalities implied by \eqref{eq:ineqs} are linear in the variables $(\beta_2,\dots,\beta_p)$. Matrix $A$ is given by columns with values $(-1,0,+1)$, while vector $b$ is governed by the costs $c$. 

By Lemma \ref{Lemma:farkas} it suffices to prove the following.\footnote{Equation (\ref{eq:4termseq}) is the analog of $y^\top A=0$ in the statement of Farkas' Lemma. Specifically, we use $y^\top A=0$ if and only if $y^\top A z=0$ for all $z \in \mathbb{R}^{d_1}$. Applied to our setting, where $\beta$ takes the position of $x$ in Farkas' Lemma, this states that the weighted sum of all left-hand sides in the system of inequalities equals zero. Equation (\ref{eq:4termsineq}) below is similarly the analog of $y^\top b\leq 0$ in the statement of Farkas' Lemma.} For any set of non-negative weights $(\lambda^+_{mn},\lambda^-_{mn},\omega^+_{mn},\omega^-_{mn}),$ ${1\leq n<m\leq p}$ on each of the inequalities above such that 
\begin{align}\label{eq:4termseq}
    \sum_{1\leq n<m\leq p}(\lambda^+_{mn}+\omega^+_{mn})\sum_{k=n+1}^m\beta_k= \sum_{1\leq n<m\leq p}(\lambda^-_{m,n}+\omega^-_{m,n})\sum_{k=n+1}^m\beta_k, \text{ for all }(\beta_2,\dots,\beta_p),
\end{align}
it holds that
\begin{align}
    &\sum_{1\leq n<m\leq p} \Big( \lm_{mn}(c_{00}+c_{nn}-c_{0n}-c_{m0})+\om_{mn}(c_{nn}-c_{mn}) \Big) \nonumber\\
    &\hspace{4 cm}\leq \sum_{1\leq n<m\leq p} \Big( \lp_{mn}(c_{0m}+c_{n0}-c_{00}-c_{mm})+\op_{mn}(c_{nm}-c_{mm}) \Big) .\label{eq:4termsineq}
\end{align}

We start by simplifying equations \eqref{eq:4termseq} and \eqref{eq:4termsineq}. We first simplify equation \eqref{eq:4termseq}. Since \eqref{eq:4termseq} has to hold for all $(\beta_2,\dots,\beta_p)$, we note that the coefficient on each $\beta_k$ has to equal zero. For each $2\leq k\leq p$, equating the coefficients for $\beta_k$ requires
\begin{align}\label{eq:4eqform2}
    \sum_{m,n} \big( \lambda^+_{mn}+\omega^+_{mn} \big) =\sum_{m,n} \big( \lambda^-_{mn}+\omega^-_{mn} \big),
\end{align}
where we sum over all $(m,n)$ satisfying $1\leq n<k\leq m\leq p$, that is, we sum over all equations where $\beta_k$ appears. Furthermore, subtracting equation \eqref{eq:4eqform2} evaluated at $k$ from equation \eqref{eq:4eqform2} evaluated at $k+1$ yields:
\begin{align}\label{eq:4eqform3}
    \sum_{k<m\leq p}(\lp_{mk} +\op_{mk}) - \sum_{k<m\leq p}(\lm_{mk}+\om_{mk}) = \sum_{1\leq n<k}(\lp_{kn}+\op_{kn}) - \sum_{1\leq n<k}(\lm_{kn}+\om_{kn}),
\end{align}
 for all $2 \leq k<p$. 

We next simplify \eqref{eq:4termsineq}. Rearranging \eqref{eq:4termsineq} by collecting terms by coefficients in front of $c_{ij}$ leads to the equivalent form:
\begin{align}
&\sum_{1\leq n<m\leq p}\big( \lambda^{-}_{mn} + \lambda^{+}_{mn} \big) c_{00} + \sum_{1<m\leq p} \big( \lambda^{-}_{m1}+\om_{m1} \big) c_{11} + \sum_{1\leq n<p} \big( \lambda^{+}_{pn}+\op_{pn} \big) c_{pp}\nonumber\\
&\hspace{2cm}+\sum_{k=2}^{p-1} \Big( \sum_{k<m\leq p} \big( \lambda^{-}_{mk}+\om_{mk} \big) + \sum_{1\leq n<k} \big( \lambda^{+}_{kn} + \op_{kn} \big) \Big) c_{kk}\nonumber\\
&\leq \sum_{1<m\leq p}\lambda^{-}_{m1}c_{01}+\sum_{1\leq n<p} \lambda^{+}_{pn}c_{0p}+\sum_{1\leq n<p}\lambda^{-}_{pn}c_{p0}+\sum_{1<m \leq p}\lambda^{+}_{m1} c_{10} \nonumber\\
&\hspace{2cm}+\sum_{k=2}^{p-1} \Big( \sum_{1\leq n<k}\lambda^{+}_{kn}+\sum_{k<m\leq p}\lambda^{-}_{mk}\Big) c_{0k}+ \Big( \sum_{1\leq n<k}\lambda^{-}_{kn} + \sum_{k<m\leq p}\lambda^{+}_{mk}\Big) c_{k0}\nonumber\\
&\hspace{2cm}+\sum_{1\leq n<m\leq p}\om_{mn} c_{mn} +\sum_{1\leq n<m\leq p}\op_{mn}c_{nm},\label{eq:verylong}
\end{align}
where the left-hand side of the inequality collects all ``diagonal'' elements, and the right-hand side collects all other elements.

Our next step in proving that equation \eqref{eq:verylong} indeed holds, is to show that both sides of equation \eqref{eq:verylong} represent transport costs of an assignment between a measure of workers $\tilde{F}$ and a measure of jobs $\tilde{G}$. Specifically, consider the assignment problem between a measure $\tilde{F}$ and a measure $\tilde{G}$, satisfying:
\begin{align}
    \tilde{F}&=\sum_{1\leq n<m\leq p} \big(\lambda^{-}_{mn}+\lambda^{+}_{mn}\big)\delta_{x_0}+\sum_{1<m\leq p}\big(\lambda^{-}_{m1}+\om_{m1}\big)\delta_{x_1}+\sum_{1\leq n<p}\big(\lambda^{+}_{pn}+\op_{pn}\big)\delta_{x_p}\nonumber\\
    &\hspace{2cm}+\sum_{k=2}^{p-1}\Big(\sum_{k<m\leq p}\big(\lambda^{-}_{mk}+\om_{mk} \big)+\sum_{1\leq n<k}\big( \lambda^{+}_{kn}+\op_{kn} \big)\Big)\delta_{x_k},\label{eq:mutilde}
\end{align}
and, similarly,
\begin{align}
    \tilde{G} & = \sum_{1\leq n<m \leq p} \big( \lambda^{-}_{mn} + \lambda^{+}_{mn} \big) \delta_{z_0} + \sum_{1<m\leq p}\big( \lambda^{-}_{m1} + \om_{m1} \big) \delta_{z_1} + \sum_{1\leq n<p} \big( \lambda^{+}_{pn} + \op_{pn} \big) \delta_{z_p} \nonumber\\
    &\hspace{2cm}+ \sum_{k=2}^{p-1} \Big( \sum_{k<m\leq p} \big(\lambda^{-}_{mk}+\om_{mk}\big)+\sum_{1\leq n<k}\big( \lambda^{+}_{kn}+\op_{kn} \big) \Big) \delta_{z_k}.\label{eq:nutilde}
\end{align}
Both measures may not be probability measures, but they do have the same total mass. 

The fact that the left-hand side of \eqref{eq:verylong} represents a transport cost between workers $\tilde{F}$ and jobs $\tilde{G}$ is evident. Under this assignment, each worker type is assigned to an identically indexed job, which has an identical mass by construction of the worker distribution $\tilde{F}$ in \eqref{eq:mutilde} and the job distribution $\tilde{G}$ in \eqref{eq:nutilde}. To establish the same on the right-hand side requires work. Consider first the worker $x$ marginal on the right-hand side of \eqref{eq:verylong}.

\begin{enumerate}[noitemsep]
    \item The mass on $x_0$ is 
    \begin{align*}
        &\sum_{1<m\leq p}\lambda^{-}_{m1}+\sum_{1\leq n<p} \lambda^{+}_{pn}+\sum_{k=2}^{p-1} \Big(\sum_{1\leq n<k}\lambda^{+}_{kn}+\sum_{k<m\leq p}\lambda^{-}_{mk}\Big) =\sum_{1\leq n<m\leq p} \Big( \lambda^{-}_{mn}+\lambda^{+}_{mn} \Big).
    \end{align*}
    \item Using equation \eqref{eq:4eqform2} with $k=2$, the mass on $x_1$ is
    \begin{align*}
        \sum_{1<m\leq p}\lambda^{+}_{m1}+\sum_{1<m\leq p}\op_{m1}=\sum_{1<m\leq p} \big( \lambda^{-}_{m1}+\om_{m1} \big).
    \end{align*}
    \item For $2\leq k<p$, using \eqref{eq:4eqform3} and grouping terms, the mass on $x_k$ is 
    \begin{align*}
        \sum_{1\leq n<k}\lambda^{-}_{kn}+\sum_{k<m\leq p}\lambda^{+}_{mk}+\sum_{1\leq n<k} \om_{k,n}+\sum_{k<m\leq p}\op_{mk}=\sum_{k<m\leq p}(\lambda^{-}_{mk}+\om_{mk})+\sum_{1\leq n<k}(\lambda^{+}_{kn}+\op_{kn}).
    \end{align*}   
    \item Using equation \eqref{eq:4eqform2} with $k=p$, the mass on $x_p$ is 
    \begin{align*}
        \sum_{1\leq n<p} \big( \lambda^{-}_{pn}+\om_{pn} \big)= \sum_{1\leq n<p} \big( \lambda^{+}_{pn} + \op_{pn} \big).
    \end{align*}
\end{enumerate}
Combining these four terms we see that the $x$-marginal of the right-hand side of \eqref{eq:verylong} corresponds with that of \eqref{eq:mutilde}. We proceed to show that the same is true for the distribution of jobs.

\begin{enumerate}[noitemsep]
    \item The mass on $z_0$ is
    \begin{align*}
        \sum_{1\leq n<p}\lm_{pn}+\sum_{1<m\leq p}\lp_{m1}+\sum_{k=2}^{p-1} \Big( \sum_{1\leq n<k}\lambda^{-}_{kn}+\sum_{k<m\leq p}\lambda^{+}_{mk} \Big)=\sum_{1\leq n<m\leq p}\big( \lambda^{-}_{mn}+\lambda^{+}_{mn} \big),
    \end{align*}
    where the equality follows by simple accounting. 
    \item The mass on $z_1$ is
    \begin{align*}
        \sum_{1<m\leq p}\lm_{m1}+\sum_{1<m\leq p}\om_{m1}=\sum_{1<m\leq p} \big( \lm_{m1}+\om_{m1} \big).
    \end{align*}
    \item For $2\leq k<p$, the mass on $z_k$ is
    \begin{align*}
        \sum_{1\leq n<k}\lambda^{+}_{kn}+\sum_{k<m\leq p}\lambda^{-}_{mk}+\sum_{k<m\leq p}\om_{mk}+\sum_{1\leq n<k}\op_{kn}=\sum_{k<m\leq p} \big( \lambda^{-}_{mk}+\om_{mk} \big) +\sum_{1\leq n<k} \big( \lambda^{+}_{kn}+\op_{kn} \big).
    \end{align*}
    \item Finally, the mass on $z_p$ is given by
    \begin{align*}
        \sum_{1\leq n<p}\lp_{pn}+\sum_{1\leq n<p}\op_{pn}=\sum_{1\leq n<p} ( \lambda^{+}_{pn}+\op_{pn} ).
    \end{align*}
\end{enumerate}
We have thus proved the marginal distributions on both sides of the costs \eqref{eq:verylong} are the worker distribution $\tilde{F}$ and the job distribution $\tilde{G}$. 


Why is the left-hand side of equation \eqref{eq:verylong} the optimal transportation cost between the worker distribution $\tilde{F}$ and job distribution $\tilde{G}$? To characterize an optimal assignment between the constructed measures $\tilde{F}$ and $\tilde{G}$, we decompose the corresponding measure of underqualification $\tilde{H} := \tilde{F}-\tilde{G}$ into layers. By the definition of the worker measure $\tilde{F}$ in equation (\ref{eq:mutilde}) and the job measure $\tilde{G}$ in equation (\ref{eq:nutilde}), we know that for each $k$ we have $\tilde{F}(x_k)=\tilde{G}(z_k)$. This means each layer $\ell$ will consist of a subset $S\subseteq\{0,\dots,p\}$ and the distributions within the layer $F_\ell$ and $G_\ell$ will be uniform on $\{x_k\}_{k\in S}$ and $\{z_k\}_{k\in S}$ respectively. From the assumption of the theorem, we recall that the optimal assignment $\pi$ pairs $x_k$ with $z_k$ for every $k$ in the optimal assignment problem with uniform distributions on $\{x_k\}_{0\leq k\leq p}$ and $\{z_k\}_{0\leq k\leq p}$. Since a restriction of an optimal assignment is also optimal on the restricted marginals, we know that an optimal assignment between $F_\ell$ and $G_\ell$ matches $x_k$ to $z_k$ for each $k\in S$. After adding the layers, the same holds for an optimal assignment between $\tilde{F}$ and $\tilde{G}$ by the principle of layering. Therefore, the pairs $\{ (x_k,z_k) \}$ are paired under an optimal assignment between $\tilde{F}$ and $\tilde{G}$. This establishes the inequality \eqref{eq:verylong}, hence we finally conclude \eqref{eq:ineqs} has a solution. \end{proof}

\noindent Next, we continue to prove the second part of the result, that the function $\phi$ defined in \eqref{eq:phidef}-\eqref{e:phix0} is indeed a local dual optimizer on $S_{[x_0,z_0]}$. 

First, we record the following observation for our construction of $\phi$.

\begin{lemma}\label{lemma:concave}
    Suppose $h:[0,\infty)\to\R$ is concave. Then for $0\leq x\leq y$ and $a>0$ we have
    \begin{align*}
        h(x+a)-h(x)\geq h(y+a)-h(y).
    \end{align*}
\end{lemma}

\begin{proof}
    From concavity, it follows that
    \begin{align*}
        h(y)+h(x+a)\geq \frac{(y-x)h(y+a)+ah(x)}{y+a-x}+\frac{(y-x)h(x)+ah(y+a)}{y+a-x}=h(x)+h(y+a),
    \end{align*}
    completing the proof.
\end{proof}

\noindent We also use the following lemma to prove that \eqref{e:phiz0} and \eqref{e:phix0} define a dual optimizer.
\begin{lemma}\label{lemma:maxmin}
  Suppose that $\phi$ is a dual potential on $I_{(x_0,z_0)}$, where $\{(x_i,z_i)\}_{i\in \{1,\dots,p\}}$ are subpairs of the pair $(x_0,z_0)$.  It holds that
    \begin{align}
        \max_{i\in \{1,\dots,p\} } (\phi(x_i)-c(x_i,z_0))=\max_{x\in I_{[x_0,z_0]}\setminus\{x_0\}} (\phi(x)-c(x,z_0))\label{eq:max}
    \end{align}
    and 
    \begin{align}
        \min_{i\in \{1,\dots,p\} }(\phi(z_i)+c(x_0,z_i))=\min_{z\in J_{[x_0,z_0]}\setminus\{z_0\}}(\phi(z)+c(x_0,z)).\label{eq:min}
    \end{align}
\end{lemma}

\begin{proof}
    To prove \eqref{eq:max}, it suffices to show that for all $i\in \{1,\dots,p\}$ and $x\in I_{(x_i,z_i)}$, 
    \begin{align}
        \phi(x_i)-c(x_i,z_0)\geq \phi(x)-c(x,z_0).\label{eq:toshowmax}
    \end{align}
    Using properties of the dual potential $\phi$, we have
    \begin{align*}
       \phi(x_i)-c(x_i,z_0)-\phi(x)+c(x,z_0)&=\phi(z_i)+c(x_i,z_i)-c(x_i,z_0)-\phi(x)+c(x,z_0)\\
       &\geq c(x_i,z_i)-c(x_i,z_0)-c(x,z_i)+c(x,z_0).
    \end{align*}
    Since both $c(x_i,z_i)+c(x,z_0)$ and $c(x_i,z_0)+c(x,z_i)$ describe the matching costs between $\{x_i,x\}$ and $\{z_0,z_i\}$ where the arcs $(x_i,z_i)$ and $(x,z_0)$ intersect, we must have by Lemma \ref{l:nocross} that
    \begin{align*}
         c(x_i,z_i)-c(x_i,z_0)-c(x,z_i)+c(x,z_0)\geq 0.
    \end{align*}
    This proves \eqref{eq:max}. The other claim \eqref{eq:min} is similar. Again it suffices to prove 
     \begin{align}
        \phi(z_i)+c(x_0,z_i)\leq \phi(z)+c(x_0,z).\label{eq:toshowmin}
    \end{align}
    By  properties of the dual potential $\phi$, we have
    \begin{align*}
        \phi(z_i)+c(x_0,z_i)-\phi(z)-c(x_0,z)&=\phi(x_i)-c(x_i,z_i)+c(x_0,z_i)-\phi(z)-c(x_0,z)\\
        &\leq c(x_i,z)-c(x_i,z_i)+c(x_0,z_i)-c(x_0,z)\leq 0.
    \end{align*}
    This proves \eqref{eq:toshowmin}.
\end{proof}

\begin{proof}[Proof of Theorem \ref{thm:dualalg}] 
    Suppose that $(x_1,z_1),\dots,(x_p,z_p)$ are ordered subpairs of pair $(x_0,z_0)$ in the optimal assignment $\pi$, and that $\phi_i$ are  dual potentials on $S_{[x_i,z_i]}:=I_{[x_i,z_i]}\cup J_{[x_i,z_i]}$ for all $1\leq i\leq p$. We first prove that, with the possibilities of multiple workers on the same skill level and multiple jobs on the same difficulty level, our  $\phi$ in \eqref{eq:phidef}-\eqref{e:phix0} is well-defined.\footnote{Note that the cases where a worker has the same skill level as the difficulty level of a job has been excluded, as a consequence of Lemma \ref{l:common}.} 
    
    Suppose that $x_n=x_{n+1}$ or $z_n=z_{n+1}$ for some $1\leq n<p$. Then any solution $(\beta_2,\dots,\beta_p)$ to the system of inequalities \eqref{eq:ineqs} satisfies
        \begin{equation*}
            \max(c_{00}-c_{0n}-c_{n+1,0},-c_{n+1,n})+c_{nn}\leq\beta_{n+1}\leq \min(c_{0,n+1}+c_{n0}-c_{00},c_{n,n+1})-c_{n+1,n+1},
        \end{equation*}
        where the inequality follows from \eqref{eq:ineqs} with $m=n+1$. As a consequence,
        \begin{equation*}
            c_{nn}-c_{n+1,n}\leq\beta_{n+1}\leq c_{n,n+1}-c_{n+1,n+1},
        \end{equation*} and hence we must have 
        $\beta_{n+1}=c_{nn}-c_{n+1,n+1}$. In particular, the $\phi$ defined in \eqref{eq:phidef} satisfies $\phi(x_n)=\phi(x_{n+1})$ in the case $x_n=x_{n+1}$, and $\phi(z_n)=\phi(z_{n+1})$ in the case $z_n=z_{n+1}$.
                
 We next prove that $\phi$ is a local dual optimizer on the domain $S_{[x_0,z_0]}\setminus\{x_0,z_0\}$. Suppose that $x\in I_{[x_n,z_n]}$ and $z\in J_{[x_m,z_m]}$. The equality $\phi(x)-\phi(z)=c(x,z)$ when $(x,z)$ is a worker-job pair immediately follows because the same condition is satisfied by $\phi_i$ for all $1\leq i\leq p$. Our goal is to prove $\phi(x)-\phi(z)\leq c(x,z)$ when worker $x$ and job $z$ are not paired.

We consider three cases.

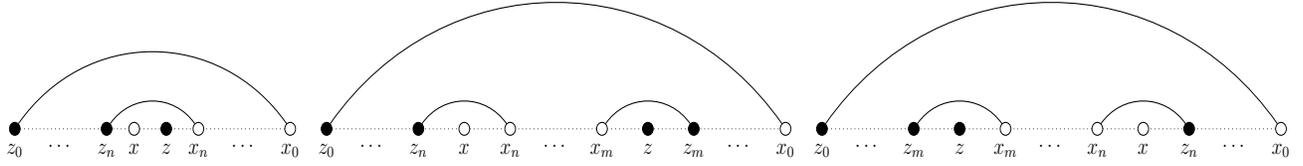
\begin{figure}[!t]
    \begin{center}\resizebox{17.4cm}{2.7cm}{
        \begin{tikzpicture}
\node[circle,fill=black,draw, minimum size=0.1pt,scale=0.6,label=below:{$z_0$}] (A) at  (6,0){};

\node[label=below:{$\dots$}] (D) at  (7,0) {} ;
\node[circle,draw,fill=black, minimum size=0.1pt,scale=0.6,label=below:{$z_{n}$}] (E) at  (8,0){};
\node[circle,draw,minimum size=0.1pt,scale=0.6,label=below:{$x$}] (x) at  (8.6,0){};
\node[circle,draw,fill=black,minimum size=0.1pt,scale=0.6,label=below:{$z$}] (z) at  (9.3,0){};

\node[circle,draw, minimum size=0.1pt,scale=0.6,label=below:{$x_{n}$}] (I) at  (10,0){};
\node[label=below:{$\dots$}] (K2) at  (11,0){};

\node[circle,draw, minimum size=0.1cm,scale=0.6,label=below:{$x_0$}] (L) at  (12,0) {} ;

\draw[dotted] (A) -- (x);
\draw[dotted] (x) -- (I);
\draw[dotted] (L) -- (I);
\path[-,every node/.style={font=\sffamily\small}] (L) edge[bend right=50] node [left] {} (A);

\path[-,every node/.style={font=\sffamily\small}] (I) edge[bend right=50] node [left] {} (E);
        \end{tikzpicture}
        
        \begin{tikzpicture}
\node[circle,fill=black,draw, minimum size=0.1pt,scale=0.6,label=below:{$z_0$}] (A) at  (6,0){};

\node[label=below:{$\dots$}] (D) at  (7,0) {} ;
\node[label=below:{$\dots$}] (d) at  (11,0) {} ;
\node[circle,draw,fill=black, minimum size=0.1pt,scale=0.6,label=below:{$z_{n}$}] (E) at  (8,0){};
\node[circle,draw,minimum size=0.1pt,scale=0.6,label=below:{$x$}] (x) at  (9,0){};

\node[circle,draw, minimum size=0.1cm,scale=0.6,label=below:{$x_m$}] (H) at  (12,0) {} ;\node[circle,fill=black,draw, minimum size=0.1cm,scale=0.6,label=below:{$z$}] (H2) at  (13,0) {} ;
\node[circle,fill=black,draw, minimum size=0.1cm,scale=0.6,label=below:{$z_m$}] (H2) at  (14,0) {} ;
\node[circle,draw, minimum size=0.1pt,scale=0.6,label=below:{$x_{n}$}] (I) at  (10,0){};
\node[label=below:{$\dots$}] (K2) at  (15,0){};

\node[circle,draw, minimum size=0.1cm,scale=0.6,label=below:{$x_0$}] (L) at  (16,0) {} ;

\draw[dotted] (A) -- (x);
\draw[dotted] (x) -- (I);
\draw[dotted] (H) -- (I);
\draw[dotted] (H) -- (L);
\path[-,every node/.style={font=\sffamily\small}] (L) edge[bend right=50] node [left] {} (A);

\path[-,every node/.style={font=\sffamily\small}] (I) edge[bend right=50] node [left] {} (E);
\path[-,every node/.style={font=\sffamily\small}] (H2) edge[bend right=50] node [left] {} (H);
        \end{tikzpicture}

        \begin{tikzpicture}
\node[circle,fill=black,draw, minimum size=0.1pt,scale=0.6,label=below:{$z_0$}] (A) at  (6,0){};

\node[label=below:{$\dots$}] (D) at  (7,0) {} ;
\node[label=below:{$\dots$}] (d) at  (11,0) {} ;
\node[circle,draw,fill=black, minimum size=0.1pt,scale=0.6,label=below:{$z_{m}$}] (E) at  (8,0){};
\node[circle,draw,fill=black,minimum size=0.1pt,scale=0.6,label=below:{$z$}] (x) at  (9,0){};

\node[circle,draw, minimum size=0.1cm,scale=0.6,label=below:{$x_n$}] (H) at  (12,0) {} ;\node[circle,draw, minimum size=0.1cm,scale=0.6,label=below:{$x$}] (H2) at  (13,0) {} ;
\node[circle,fill=black,draw, minimum size=0.1cm,scale=0.6,label=below:{$z_n$}] (H3) at  (14,0) {} ;
\node[circle,draw, minimum size=0.1pt,scale=0.6,label=below:{$x_{m}$}] (I) at  (10,0){};
\node[label=below:{$\dots$}] (K2) at  (15,0){};

\node[circle,draw, minimum size=0.1cm,scale=0.6,label=below:{$x_0$}] (L) at  (16,0) {} ;

\draw[dotted] (A) -- (x);
\draw[dotted] (x) -- (I);
\draw[dotted] (H) -- (I);
\draw[dotted] (H) -- (H2);
\draw[dotted] (L) -- (H2);
\path[-,every node/.style={font=\sffamily\small}] (L) edge[bend right=50] node [left] {} (A);

\path[-,every node/.style={font=\sffamily\small}] (I) edge[bend right=50] node [left] {} (E);
\path[-,every node/.style={font=\sffamily\small}] (H3) edge[bend right=50] node [left] {} (H);
        \end{tikzpicture}
        }
    \end{center}
\vspace{-0.35 cm}
    
    \caption{Three Cases for the proof of Theorem \ref{thm:dualalg} } \label{f:xnxm}
    {\scriptsize \vspace{.2 cm} Figure \ref{f:xnxm} illustrates the three different cases that we consider in the proof of Theorem \ref{thm:dualalg}. The first case is when $n=m$, the second case is when $n<m$, and the third case is when $m<n$. In each case, worker $x \in (x_n,z_n)$ and job $z \in (x_m,z_m)$.}
\end{figure}

\emph{Case I: $n=m$}.  This follows immediately since $\phi_n$ is a dual potential on $S_{[x_n,z_n]}$.

\emph{Case II: $n<m$}. Observe
\begin{align*}
    \phi(x)-\phi(z)& = \big( \phi(x)-\phi(z_n) \big)-\sum_{k=n}^{m}(\phi(x_k)-\phi(z_k))+\sum_{k=n+1}^{m}(\phi(x_{k-1})-\phi(z_k))+(\phi(x_m)-\phi(z))\\
    &\leq c(x,z_n)-\sum_{k=n}^{m}c_{kk}+\sum_{k=n+1}^{m}(\beta_k+c_{kk})+c(x_m,z) \\
    &= c(x,z_n)-c_{nn}+c(x_m,z)+\sum_{k=n+1}^{m}\beta_k\\
    &\leq c(x,z_n)-c_{nn}+c_{nm}-c_{mm}+c(x_m,z).
\end{align*}
The first inequality follows since both $\phi(x) - \phi(z_n) \leq c(x,z_n)$ and $\phi(x_m) - \phi(z) \leq c(x_m,z)$ follow from the dual potential within the same part, and $\phi(x_k)-\phi(z_k) = c_{kk}$ follows by the dual potential within the same part for paired workers and jobs, and finally $\beta_k+c_{kk} = \phi(x_{k-1})-\phi(z_k)$ by equation (\ref{eq:phidef}). The final inequality is implied by the upper bound on $\sum_{k=n+1}^{m}\beta_k$ from (\ref{eq:ineqs}). To show $\phi(x)-\phi(z)\leq c(x,z),$ it suffices to prove
\begin{align}\label{eq:concavity}
    c(x,z_n)+c(x_m,z)-c(x,z)\leq c_{nn}+c_{mm}-c_{nm}.
\end{align} 
Note that $c(x,z_n)\leq c(x_n,z_n)$ and $c(x_m,z)\leq c(x_m,z_m)$ since $x\in I_{[x_n,z_n]}$ and $z\in J_{[x_m,z_m]}$. By symmetry, we may without loss of generality assume that $\max(x_n,z_n)\leq \min(x_m,z_m)$. There are four cases:
\begin{enumerate}[noitemsep]
    \item $z_n\leq x_n\leq z_m\leq x_m$. Then $c(x,z)\geq c(x_n,z_m)$ and the claim follows.
    \item $x_n\leq z_n\leq z_m\leq x_m$. In this case $z_n-x\leq z_n-x_n$. Applying Lemma \ref{lemma:concave} with $a=z-z_n$ yields $c(x,z_n)-c(x,z)\leq c(x_n,z_n)-c(x_n,z)\leq c(x_n,z_n)-c(x_n,z_m)$. Using that $c(x_m,z)\leq c(x_m,z_m)$ equation (\ref{eq:concavity}) follows.
    \item $z_n\leq x_n\leq x_m\leq z_m$. In this case $z-x_m\leq z_m-x_m$. Applying Lemma \ref{lemma:concave} with $a=x_m-x_n$ yields $c(x_m,z)-c(x,z) \leq c(x_m,z)-c(x_n,z)\leq c(x_m,z_m)-c(x_n,z_m)$, where the first inequality follows from $c(x_n,z) \leq c(x,z)$. Further using $c(x,z_n)\leq c(x_n,z_n)$, inequality (\ref{eq:concavity}) follows.
    \item $x_n\leq z_n\leq x_m\leq z_m$. In this case, the configuration between $x_n$ and $z_n$ is identical to Case 2, and the configuration between $x_m$ and $z_m$ is identical to Case 3. We apply Lemma \ref{lemma:concave} as in Case 2 and in Case 3. First, $c(x,z_n)-c(x,z)\leq c(x_n,z_n)-c(x_n,z)$. Second, $c(x_m,z)-c(x_n,z)\leq c(x_m,z_m)-c(x_n,z_m)$. Summing the two inequalities delivers inequality (\ref{eq:concavity}).
\end{enumerate}
This completes the proof of \eqref{eq:concavity} for Case II.

\emph{Case III: $m<n$}. Observe that
\begin{align*}
    \phi(x)-\phi(z)&=(\phi(x)-\phi(z_n))+\sum_{k=m+1}^{n-1}(\phi(x_k)-\phi(z_k))-\sum_{k=m+1}^{n}(\phi(x_{k-1})-\phi(z_k))+(\phi(x_m)-\phi(z))\\
    &\leq c(x,z_n)+\sum_{k=m+1}^{n-1}c_{kk}-\sum_{k=m+1}^{n}(\beta_k+c_{kk})+c(x_m,z) \\
    &= c(x,z_n)-c_{nn}+c(x_m,z)-\sum_{k=m+1}^{n}\beta_k\\
    &\leq c(x,z_n)-c_{nn}+c_{nm}-c_{mm}+c(x_m,z).
\end{align*}
The first inequality follows since both $\phi(x) - \phi(z_n) \leq c(x,z_n)$ and $\phi(x_m) - \phi(z) \leq c(x_m,z)$ follow from the dual potential within the same part, and $\phi(x_k)-\phi(z_k) = c_{kk}$ follows by the dual potential within the same part for paired workers and jobs, and finally $\beta_k+c_{kk} = \phi(x_{k-1})-\phi(z_k)$ by equation (\ref{eq:phidef}). The final inequality is implied by the lower bound on $\sum \beta_k$ from (\ref{eq:ineqs}). The rest follows similarly as in Case II.

\vspace{0.4 cm}
We next check that $\phi$ is indeed a dual potential on $S_{[x_0,z_0]}$. Again, due to the possibilities of multiple workers on the same skill level and multiple jobs on the same difficulty level, we first verify whether our $\phi$ in \eqref{eq:phidef}-\eqref{e:phix0} is well-defined.

    \begin{enumerate}[noitemsep]
\item Suppose that $z_0=z_p$ and $x_0\neq x_1$. Using \eqref{e:phiz0}, we obtain
        \begin{equation*}
            \phi(z_0)=\max_{i\in\{1,\dots,p\}}(\phi(x_i)-c_{i0})=\max_{i\in\{1,\dots,p\}}(\phi(x_i)-c_{ip})\geq \phi(x_p)-c_{pp}=\phi(z_p),
        \end{equation*}
        where the second equality follows since $z_0 = z_p$. To show the inverse inequality, let $i\in\{1,\dots,p\}$. By construction, $\phi(x_i)-c_{ip}\leq \phi(z_p)$. Since $i$ is arbitrary, the finishes the proof.
        \item Suppose that $x_0=x_1$ and $z_0\neq z_p$. Using \eqref{e:phiz0} and \eqref{e:phix0}, we have 
        \begin{equation*}
            \phi(x_0)=\min_{i\in\{1,\dots,p\}}(\phi(z_i)+c_{0i})=\min_{i\in\{1,\dots,p\}}(\phi(z_i)+c_{1i})=\phi(x_1),
        \end{equation*}
                where the second equality follows as $x_0 = x_1$, and where the last equality follows because $\phi(z_1)+c_{11}=\phi(x_1)$ and for any $i$, $\phi(z_i)+c_{1i}\geq \phi(x_1)$.
        \item Suppose that $x_0=x_1$ and $z_0=z_p$. That $\phi(x_0)=\phi(x_1)$ is the same as done in the previous case, and hence we only need to show that $\phi(z_0)=\phi(z_p)$. By \eqref{e:phiz0} and using $x_0=x_1$, 
        \begin{equation*}
            \phi(z_0)=\min_{i\in\{1,\dots,p\}}(\phi(z_i)+c_{0i})-c_{00}=\min_{i\in\{1,\dots,p\}}(\phi(z_i)+c_{1i})-c_{1p}\leq \phi(z_p),
        \end{equation*}
        where the final inequality follows as one of $i =p$. It then suffices to show that the last inequality can be reversed. Let $i\in\{1,\dots,p\}$. By \eqref{eq:ineqs}, we know that
        \begin{equation*}
            \sum_{k=i+1}^p\beta_k\geq c_{ii}+c_{00}-c_{0i}-c_{p0}=c_{ii}+c_{0p}-c_{0i}-c_{pp}.
        \end{equation*}
        Inserting into \eqref{eq:phidef} and using that $\phi_i(x_i)-\phi_i(z_i)=c_{ii}$, we obtain
        \begin{equation*}
            \phi(z_i)+c_{0i}=\sum_{k=i+1}^p\beta_k+\phi(x_p)+c_{0i}-c_{ii}\geq \phi(x_p)+c_{0p}-c_{pp}=\phi(z_p)+c_{0p}.
        \end{equation*}
        Since $x_0=x_1$, we know that 
        $\phi(z_i)+c_{1i}-c_{1p}\geq \phi(z_p)$. Since $i$ is arbitrary, this completes the proof.
        
        \end{enumerate}

\vspace{0.3 cm}

\noindent It then remains to check that $\phi$ is a dual potential on $S_{[x_0,z_0]}$, that is, to show that
\begin{enumerate}[noitemsep]
    \item $\phi(x_0)-\phi(z_0)=c(x_0,z_0)$;
    \item For $x\in I_{[x_0,z_0]}\setminus\{x_0\}$, $\phi(x)-\phi(z_0)\leq c(x,z_0)$;
    \item For $z\in J_{[x_0,z_0]}\setminus\{z_0\}$, $\phi(x_0)-\phi(z)\leq c(x_0,z)$.
\end{enumerate}That $\phi(x_0)-\phi(z_0)=c(x_0,z_0)$ follows from the definition in \eqref{e:phix0}. In other words, we need to show that
\begin{align*}
    \phi(z_0)\geq\max_{x\in I_{[x_0,z_0]}\setminus\{x_0\}} (\phi(x)-c(x,z_0))\hspace{1.5 cm}\text{and}\hspace{1.5 cm} \phi(x_0)\leq \min_{z\in J_{[x_0,z_0]}\setminus\{z_0\}}(\phi(z)+c(x_0,z)).
\end{align*}
Indeed, this follows immediately from \eqref{e:phiz0}, \eqref{e:phix0} and Lemma \ref{lemma:maxmin}. It remains to prove that the interval in the second case of \eqref{e:phiz0} is non-empty, or
\begin{align*}
    \max_{i\in \{1,\dots,p\} } (\phi(x_i)-c_{i0})\leq \min_{i\in \{1,\dots,p\} }(\phi(z_i)+c_{0i})-c_{00}.
\end{align*}
To see this, let $m,n\in\{1,\dots,p\}$ be arbitrary, and we show that
\begin{align}
    \phi(x_n)-\phi(z_m)\leq c_{0m}+c_{n0}-c_{00}.\label{eq:phixz00}
\end{align}
Recall that $\phi(x_{i-1})-\phi(z_{i})=\beta_i+c_{ii}$ for $2\leq i\leq p$ by equation (\ref{eq:phidef}). Again, we have the three cases of Figure \ref{f:xnxm} to consider to show that the sufficient condition (\ref{eq:phixz00}) is satisfied.

\emph{Case I: $n=m$}. In this case, \eqref{eq:phixz00} becomes 
$c_{nn}+c_{00}\leq c_{n0}+c_{0n}.$
This is immediate from optimality of the assignment.

\emph{Case II: $n<m$}. Similar to the other Case II above, evaluated at $x=x_n$ and $z=z_m$, we obtain
\begin{align}
    \phi(x_n)-\phi(z_m)&\leq  c(x_n,z_n)-c_{nn}+c(x_m,z_m)+\sum_{k=n+1}^{m}\beta_k\leq c_{0m}+c_{n0}-c_{00},\label{eq:same}
\end{align}as desired, 
where the second inequality follows from the upper bound in \eqref{eq:ineqs}.



\emph{Case III: $n>m$}. Similar to the other Case III above, evaluated at $x=x_n$ and $z=z_m$, we have
\begin{align}
    \phi(x_n)-\phi(z_m)&\leq  c(x_n,z_n)-c_{nn}+c(x_m,z_m)-\sum_{k=m+1}^{n}\beta_k\leq -c_{00}+c_{0m}+c_{n0}, \label{e:ineqfinal}
\end{align}as desired, 
where the last step follows from \eqref{eq:ineqs}. \end{proof}


\noindent Finally, we emphasize that our construction relies on the concavity of the mismatch cost function $c(x,z)$ in two respects. First, the property of no intersecting pairs is essential for our induction structure. Second, \eqref{eq:concavity} requires concavity as well.

\subsection{Proof of Theorem \ref{prop:hiereff}} \label{a:hiereff}

In this appendix we provide a proof of Theorem \ref{prop:hiereff}.

\vspace{0.35 cm}
\noindent There are at most $N_s$ layers in the measure of underqualification. Each of these layers corresponds to a collection $\mathcal I$ of pairs, each of which contains ordered subpairs $(x_1,z_1),\dots,(x_p,z_p)$ that induces a system of inequalities (\ref{eq:ineqs}) where the algorithm is performed, with complexity $O(p^4)$. Given layer $\ell$ of the measure of underqualification, the sum of the sizes of such systems (that is, the number of subpairs) is bounded from above by the total number of pairs in layers $\{ \ell-1,\ell,\ell+ 1\}$, and hence is at most $3(m+n-1)$ by Theorem \ref{t:numberpair}.

In other words, the total complexity to solve the systems is at most $O(\sum\limits_{i \in \mathcal{I}} p_i^4)$, where $\sum\limits_{i \in \mathcal{I}} p_i\leq 3(m+n-1)$, and hence is at most $O((m+n)^4)$. The number of such systems is bounded above by the number $N_s$ of layers. Therefore, the complexity of our dual algorithm is $O((m+n)^4 N_s)$.

\subsection{Proof of Theorem \ref{p:dual}}\label{s:pdual}

In this appendix, we prove Theorem \ref{p:dual}. We make use of Lemma \ref{l:triangle} and Lemma \ref{lemma:tildepsi}, which we prove first. 

\begin{lemma}{\textit{Triangle Inequality}.} \label{l:triangle}
    For all $x,y,z\in\R$, it holds that $c(x,y)+c(y,z)\geq c(x,z)$.
\end{lemma}

\begin{proof}
Our cost of mismatch takes the form $c(x,z)=h(z-x)$ where $h$ is strictly concave and increasing on $[0,\infty)$, strictly concave and decreasing on $(-\infty,0]$, satisfying $h(0)=0$. 

The case where $x=z$ is trivial, so we focus our attention on the case where $x \neq z$. By symmetry, we assume $x< z$ without loss of generality. If $z-y\geq z-x > 0$, then necessarily $c(y,z)\geq c(x,z)$ and hence $c(x,y)+c(y,z)\geq c(x,z).$ The same argument applies when $y - x \geq z-x > 0$. In the remaining scenario where both $(z-y)$ and $(y-x)$ are in $ [0,z-x)$, we have by concavity of $h$ that\footnote{The interval is bounded below by zero because when $z-y \geq z-x$ is not true, then $y-x>0$, and similarly when $y-x \geq z-x$ is not true then $z-y > 0$.} 
\begin{align*}
    c(x,y)+c(y,z)&=h(y-x)+h(z-y)\geq \frac{y-x}{z-x}\, h(z-x) + \frac{z-y}{z-x}\, h(z-x)=h(z-x)=c(x,z),
\end{align*}
where the inequality follows since $y - x = \frac{y-x}{z-x} \times (z-x) + \frac{z-y}{z-x} \times 0$. This completes the proof.    
\end{proof}

In the main text, we established the connection between the dual optimizers for the cost minimization problem and the dual optimizers for the output maximization problem (\Cref{s:wagesandfirmvalues}, discussion follows \Cref{thm:dualalg}). In this appendix, we exploit this connection to simplify the exposition to the proof of Theorem  \ref{p:dual}. Specifically, we use that we can equivalently characterize the dual functions $(\phi,\psi)$ for the overlapping segments of the worker and the job distribution, with the understanding that we can obtain wages and job values using $w(x) = \alpha(x) - \phi(x)$ and $v(z)=\theta(z) - \psi(z)$.

To formulate Theorem \ref{p:dual} in terms of the dual potentials for the mismatch cost minimization problem, we need to describe our sequential construction of the functions. For the interpretation of these objects, we refer the reader to Section \ref{s:wagesandfirmvalues}. We define sequentially the dual maps, analogous to our previous definitions \eqref{e:hatphi} and \eqref{e:psiz2}. Starting from $\tilde{\phi}=g-\tilde{w}$, where $\tilde{w}$ are the dual values for mismatched workers $x \in I$, let 
\begin{align}
\tilde{\psi}(z):=\min \limits_{x \in I} \big( c(x,z) - \tilde{\phi}(x) \big) \ \hspace{1 cm} \text{ and }\ \hspace{1 cm} \hat{\phi}(x):=\min_{z\in I\cup J} \big(c(x,z)-\tilde{\psi}(z) \big) , \label{e:hatphi2}
\end{align} 
where we recall that $J$ is the set of mismatched jobs. Moreover, let
\begin{align}
\hat{\psi}(z):=\min_{x \in I \cup J} \; \big( c(x,z) - \hat{\phi}(x) \big), \label{e:hatpsi2}
\end{align}
 $\phi(x)=\hat{\phi}(x)$ for $x\in I$,  $\psi(z)=\hat{\psi}(z)$ for $z\in J$, and set $\phi(x)=-\psi(x)$  for $x\in J$ and $\psi(z)=-\phi(z)$ for $z\in I$. Finally, we define for $x\in K$
\begin{align}
    \phi(x)=\min_{z\in I\cup J} \big( c(x,z)-\psi(z) \big) \label{eq:defK}
\end{align}
and $\psi(z)=-\phi(z)$ for $z\in K$. 
It is easy to check that with these definitions, $\phi=g-w$ and $\psi=h-v$, with $w,v$ given in Theorem \ref{p:dual}.

To prove the result, we first define $c$-conjugate functions and analyze some of their properties.

\vspace{0.2 cm}
\noindent For $\phi: I \rightarrow \R$, we define the $c$-conjugate function for all jobs $z \in J$ as
\begin{equation}
\phi^c(z) := \min_{x\in I} \; \big( c(x,z) - \phi(x) \big). \label{e:ctrans}
\end{equation}
Denote by $\bar{c}(z,x)=c(x,z)$ and we further define for $x\in I$
\begin{equation}
    \phi^{c\bar{c}}(x)=(\phi^c)^{\bar{c}}(x) = \min\limits_{z\in J} \; \big( \bar{c}(z,x) - \phi^c(z) \big) =\min\limits_{z\in J} \; \big({c}(x,z) - \phi^c(z) \big). \label{e:ctrans2}
\end{equation} 

\noindent Given these definitions, the following statements follow:\footnote{See Chapter 1 of  \citet{Santambrogio:2015} for further details.}
\begin{enumerate}[noitemsep]
    \item $\phi^{c\bar{c}}\geq \phi$. \\
    This follows since for each $x\in I$ and $z\in J$, $\phi(x)+\phi^c(z)\leq c(x,z)$ or $\phi(x) \leq c(x,z) - \phi^c(z)$ by the definition of the $c$-conjugate function. By taking the infimum $z \in J$ this leads to $\phi^{c\bar{c}}(x)\geq \phi(x)$ by the definition \eqref{e:ctrans2}. 
    \item If  $\phi=\psi^{\bar{c}}$ for some $\psi$, then $\phi=\phi^{c\bar{c}}$.\\ 
    First, we observe that $\phi=\psi^{\bar{c}}$ naturally implies $\phi^c=\psi^{\bar{c}c}$. To see this, note that $\phi^c=\psi^{\bar{c}c}\geq \psi$, which follows from the previous statement. This inequality, by uniformly decreasing from $\phi^c$ to $\psi$, together with the definition \eqref{e:ctrans2}, implies we uniformly increase the conjugate, or $\phi^{c\bar{c}}=(\phi^c)^{\bar{c}}\leq \psi^{\bar{c}}=\phi$. We establish $\phi=\phi^{c\bar{c}}$ by combining $\phi^{c\bar{c}}\leq \phi$ with Statement 1. 

    
    \item If $(\phi,\psi)$ is an optimal dual pair, then so is $(\phi,\phi^c)$. \\
    Suppose $(\phi,\psi)$ is a dual pair, then $\phi(x)+\psi(z)\leq c(x,z)$. It holds by the definition in (\ref{e:ctrans}) that $\phi(x)+\phi^c(z)\leq c(x,z)$ as well as $\phi^{c}(z)\geq \psi(z)$. Since $\phi^{c}(z)\geq \psi(z)$ and $(\phi,\phi^c)$ is a dual solution, it follows that if $(\phi,\psi)$ is a solution to the dual maximization problem, then so is $(\phi,\phi^c)$. 
\end{enumerate}
 
\begin{lemma}\label{lemma:tildepsi}
    $\hat{\phi}(x)+\tilde{\psi}(z)\leq c(x,z)$ for all workers $x$ and jobs $z$ such that $x,z \in I\cup J$, and equality holds for $(x,z)\in \Gamma_\pi$.
\end{lemma}

\begin{proof} That $\hat{\phi}(x)+\tilde{\psi}(z)\leq c(x,z)$ follows from definition of the dual optimizer for workers $\hat{\phi}(x)$ in (\ref{e:hatphi}). Next, we prove $\hat{\phi}(x)+\tilde{\psi}(z) = c(x,z)$ for all workers and jobs $(x,z) \in \Gamma_\pi$. 

To prove that $\hat{\phi}(x)+\tilde{\psi}(z) = c(x,z)$ for workers and jobs $(x,z) \in \Gamma_\pi$, we fix some pair $(x,z) \in \Gamma_\pi$. Since $(\tilde{\phi}, \tilde{\phi}^c)$ is a dual solution to the assignment problem between remaining workers and jobs, $\tilde{\phi}(x)+\tilde{\phi}^c(z) = c(x,z)$ for all $(x,z) \in \Gamma_\pi$. Given the definition of the dual value for jobs $z \in I \cup J$ in (\ref{e:hatphi2}) we obtain that $\tilde{\psi}(z) =\min\limits_{x \in I} \big( c(x,z) - \tilde{\phi}(x) \big) = \tilde{\phi}^c(z)$ and hence that $\tilde{\phi}(x)+\tilde{\psi}(z) = c(x,z)$ for all $(x,z) \in \Gamma_\pi$. To conclude the proof it remains to show that $\hat{\phi}(x)=\tilde{\phi}(x)$ for every worker $x\in I$. 

We next show that $\hat{\phi}(x)=\tilde{\phi}(x)$ for every worker $x\in I$. Since we replaced, without loss of generality, the dual potential $\tilde{\phi}$ by the $c$-transform $\tilde{\phi}^{c\bar{c}}$,\footnote{We can always improve on the original $\tilde{\phi}$ by doing a double $c$-conjugate transform by Statement 1 that is weakly better in terms of the dual maximization problem, see Remark 1.13 in \citet{Santambrogio:2015}.}  it follows from the definition of the $c$-transform that for all $x \in I$:
\begin{equation*}
\tilde{\phi}(x) = \tilde{\phi}^{c\bar{c}}(x) = \min_{z \in J} \big( c(x,z) - \tilde{\phi}^c(z) \big). \label{e:ctrans}
\end{equation*}
Further, since $z\in J$, by definition of the dual potential for jobs $\tilde{\psi}(z) =\min\limits_{x \in I} \big( c(x,z) - \tilde{\phi}(x) \big) = \tilde{\phi}^c(z)$, where the second equality follows from the definition of the $c$-transform. We substitute this relationship into the previous expression for $\tilde{\phi}(x)$ to write
\begin{equation*}
\tilde{\phi}(x) = \min_{z \in J} \big( c(x,z) - \tilde{\phi}^c(z) \big) = \min_{z \in J} \big( c(x,z) - \tilde{\psi}(z) \big).
\end{equation*}

We can use the definition of the dual optimizers (\ref{e:hatphi}) to write that for all workers $x \in I$: 
\begin{align*}
\hat{\phi}(x)=\min \big( \min_{z\in J} \big( c(x,z) - \tilde{\psi}(z)\big) ,\, \min_{z\in I} \big( c(x,z)-\tilde{\psi}(z) \big)\big) = \min \big( \tilde{\phi}(x),\, \min_{z\in I}\big( c(x,z)-\tilde{\psi}(z) \big) \big)
\end{align*}
where the first equality follows by splitting the sets in (\ref{e:ctrans}) and the second equality follows from the equation above. 

Hence, we want to show for $(x,z) \in I$ the infimum is attained by $\tilde{\phi}(x)$. We show $c(x,z) \geq \tilde{\psi}(z) + \tilde{\phi}(x)$. This follows since the dual optimizer for all jobs is defined $\tilde{\psi}(z):=\min \limits_{x \in I} \big( c(x,z) - \tilde{\phi}(x)\big)$
for all $z \in I\cup J$.
\end{proof}


\noindent Having established the two claims, we next prove Theorem \ref{p:dual}. 

\begin{proof}[Proof of Theorem \ref{p:dual}] The proof is divided in three parts. We first show the inequality holds on $I\cup J$. To do so, we consider four cases:

\begin{enumerate}[noitemsep]
\item If $x\in I$ and $z\in J$, we have
\begin{align*}
    \phi(x)+\psi(z)&=\phi(x)-\phi(z)=\hat{\phi}(x)-\max_{x\in I\cup J}(\hat{\phi}(x)-c(x,z)) \leq c(x,z)
\end{align*}
because in the final step we subtract the maximum, but a feasible deduction is $\hat{\phi}(x)-c(x,z)$.
\item If $x,z \in I$, then by Lemma \ref{l:triangle} we have
\begin{align*}
    \phi(x)+\psi(z)=\phi(x)-\phi(z)&=\hat{\phi}(x)-\hat{\phi}(z) =\min_{y\in I\cup J}(c(x,y)-\tilde{\psi}(y))- \hspace{-0.2 cm} \min_{y\in I\cup J}(c(z,y)-\tilde{\psi}(y))\leq c(x,z) 
\end{align*}
where the final equality follows by (\ref{e:hatphi2}). The concluding inequality is obtained as follows. Suppose $y_0$ attains the infimum for the second term, the same $y_0$ may not attain the infimum for the first term but is feasible, so the left-hand side will be lower than when both terms are evaluated at $y_0$, or $\min\limits_{y\in I\cup J}(c(x,y)-\tilde{\psi}(y))-\min\limits_{y\in I\cup J}(c(z,y)-\tilde{\psi}(y))\leq c(x,y_0) - c(z,y_0)$. We combine the right-hand side with the triangle inequality of Lemma \ref{l:triangle} to write $c(x,z) + c(z,y_0) \geq c(x,y_0)$ or $c(x,z) \geq c(x,y_0)- c(z,y_0)$ to obtain the inequality. 
\item If $x\in J$ and $z \in I$. 
\begin{align*}
    \phi(x) + \psi(z) & = \max_{y\in I\cup J} \big( \hat{\phi}(y)-c(y,x) \big) - \hat{\phi}(z).
\end{align*}
We next want to show that this expression is less than $c(x,z)$. This is equivalent to showing that $\hat{\phi}(y)-c(y,x) \leq c(x,z) + \hat{\phi}(z)$ for all $y \in I\cup J$. To establish this, fix $y$, and evaluate: 
\begin{align*}
    \hat{\phi}(y) - \hat{\phi}(z) & = \min_{w\in I\cup J} \; ( c(y,w)-\tilde{\psi}(w) ) - \min_{w\in I\cup J} \; ( c(z,w)-\tilde{\psi}(w) )
\end{align*}
where the  equality follows from the definition of $\hat{\phi}$ in equation (\ref{e:hatphi2}). Let $w_0$ be the value that attains the infimum in the second term on the right, which is also feasible for the first term so that $\hat{\phi}(y) - \hat{\phi}(z) \leq c(y,w_0) - c(z,w_0)$. To bound this further, we use the triangle inequality of Lemma \ref{l:triangle} twice to write $c(y,w_0)- c(z,w_0) \leq c(y,z)$ as well as $c(y,z) \leq c(y,x) + c(x,z)$. Using the triangle inequalities, we thus write $\hat{\phi}(y) - \hat{\phi}(z) \leq c(y,x) + c(x,z)$, which is what we wanted to show since $y$ is arbitrary.

\item If worker and job $x,z\in J$, use (\ref{e:hatpsi2}) to write
\begin{align*}
    \phi(x) + \psi(z) & = \max\limits_{y\in I\cup J}(\hat{\phi}(y)-c(y,x)) - \max\limits_{y\in I\cup J}(\hat{\phi}(y)-c(y,z)).
\end{align*}
To bound the right-hand side, let $y_0$ denote the value that attains the supremum in the first term, which is also feasible for the second term. Hence, the right-hand side is bounded above by $-c(y_0,x) + c(y_0,z)$. By the triangle inequality of Lemma \ref{l:triangle} it follows that $-c(y_0,x) + c(y_0,z) \leq c(x,z)$ and hence we have $\phi(x) + \psi(z) \leq c(x,z)$.
\end{enumerate}

\noindent The second part of the proof shows that the equality holds everywhere on $I\cup J$ with respect to the optimal assignment $\pi$. We distinguish two cases:
\begin{enumerate}[noitemsep]
    \item The worker is perfectly matched to their job, or $(x,z)\in\{(x,x):x\in\R\}$. Since the dual functions are defined as $\psi(x) = - \phi(x)$ we have $\phi(x)+\psi(x)=0$. As a result, $\phi(x)+\psi(x)=0=c(x,x)$, as the cost of mismatch is zero.
    \item The worker is mismatched in their job, or $(x,z)\in \Gamma_\pi$, implying worker $x\in I$ and job $z\in J$. Using definition (\ref{e:hatpsi2}), $\psi(z) = - \phi(z)$,
\begin{align*}
    \phi(x)+\psi(z)&=\hat{\phi}(x)-\max_{y\in I\cup J}(\tilde{\phi}(y)-c(y,z)).
\end{align*}
By Lemma \ref{lemma:tildepsi}, $\hat{\phi}(x)+\tilde{\psi}(z)\leq c(x,z)$ for all $(x,z)$. In particular, for a given job $z$, $\hat{\phi}(x)+\tilde{\psi}(z)\leq c(x,z)$ for all $x$, and $\min\limits_{x\in I\cup J} (c(x,z) -\hat{\phi}(x) ) \geq \tilde{\psi}(z)$, or, equivalently, $-\max\limits_{x\in I\cup J}\big(\hat{\phi}(x)-c(x,z)\big) \geq \tilde{\psi}(z)$. Combining this inequality with the previous expression, we obtain the inequality 
\begin{align*}
    \phi(x)+\psi(z)& \geq \hat{\phi}(x) + \tilde{\psi}(z) = c(x,z)
\end{align*}
where the final equality follows by Lemma \ref{lemma:tildepsi}. Since we have shown the opposite inequality above in the first case of the first part of this proof, we obtain that $\phi(x)+\psi(z)= c(x,z)$.
\end{enumerate}

\noindent In the third part of the proof we further establish that the dual inequality $\phi(x)+\psi(z)\leq c(x,z)$ holds when $x\in K$ or $z\in K$. There are three cases.
\begin{enumerate}[noitemsep]
    \item $x\not\in K,~z\in K$. For any $x'\not\in K$, we have $\phi(x)+\psi(x') = \phi(x) - \phi(x') \leq c(x,x')$ when $x\not\in K$ by the first part of this proof. Following the triangle inequality of Lemma \ref{l:triangle},
    $\phi(x)-\phi(x') \leq c(x,x')\leq c(x,z)+c(z,x')$, giving $\phi(x)- c(x,z) \leq c(z,x') + \phi(x') = c(z,x') - \psi(x')$. Taking infimum over $x'\in I\cup J$ gives $\phi(x)-c(x,z)\leq \phi(z)=-\psi(z)$ using the definition of $\phi$.
    
    \item $x\in K,~z\not\in K$. For any $x'\in I\cup J$, by the definition of the wage function (\ref{eq:defK}), we have that
    $\phi(x) = \min\limits_{x' \in I \cup J} \big(c(x,x') - \psi(x')\big)$, such that $\phi(x) - \phi(x') \leq c(x,x') \leq c(x,z)+c(z,x')$, where the final step follows by the triangle inequality of Lemma \ref{l:triangle}. Alternatively, we write $\phi(x) - c(x,z) \leq  c(z,x') - \psi(x')$. Taking infimum in $x'\in I\cup J$ gives $\phi(x)+\psi(z)\leq c(x,z)$ using the definition of $\phi$.
    \item $x,z\in K$. We want to establish 
    $\phi(x)+\psi(z)\leq c(x,z)$. Using the definitions of the dual potentials in \eqref{eq:defK}, 
\begin{equation*}
\phi(x)+\psi(z) = \min \limits_{x' \in I \cup J} \big( c(x,x') - \psi(x') \big) + \max \limits_{x' \in I \cup J} \big( - c(z,x') + \psi(x') \big) .
\end{equation*}
    Suppose the maximum in the second term is attained by the worker value $x_0$, and also evaluate the first term at $x_0$ where it may not attain the minimum, implying $\phi(x)+\psi(z) \leq c(x,x_0) - c(z,x_0)$. By the triangle inequality $c(x,x_0) - c(z,x_0) \leq c(x,z) $ and hence it indeed follows that $\phi(x)+\psi(z)\leq c(x,z)$.
\end{enumerate}
By observing that $\phi=-\psi$ on the set $K$, the equality $\phi(x)+\psi(z)=0= c(x,z)$ holds when $x,z\in K$ and $(x,z)\in\Gamma_\pi$. This completes the proof in view of Lemma \ref{lemma:dual}.
\end{proof}

\subsection{Proof of Theorem \ref{p:cs}} \label{a:cs}

We provide a proof of Theorem \ref{p:cs}. To do so, we first provide a formal definition of the concordance order, and then proceed to establish two intermediary results. 

\vspace{0.35 cm}
\noindent The distribution function $\pi$ is smaller in concordance order than $\hat{\pi}$, written $\pi \preceq \hat{\pi}$, if for any cutoff coordinate $(x_c,z_c)$ we have $\pi((-\infty,x_c]\times(-\infty,z_c])\leq \hat{\pi}((-\infty,x_c]\times(-\infty,z_c])$ and $\pi([x_c,\infty)\times[z_c,\infty))\leq \hat{\pi}([x_c,\infty)\times[z_c,\infty))$.\footnote{By definition, two measures are comparable in concordance order only when they have the same pair of marginal distributions, making the concordance order a natural tool to compare different assignments. For example, when $\pi$ is the negative sorting or $\hat{\pi}$ is the positive sorting, $\pi \preceq \hat{\pi}$. We observe that $\preceq$ is a partial order on the probability measures with fixed marginals.}

%
%
%
%
%

\begin{lemma}{\textit{Local Cyclical Monotonicity.}} \label{lemma:cyclical}
Suppose $\pi_\ell$ is an assignment between distributions $F_\ell$ and $G_\ell$ on a given layer satisfying the non-crossing property. Then $\pi_\ell$ is an optimal assignment if and only if the following holds:
\begin{enumerate}[noitemsep]
    \item For any arc $(x_0,z_0)$ in $\pi_\ell$ and subpairs $\{ (x_i,z_i) \}^p_{i = 1}$ of $(x_0,z_0)$, $\pi_\ell$ is optimal on the assignment problem with workers $\{ x_i \}_{i=0}^p$ and jobs $\{ z_i \}_{i=0}^p$;
    \item For exposed arcs $\{ (\tilde{x}_i,\tilde{z}_i) \}^{\tilde{p}}_{i = 1}$ in $\pi_\ell$, $\pi_\ell$ is optimal on the assignment problem with workers $\{ \tilde{x}_i \}_{i=0}^p$ and jobs $\{ \tilde{z}_i \}_{i=0}^p$.
\end{enumerate} 
\end{lemma}

\begin{proof}The ``only if" direction follows directly by cyclical monotonicty. The proof shows the ``if" direction, where we iteratively eliminate concealed pairs to construct an optimal assignment between $F_\ell$ and $G_\ell$, and show that such an assignment is precisely $\pi_\ell$.

At each step of the procedure, consider pairs of the assignment $\pi_\ell$ whose subpairs contain no further subpairs. These pairs are mutually disjoint since the assignment satisfies the non-crossing property. Take such a pair $(x_0,z_0)$ with subpairs $\{ (x_i,z_i) \}^p_{i = 1}$. The skill interval $\mathcal I_0$ formed by pairing $(x_0, z_0)$ consists precisely of the workers $x_i$ and jobs $z_i$, where $0\leq i\leq p$. By assumption, $\pi_\ell|_{\mathcal I_0}$ is an optimal assignment between these workers and jobs. Observe that all pairs $\{ (x_i,z_i) \}^p_{i = 1}$ are therefore concealed under $(x_0,z_0)$ in this assignment. By Lemma \ref{l:hidden} in the Technical Appendix, there exists an optimal sorting between $F_\ell$ and $G_\ell$ that contains the pairs $\{ (x_i,z_i) \}^p_{i = 1}$. Since $\pi_\ell$ will contain the pairs $\{ (x_i,z_i) \}^p_{i = 1}$, we remove those pairs from our consideration, that is, we replace $F_\ell$ by $F_\ell\setminus\{x_i\}^p_{i = 1}$ and $G_\ell$ by $G_\ell\setminus\{z_i\}^p_{i = 1}$. 

Since $\pi_\ell$ has finitely many arcs, we can continue this procedure until we have nothing left but exposed arcs $\{ (\tilde{x}_i,\tilde{z}_i) \}^{\tilde{p}}_{i = 1}$ in the assignment $\pi_\ell$. By our assumption, the exposed arcs are locally optimal in $\pi_\ell$, 
and hence we conclude $\pi_\ell$ is an optimal assignment between the distribution of workers $F_\ell$ and the distribution of jobs $G_\ell$.
\end{proof}

\begin{lemma}
    \label{l:positive}
    Suppose cost function $c(x,z)$ is of the concave form \eqref{eq:cxz} and for some increasing and convex $\kappa$, $\hat{c}(x,z)=\kappa(c(x,z))$ is also of the concave form \eqref{eq:cxz}. On a layer $\ell$, if positive sorting is optimal with cost $c$, then it is optimal with cost $\hat{c}$.
\end{lemma}

\begin{proof}
By contradiction, suppose positive sorting is not optimal for the cost function $\hat{c}$. Then there exists an optimal assignment for cost function $\hat{c}$ that contains negatively sorted pairs $(x,z)$ and $(x',z')$ where $x<x'$ and $z'<z$ and is such that
\begin{align}
    \hat{c}(x,z')+\hat{c}(x',z)> \hat{c}(x,z)+\hat{c}(x',z').\label{eq:4c'}
\end{align}

Without loss of generality, we consider $x<z$, so that by the non-crossing property, and since $(x,z)$ and $(x',z')$ are negatively sorted, we either have $x<z'<x'<z$, or we have $x<x'<z'<z$. When $x<x'<z'<z$, the non-crossing principle directly contradicts that positive sorting is optimal under costs $c$. 


Suppose $x<z'<x'<z$, and under the original cost function $c$, we have by assumption that
\begin{align}
    c(x,z')+c(x',z)\leq c(x,z)+c(x',z').\label{eq:4c}
\end{align} 
Observe $\max (c(x,z'),c(x',z))\leq c(x,z)$, and that by optimality of negative sorting in \eqref{eq:4c'},  $\hat{c}(x',z')\le \min(\hat{c}(x,z'),\hat{c}(x',z))$, and hence, ${c}(x',z')\leq \min({c}(x,z'),{c}(x',z))$ since the function $\kappa$ is increasing. By combining the two previous inequalities ${c}(x',z')\leq \min({c}(x,z'),{c}(x',z)) \leq \max (c(x,z'),c(x',z))\leq c(x,z)$.

To arrive at a contradiction, choose some weight $\lambda, \beta \in [0,1]$ to average the minimum and maximum cost such that:
\begin{align*}
c(x',z) & = \lambda c(x',z') + (1-\lambda) c(x,z); \\
c(x,z') & = \beta c(x',z') + (1-\beta) c(x,z).
\end{align*}
By adding these two equalities and comparing to inequality (\ref{eq:4c}), it has to be true that $\lambda + \beta \geq 1$. Next, we use the convexity of the function $\kappa$ to establish 
\begin{align*}
\kappa(c(x',z)) & = \kappa \big( \lambda c(x',z') + (1-\lambda) c(x,z) \big) \leq  \lambda \kappa(c(x',z')) + (1-\lambda) \kappa(c(x,z)); \\
\kappa(c(x,z')) & = \kappa \big( \beta c(x',z') + (1-\beta) c(x,z) \big) \leq  \beta \kappa(c(x',z')) + (1-\beta) \kappa(c(x,z)).
\end{align*}
By adding these two inequalities, we write:
\begin{align*}
\hat{c}(x',z) + \hat{c}(x,z') & \leq (\lambda + \beta) \hat{c}(x',z') + (2 - (\lambda + \beta)) \hat{c}(x,z) \leq \hat{c}(x',z') + \hat{c}(x,z)
\end{align*}
where the final equality follows as $\lambda + \beta \geq 1$ as well as $\hat{c}(x',z') \leq \hat{c}(x,z)$ because $c(x',z') \leq c(x,z)$. This contradicts \eqref{eq:4c'}, concluding the proof.  \end{proof}

\noindent Using the two Lemmas we prove Theorem \ref{p:cs}.

\begin{proof}[Proof of Theorem \ref{p:cs}]
Since the marginal distributions of workers $F$ and jobs $G$ are fixed, the optimal assignments $\pi$ and $\hat{\pi}$ have the same layering structure. Hence, we fix a layer $\ell$ and consider the assignment problem between workers $F_\ell$ and jobs $G_\ell$ in that layer. We treat $\pi_\ell$ and $\hat{\pi}_\ell$ as assignments for layer $\ell$.

If $\pi_\ell$ is an optimal assignment with cost $\hat{c}$, the result directly follows when we pick $\hat{\pi}_\ell=\pi_\ell$. Suppose instead $\pi_\ell$ is not optimal with cost $\hat{c}$. Then, by Lemma \ref{lemma:cyclical}, there are two possibilities:
\begin{enumerate}[noitemsep]
    \item There exists a pair $(x_0,z_0)$ of $\pi_\ell$ with subpairs $\{ (x_i,z_i) \}^p_{i = 1}$ such that the locally optimal assignment with workers $x_i$ and jobs $z_i$, where $0\leq i\leq p$, can be improved for cost $\hat{c}$.
    \item The locally optimal assignment on the set of exposed arcs $\{ (\tilde{x}_i,\tilde{z}_i) \}^{\tilde{p}}_{i = 1}$ can be improved for cost $\hat{c}$.
\end{enumerate} 
In both cases, the assignment that improves upon $\pi$ also satisfies the non-intersecting property and has a strictly smaller total cost of skill gaps. The improved assignment may still not be optimal. However, by iterating the procedure outlined in this paragraph, and since $F_\ell$ and $G_\ell$ are finite, we eventually reach an optimal assignment for the mismatch costs $\hat{c}$. In other words, step-by-step improvements on the local assignment problems give an optimal assignment. 

Before proceeding, we observe that the second case can be directly ruled out by \Cref{l:positive}. Because the set of exposed arcs $\{ (\tilde{x}_i,\tilde{z}_i) \}^{\tilde{p}}_{i = 1}$ is positively sorted, there is no way to improve this further since whenever positive sorting is optimal for the concave costs, positive sorting is also optimal for the less concave costs by \Cref{l:positive}. Hence, we focus on the first case in the remainder of this proof.

By the transitivity of the concordance order, it suffices to show that for each improvement step (on a pair and its subpairs), the concordance order is increased.

To show that for each improvement step, the concordance order is increased, fix a pair $(x_0,z_0)$ with subpairs $\{ (x_i,z_i) \}^p_{i = 1}$ as in the top panel of Figure \ref{f:cc'}, and assume they form a locally optimal assignment $\pi_c$ for the mismatch cost function $c$. Without loss of generality, take $x_0<z_0$. Since we are working on a single layer, we must have $x_0<z_1<x_1<\dots<z_p<x_p<z_0$.   

We claim that there exists some locally optimal assignment $\pi_{\hat{c}}$ for the more linear mismatch cost $\hat{c}$ with the following structure. For some $q\geq 1$, $\pi_{\hat{c}}$ consists of pairs $\{ (\tilde{x}_i,\tilde{z}_i) \}^q_{i=1}$ that are not contained in any pair, and for each $i$, optimal sorting on the interval $(\tilde{x}_i,\tilde{z}_i)$ is positive, that is, $\pi_{\hat{c}}|_{(\tilde{x}_i,\tilde{z}_i)}$ is the positive sorting. This configuration is shown in the bottom panel of Figure \ref{f:cc'}. When this claim is true, it is straightforward to verify that $\pi_c \preceq \pi_{\hat{c}}$. Indeed, the positive sorting patterns within each interval $(\tilde{x}_i,\tilde{z}_i)$, for all $1\leq i\leq q$ coincide under the assignments $\pi_c$ and $\pi_{\hat{c}}$, and the remaining part is positive sorting for $\pi_{\hat{c}}$, where we recall that the positive sorting has the largest concordance order among all assignments. 

To prove the above claim, suppose that in assignment $\pi_{\hat{c}}$, the worker $x_0$ is optimally paired to job $z_k$ for some $k$. Consider the sorting problem between workers and jobs on the skill interval $(x_0,z_k)$. Since $\pi_c$ is optimal with the cost of skill gaps $c$,  positive sorting is locally optimal on $(x_0,z_k)$ given the cost of skill gaps $c$. By Lemma \ref{l:positive}, positive sorting is also locally optimal on $(x_0,z_k)$ for the cost $\hat{c}$. Therefore, $\pi_{\hat{c}}$ consists of positive sorting on the region $(x_0,z_k)$. 

We continue this procedure to the right, that is, we start with worker $x_{k+1}$, and repeat the above argument. Continuing this procedure thus yields the desired structure for $\pi_{\hat{c}}$, meaning that within each exposed arc there is positive sorting.\end{proof}

\vspace{0.15 cm}
\noindent We show that the sufficient conditions of \citet{Anderson:2022} for sorting to be more positive as the output function changes do not apply in our setting in Technical \Cref{s:andersonsmith}.

\subsection{Proof of \Cref{t:tolinear}}\label{App:linearcost}

After maximizing perfect pairs, by Lemma \ref{l:common}, we can restrict attention to assignments between worker and job distributions that are supported on disjoint sets. This means that the distributions $F$ and $G$ are supported on a finite set $S$, and we denote by $\delta$ the smallest pairwise distance between elements in $S$, and by $D$ the largest pairwise distance between elements.\footnote{This proof can be extended to the case of continuous distributions when $F$ and $G$ are compactly supported with the measure of underqualification $H:=F-G$ satisfying that both $H$ and $-H$ have finitely many local maxima and those maxima are strictly above zero.}



We show there exists $0<\bar{\zeta}<1$ such that for any $\zeta_p, \zeta_u \in[\bar{\zeta},1]$, the layered positive assignment $\pi$ is optimal with respect to the mismatch cost $c(x,z)$.  To prove the result, consider $\bar{\zeta}$ such that for any pair $(\delta',D')\in\{(\delta',D')\mid \delta\leq \delta'\leq D'\leq D,~D'-\delta'\geq 2\delta\}$:
\begin{align}
2^{1-\bar{\zeta}}(D'-\delta')^{\bar{\zeta}} \leq D'^{\bar{\zeta}}.\label{eq:p0}
\end{align}
Such a $\bar{\zeta}$ exists because both sides of \eqref{eq:p0} are uniformly continuous in $\bar{\zeta}$ on $\{(\delta',D')\mid \delta\leq \delta'\leq D'\leq D,~D'-\delta'\geq 2\delta\}$ and ``$<$" holds uniformly when $\bar{\zeta}=1$. Consider $\zeta_p,\zeta_u \in [\bar{\zeta},1]$. It suffices to prove that the optimal assignment within a layer does not contain any nested arc for the mismatch cost 
\begin{equation}
c(x,z)=\begin{cases}
    B_p(z-x)^{\zeta_p} &\text{ if }z\geq x;\\
    B_u(x-z)^{\zeta_u} &\text{ if }z<x.
    \end{cases}
\end{equation}

By the principle of layering in Lemma \ref{l:layer}, we decompose both measures $F_n$ and $G_n$ into layers. Here we assume without loss of generality that the lowest skill worker comes before the lowest skill job: $x_1<z_1$. On each layer there are $2 k$ equal masses on the skill levels $x_1 < z_1 < \dots < x_k < z_k$ in $S$ that are at least $\delta$ apart. The maximum distance within this layer is exceeded by $D \geq z_k - x_1$. Let $\{x_j\}_{1\leq j\leq k}$ be the locations of mass on the layer for workers $F_n$, and let $\{z_j\}_{1\leq j\leq k}$ be the locations of mass on the layer for jobs $G_n$. By contradiction, suppose that the optimal assignment within this layer instead does contain a nested arc, so it holds for some $x_1\leq u<v<s<t\leq z_k$ $-$ where we assume $t$ and $v$ are jobs and $u$ and $s$ are workers $-$ that:
\begin{equation*}
B_p(t-u)^{\zeta_p} + B_u(s-v)^{\zeta_u} \leq B_p(t-s)^{\zeta_p} + B_p(v-u)^{\zeta_p}.
\end{equation*}
By concavity of the function $x\mapsto B_px^{\zeta_p}$ for $ x \geq 0$,
\begin{equation*}
B_p(t-s)^{\zeta_p} + B_p(v-u)^{\zeta_p}\leq 2B_p\left(\frac{1}{2}(t-s)+\frac{1}{2}(v-u)\right)^{\zeta_p}=2^{1-{\zeta_p}}B_p\big((t-u)-(s-v)\big)^{\zeta_p}.
\end{equation*}
Putting $\zeta=\zeta_p$, $t-u=D'$, and $s-v=\delta'$ in equation \eqref{eq:p0} leads to a contradiction.
Hence the optimal assignment within this layer does not contain any nested arc for $\zeta_p \in [\bar{\zeta},1]$. 

For the alternative case $-$ where $t$ and $v$ are workers and $u$ and $s$ are jobs $-$ the conclusion follows from the exact same steps, with the subscripts on $B$ and $\zeta$ interchanged in the previous paragraph. 

\vspace{0.2cm}
\noindent The implication of the proposition is that for mismatch power values close to one, the solution can be directly constructed by constructing the measure of
underqualification, and constructing the positive alternating assignment by layer.\footnote{\citet{Juillet:2020} calls the layered positive assignment an excursion coupling and shows  that the layered positive assignment is
the limit of some optimal couplings as $\zeta\to 1^-$. We complement their result by proving the existence of a threshold $\bar{\zeta}$ beyond which the layered positive assignment is optimal for our environment.}
While this assignment generates positive sorting within each
layer, we emphasize this does not imply positive sorting overall.


\newpage
\renewcommand{\theequation}{B.\arabic{equation}} \setcounter{equation}{0}
\renewcommand{\thefigure}{B.\arabic{figure}}\setcounter{figure}{0}
\renewcommand{\thetable}{B.\arabic{table}}\setcounter{table}{0}
\setcounter{page}{1}
\newpage

\begin{center}
{\Large Composite Sorting \\}
\bigskip
{\Large Technical Appendix \\}
\bigskip
{\large  Job Boerma, Aleh Tsyvinski, Ruodu Wang, and Zhenyuan Zhang \\}
\bigskip
{\large May 2025}
\end{center}
\vspace{0.8cm}

\section{Additional Results}

In this appendix, we present additional technical results.

\subsection{General Production Function} \label{s:generalprod}

In this section, we present two generalizations of the model.

\subsubsection{Concave Distance Function}

Using Legendre transformations, we show that the indirect output function is generally a strictly concave function in mismatch given strictly convex cost functions. To be specific, let $\Psi$ capture the cost function $\Psi_p$ or $\Psi_u$ in Section \ref{s:environment}, and use $d:=|x-z|$ to denote the distance to obtain 
\begin{equation}
\mathcal{C}(d) = \min_{\gamma \geq 0} \; \big( \gamma d + \Psi(\gamma) \big).
\end{equation}
This problem has a unique solution characterized by $d = - \Psi'(\gamma)$. From the envelope condition, we obtain $\mathcal{C}'(d)= \gamma > 0$, showing that the cost function is strictly increasing in the distance.

To characterize the second derivative, we write the cost minimization problem as a maximization problem of the form:
\begin{equation*}
\hat{\mathcal{C}}(d) = \max_{\gamma \geq 0} \; \big( - \gamma d - \Psi(\gamma) \big),
\end{equation*}
where $\hat{\mathcal{C}} = - \mathcal{C}$, which shows $\hat{\mathcal{C}}$ is the Legendre transformation of the strictly convex function $\Psi$. Since the Legendre transformation of a strictly convex function is also strictly convex, the indirect cost function $\mathcal{C}$ is a strictly concave function of the distance.  As a result, choosing an assignment to maximize:
\begin{equation*}
y(x,z) = z + x - \mathcal{C}(|x-z|)
\end{equation*}
where $\mathcal{C}$ is now our concave cost distance function.


\subsubsection{Asymmetric Distance Function}

Next, we incorporate differential distance functions for both $x - z > 0$ and $x - z < 0$. This is a trivial extension, let $\bar{\Psi}$ denote the cost function for $x - z > 0$ and $\underline{\Psi}$ denote the cost function for $x - z < 0$. In this case, the cost minimization problem is:
\begin{equation*}
\bar{\mathcal{C}}(d) = \min_{\gamma \geq 0} \; \big( \gamma d + \bar{\Psi}(\gamma) \big)
\end{equation*}
when $d > 0$. By the same arguments on the Legendre transformation, this gives rise to a strictly concave function of the distance $\bar{\mathcal{C}}(d)$. Analogously, when $d < 0$, we generically obtain a distinct strictly concave function of the distance $\underline{\mathcal{C}}(d)$. As a result, we choose an assignment to maximize:
\begin{equation*}
y(x,z) = z + x - \bar{\mathcal{C}}(\{x-z \}_+) - \underline{\mathcal{C}}(\{z-x \}_+).
\end{equation*}

\subsection{Uniqueness of Optimal Sorting} \label{s:uniquenesssorting}

In this appendix, we discuss the uniqueness of optimal sorting for the mismatch cost function (\ref{eq:cxz}).

\begin{proposition}
    For any fixed distributions of workers $F$ and of jobs $G$, the set of $(\zeta_p,\zeta_u)\in (0,1)^2$ where the optimal assignment is not unique has Lebesgue measure zero.
\end{proposition}

\begin{proof}
First, recall that every optimal assignment has the non-crossing property and the layering structure. In other words, the non-uniqueness problem arises only when we solve the assignment problems in each layer. Since layering does not depend on the cost function and there are finitely many layers, we consider without loss a fixed layer $\ell$ with $2n_\ell$ points, with an alternating pattern in the marginals $F_\ell$ and $G_\ell$.

By Birkhoff's theorem \citep{Birkhoff:1946}, every assignment between $F_\ell$ and $G_\ell$ is a mixture of bijective assignments. Therefore, it suffices to restrict to the set of bijective matchings between $F_\ell$ and $G_\ell$. Since $F_\ell$ and $G_\ell$ have finite support, there exist finitely many assignments, and hence it suffices to prove that for any two assignments $\pi$ and $\pi'$, their costs equal on a set  $(\zeta_p,\zeta_u)$ of measure zero. In the following we fix $\pi$ and $\hat{\pi}$. Their respective costs are of the form $\mathcal{C}(\pi)=\sum\limits_{j=1}^{n_\ell} \big( a^{\vphantom{\zeta}}_jd_j^{\zeta_p}+b^{\vphantom{\zeta}}_je_j^{\zeta_u}\big)$ and $\mathcal{C}(\hat{\pi})=\sum\limits_{j=1}^{n_\ell} \big( \hat{a}^{\vphantom{\zeta}}_j \hat{d}_j^{\zeta_p}+\hat{b}^{\vphantom{\zeta}}_j \hat{e}_j^{\zeta_u} \big)$ for some $d_j,e_j,\hat{d}_j,\hat{e}_j\geq 0$. Equating $\C(\pi)=\C(\hat{\pi})$ leads to an equation of the form 
\begin{equation}
     \sum_{j=1}^{2 n_\ell}a^{\vphantom{\zeta}}_j d_j^{\zeta_p}=\sum_{j=1}^{2 n_\ell}b^{\vphantom{\zeta}}_j e_j^{\zeta_u},\label{eq:solve for}
\end{equation}
    where $d_j,e_j\geq 0$. Note that both sides are constant zero only if $F_\ell=G_\ell$, which is not feasible.

We next establish that the set of solutions $(\zeta_p,\zeta_u)\in (0,1)^2$ to \eqref{eq:solve for} has zero Lebesgue measure. For a fixed value $\zeta_p$, the equation \eqref{eq:solve for} has finitely many solutions $\zeta_u\in(0,1)$ (see \citet{Tossavainen:2006}). Similarly, for a fixed value $\zeta_u$ it has finitely many solutions $\zeta_p\in(0,1)$. Clearly, the zero set to equation \eqref{eq:solve for} is a measurable set. Therefore, Fubini's theorem applied to the indicator function yields that such a set must have zero measure.   
\end{proof}

\subsection{Equilibrium} \label{s:equilibriumdefn}

We formally define an equilibrium for this economy.

\vspace{0.45 cm}
\noindent To define an equilibrium, we specify the firm problem and the worker problem. A firm with job $z$ employs a worker $x$ to maximize profits taking the wage schedule $w$ as given. The firm problem is: 
\begin{align}
v(z)=\max_{x \in X} \hspace{0.2 cm}  y(x,z)-w(x) .\label{e:firm_problem}
\end{align}
Taking firm compensation $v$ as given, worker $x$ chooses to work in occupation $z$ to maximize wage income: 
\begin{align}
w(x)=\max_{z \in Z} \hspace{0.2 cm} y(x,z)-v(z) .\label{e:worker_problem}
\end{align}
\textbf{Equilibrium}. An equilibrium is a wage function $w$, a firm value function $v$, and a feasible assignment $\pi$, such that firms solve their profit maximization problem (\ref{e:firm_problem}), workers solve the worker problem (\ref{e:worker_problem}), and a feasibility constraint is satisfied
\begin{equation}
\int y(x,z ) \,\text{d} \pi = \int w(x) \,\text{d} F + \int v(z) \,\text{d} G, 
\end{equation}
which states that the total quantity of output produced, $\int y(x,z )\, \text{d} \pi$, equals the total quantity of output distributed to workers and jobs.

\subsection{Layering } \label{p:layer}

We observe that each layer consists of an alternating configuration of workers
and jobs, that is, either $x_{1}<z_{1}<x_{2}<z_{2}<\dots< x_{n}<z_{n}$ (for layers above 0) or $z_{1}<x_{1}<z_{2}<x_{2} <\dots <z_{n}<x_{n}$ (for layers below 0). We define an alternating assignment problem as an assignment
problem between $n$ workers and $n$ jobs, where workers and jobs are arranged in increasing order, and alternating
such that every worker skill level is followed by a job difficulty level, except for the last one. Let $F_{\ell}$ and $G_{\ell}$ be the measures of the workers and the jobs in each layer.

After providing a decomposition into layers with alternating configurations,
Lemma \ref{l:layer} shows how to solve the full assignment problem using the solutions to the assignment problem
within each layer. To prove Lemma \ref{l:layer}, we make use of the following result due to \citet{Villani:2009}. We repeat the result here for completeness.

\begin{lemma}{\textit{Stability of Optimal Assignment}.}\label{lemma:stability} Let $c(x,z)$ be a continuous non-negative cost function, and $\{F_n\}_{n\in\N},\,\{G_n\}_{n\in\N}$ be sequences of distributions of workers and jobs, respectively. Suppose $F_n\to F,\,G_n\to G$ weakly for some $F,\,G$,\footnote{This means that $F_n\to F$ on continuity points of $F$ and respectively for $G$.} and let $\pi_n$ be an optimal assignment between $F_n$ and $G_n$. If $\pi_n\to \pi$ in distribution, then $\pi$ is an optimal assignment between $F$ and $G$.
\end{lemma}

Given the measures of workers and jobs in each layer, we next observe that the worker and job distributions are supported on disjoint sets and on a finite set of skills $\{s_j\}_{1 \leq j \leq S}$. We smooth both the discrete distribution of workers and the discrete distribution of jobs by replacing each atom in the worker and job distribution at level $s_j$ by a uniform distribution on $[s_j,s_j + \varepsilon]$ with the same mass for every $1\leq j\leq S$, where $\varepsilon$ is small enough such that the intervals $[s_j,s_j+\varepsilon]$ for all $1\leq j\leq S$ do not intersect. In Figure \ref{f:tilt}, we provide an illustration of this procedure. We denote the smoothed measure of workers by $F_\varepsilon$ and the smoothed measure of jobs by $G_\varepsilon$, and the corresponding underqualification measure by $H_\varepsilon$. An optimal assignment given worker measure $F_\varepsilon$ and job measure $G_\varepsilon$ is given by $\pi_\varepsilon$. Since the mismatch cost $c$ is continuous, by the stability of the optimal transport, $\pi_\varepsilon \to \pi$ weakly where $\pi$ is the optimal assignment between workers $F$ and jobs $G$. 

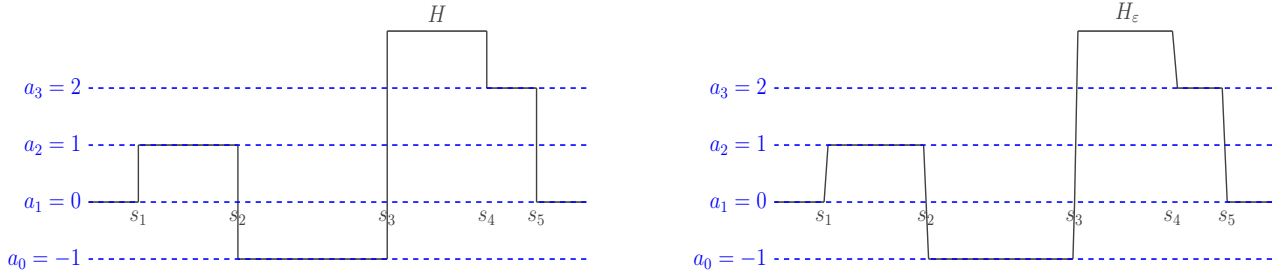
\begin{figure}[t]
    \begin{center}\resizebox{17cm}{3.75cm}{
        \begin{tikzpicture}
\draw[darkgray, thick] (-5,0) -- (-4,0) ;
\draw[darkgray, thick] (-4,1) -- (-4,0)node[below]{$s_1$};
\draw[darkgray, thick] (-4,1) -- (-2,1) ;
\draw[darkgray, thick] (-2,1) -- (-2,0) node[below]{$s_2$};
\draw[darkgray, thick] (-2,0) -- (-2,-1); 
\draw[darkgray, thick] (-2,-1) -- (1,-1);
\draw[darkgray, thick] (1,-1) -- (1,0)node[below]{$s_3$} ;
\draw[darkgray, thick] (1,0) -- (1,3) ;
\draw[darkgray, thick] (3,3) -- (1,3) node[above]{\hspace{2cm}$H$};
\draw[darkgray, thick] (3,2) -- (3,3) ;
\draw[darkgray, thick] (3,2) -- (4,2) (3,0)node[below]{$s_4$};
\draw[darkgray, thick] (4,2) -- (4,0)node[below]{$s_5$};
\draw[darkgray, thick] (4,0) -- (5,0) ;

\draw[blue, thick,dashed] (5,-1) -- (-5,-1)node[left,blue]{$a_0=-1$};
\draw[blue, thick,dashed] (5,0) -- (-5,0)node[left]{$a_1=0$};
\draw[blue, thick,dashed] (5,1) -- (-5,1)node[left]{$a_2=1$};
\draw[blue, thick,dashed] (5,2) -- (-5,2)node[left]{$a_3=2$};
        \end{tikzpicture}\hspace{2cm}\begin{tikzpicture}
\draw[darkgray, thick] (-5,0) -- (-4,0) ;
\draw[darkgray, thick] (-3.92,1) -- (-4,0)node[below]{$s_1$};
\draw[darkgray, thick] (-3.92,1) -- (-2,1) ;
\draw[darkgray, thick] (-2,1) -- (-1.95,0) node[below]{$s_2$};
\draw[darkgray, thick] (-1.95,0) -- (-1.9,-1); 
\draw[darkgray, thick] (-1.9,-1) -- (1,-1);
\draw[darkgray, thick] (1,-1) -- (1.033,0)node[below]{$s_3$} ;
\draw[darkgray, thick] (1.033,0) -- (1.1,3) ;
\draw[darkgray, thick] (3,3) -- (1.1,3) node[above]{\hspace{2cm}$H_\varepsilon$};
\draw[darkgray, thick] (3.1,2) -- (3,3) ;
\draw[darkgray, thick] (3.1,2) -- (4,2) (3,0)node[below]{$s_4$};
\draw[darkgray, thick] (4,2) -- (4.1,0)node[below]{$s_5$};
\draw[darkgray, thick] (4.1,0) -- (5,0) ;

\draw[blue, thick,dashed] (5,-1) -- (-5,-1)node[left,blue]{$a_0=-1$};
\draw[blue, thick,dashed] (5,0) -- (-5,0)node[left]{$a_1=0$};
\draw[blue, thick,dashed] (5,1) -- (-5,1)node[left]{$a_2=1$};
\draw[blue, thick,dashed] (5,2) -- (-5,2)node[left]{$a_3=2$};
        \end{tikzpicture}}
    \end{center}

\vspace{-0.5 cm}

    \caption{Smoothed Measure of Underqualification $H_\varepsilon$.} \label{f:tilt}
{\scriptsize \vspace{.2 cm} Figure \ref{f:tilt} illustrates the smoothing of the measure of underqualification $H$ displayed in the left panel. The corresponding smoothed measure of underqualification $H_\varepsilon$ is presented in the right panel.}
    
\end{figure}

Consider the support $A^\ell_{\varepsilon} := H_\varepsilon^{-1}((a_{\ell-1},a_{\ell}))$ for all layers $1\leq \ell \leq L$ and define the smoothed worker distribution $F^\ell_{\varepsilon}:= F_\varepsilon \big|_{A^\ell_{\varepsilon}}$ and the smoothed job distribution $G^\ell_{\varepsilon} := G_\varepsilon \big|_{A^\ell_{\varepsilon}}$ for every layer such that $F_\varepsilon = \sum F^\ell_{\varepsilon}$ and $G_\varepsilon=\sum G^\ell_{\varepsilon}$.\footnote{The choice of an open or closed interval $(a_{\ell-1}, a_{\ell} )$ does not matter because the inverse of the boundary points is negligible with respect to the measure $F_\varepsilon + G_\varepsilon$.} Moreover, let an optimal assignment between workers $F^\ell_{\varepsilon}$ and jobs $G^\ell_{\varepsilon}$ in layer $\ell$ be denoted by $\pi^\ell_{\varepsilon}$. Clearly, this assignment $\pi^\ell_{\varepsilon}$ is supported on the set $({A^\ell_{\varepsilon}})^2$. 

Next, we establish that the sum of optimal assignments across layers $\sum \pi^\ell_{\varepsilon}$ is an optimal assignment between the smoothed worker distribution $F_\varepsilon$ and the smoothed job distribution $G_\varepsilon$. Let $\pi_\varepsilon$ be some optimal assignment between $F_\varepsilon$ and $G_\varepsilon$. By cyclical monotonicity, the assignment $\pi_\varepsilon$ is concentrated on a support $\Gamma_\varepsilon$ that satisfies the property of no intersecting pairs. Since the smoothed distributions $F_\varepsilon$ and $G_\varepsilon$ are both atomless, this implies that any pairing $(x,z)\in \Gamma_\varepsilon$ where $x<z$ satisfies $F_\varepsilon([x,z])=G_\varepsilon([x,z])$. In turn, by the definition of the measure of underqualification $H$, this implies $H_\ee(x)=H_\ee(z)$ meaning that $x$ and $z$ are both part of the same layer $A^\ell_{\varepsilon}$. As a result, it follows that the support of the assignment $\pi_\varepsilon$ is contained in the union of the support of all layers, or $\Gamma_\varepsilon \subseteq \bigcup ({A^\ell_{\varepsilon}})^2$. Since all the supports $\{A^{\ell}_{\varepsilon}\}$ are disjoint, the assignment $\pi_\varepsilon\big|_{({A^\ell_{\varepsilon}})^2}$ transports between $F^{\ell}_{\varepsilon}$ and $G^{\ell}_{\varepsilon}$. Since $\pi^\ell_{\varepsilon}$ is an optimal assignment between workers $F^{\ell}_{\varepsilon}$ and jobs $G^{\ell}_{\varepsilon}$, it follows that the cost of mismatch for layer $\ell$ is greater under the assignment $\pi_\varepsilon\big|_{({A^\ell_{\varepsilon}})^2}$, that is, $\int c~\d \pi_\varepsilon \big|_{({A^\ell_{\varepsilon}})^2} \geq \int c \; \d \pi^\ell_{\varepsilon}$. By summing over all layers $1 \leq \ell \leq L$, we can write that
\begin{equation*}
\int c\, \d \pi_\varepsilon = \sum_{1\leq \ell\leq L} \int c\,\d \pi_\varepsilon \big|_{({A^\ell_{\varepsilon}})^2} \geq \sum_{1\leq \ell\leq L} \int c \, \d \pi^\ell_{\varepsilon}=\int c \, \d \Big( \sum_{1\leq \ell\leq L} \pi^\ell_{\varepsilon} ~ \Big).
\end{equation*}
Since $\sum \pi^\ell_{\varepsilon}$ is a feasible assignment between the smoothed distributions $F_\varepsilon$ and $G_\varepsilon$, and the mismatch cost is below the minimum mismatch cost, it follows that $\sum \pi^\ell_{\varepsilon}$ must be an optimal assignment.

To conclude the proof it follows from our construction and the stability of the optimal assignment that $\pi^\ell_{\varepsilon} \to \pi^\ell$ and $\sum \pi^\ell_{\varepsilon}\to\pi$ weakly. Thus,
\begin{equation*}
\pi=\lim_{\ee\to 0} \sum_{1\leq \ell\leq L} \pi^\ell_{\varepsilon} =\sum_{1\leq \ell\leq L}\lim_{\ee\to 0} \pi^\ell_{\varepsilon} = \sum_{1\leq \ell\leq L}\pi^\ell.
\end{equation*}

\subsection{Sorting Within a Layer} \label{s:withinlayer}

We construct a recursive characterization for an optimal assignment within a given layer. This recursive formulation reflects on the salient features of optimal sorting stemming from the concavity of the cost. We use the approach of \citet{Aggarwal:1995} that centers on the property of no intersecting pairs. Specifically, we adopt the recursive algorithm developed by \citet{Nechaev:2013}, designed to model statistical properties of polymer chains.\footnote{The properties of maximal number of perfect pairs, no intersecting pairs, and layering by themselves may be useful to construct simple algorithms to approximate optimal sorting mechanisms for concave costs. \citet{Caracciolo:2020} and \citet{Ottolini:2023}, for example, only use no intersecting pairs and layering to, respectively, construct a Dyck algorithm and greedy matching algorithm to study approximate optimal sorting for a random assignment problem. They show that the aggregate costs of skill gaps under the simple assignment scale similarly to the aggregate costs of skill gaps for the optimal assignment, that is, achieves optimum on average up to a scaling constant, in the asymptotic limit with the number of points tending to infinity.} 

The optimal assignment problem for a given layer is
an alternating assignment problem. By Birkhoff's theorem \citep{Birkhoff:1946}, an optimal assignment within a layer matches one worker with precisely one job. For notational convenience, we
order workers and jobs within each layer by their skill levels. Let there be $n_\ell$ workers and $n_\ell$ jobs in a given layer, and we denote the skill levels by $s_{1}<s_{2}<\dots<s_{2n_\ell-1}<s_{2n_\ell}$.



We write a Bellman equation to calculate the minimum aggregate cost of
skill gaps. The recursive component of the Bellman equation is that
we consider assignment problems with an increasing number of skill levels.
We start by solving all assignment problems between two consecutive
elements: the assignment problem between one worker and
one job. That is, we consider assignments between skill levels
$s_{i}$ and $s_{i+1}$, for each $i$. Using the solutions from the
previous step, we proceed to solve all assignment problems between four consecutive elements (two workers and two jobs) and so on. 



We denote by $V_{i,j}$ the minimum cost of mismatch when sorting all workers
and jobs with skill levels between $s_{i}$ and $s_{j}$ (inclusive), where $j>i$. The difference $j-i$ is odd so that
there are equal numbers of workers and jobs between $s_{i}$ and $s_{j}$.
Considering an assignment of workers and jobs with skill levels
in $[s_{i},s_{j}]$, the planner can pair the leftmost $s_{i}$ with
any  $s_{k}$ such that $k-i$ is odd. Upon pairing
$s_{i}$ with  $s_{k}$, the planner remains to optimally pair the workers and jobs in $[s_{i+1},s_{k-1}]$, and all workers and
jobs with skill levels in $[s_{k+1},s_{j}]$. The main observation that facilitates this characterization is that there are no pairings between these two segments because this violates the property of no intersecting pairs. Using the results from previous steps to obtain costs $V_{i+1,k-1}$ and $V_{k+1,j}$ delivers the Bellman equation: 
\begin{equation}
V_{i,j}=\min_{k\in\{i+1,i+3,\dots,j\}}\;\big( c(s_{i},s_{k})+V_{i+1,k-1}+V_{k+1,j}\big) \label{e:bellman_simple}
\end{equation}
with boundary conditions $V_{i+1,i}=0$ for all $i$.\footnote{The boundary conditions are invoked at either end of the choice interval.
When $k=i+1$, the minimum cost of mismatch is $c(s_{i},s_{i+1})+V_{i+2,j}$,
the cost of pairing the first worker to the first job, together with
optimally sorting all skill levels from $s_{i+2}$ to $s_{j}$. When
$k=j$, the minimum cost is $c(s_{i},s_{j})+V_{i+1,j-1}$, the cost
of pairing the first worker to the last job, together with optimally
sorting all intermediate skill levels between $s_{i+1}$ and $s_{j-1}$.}

Finally, we construct an optimal assignment. Starting from $V_{1,2n_{\ell}}$, the optimal pairing of skill $s_1$ is given by skill $s_k$ that solves equation (\ref{e:bellman_simple}). Then two corresponding continuation values, $V_{2,k-1}$ and $V_{k+1,2n_{\ell}}$, are evaluated to determine optimal pairings for skill $s_2$ and for skill $s_{k+1}$, respectively. This process of finding an optimal assignment continues until a full assignment is constructed.

\subsection{Efficiency Properties of the Dual Algorithm}\label{s:efficiency}

In this appendix, we analyze efficiency properties of the dual algorithm for empirical measures. By empirical we mean that  $X_1,\dots,X_{N}$ are random samples drawn independently from a uniform distribution on the unit interval $[0,1]$, and the workers are uniformly distributed on $\{X_1,\dots,X_{N}\}$, and similarly for the jobs. In this case, we further refine the bound for the runtime of our algorithm in \Cref{prop:hiereff}, as is shown in Proposition \ref{rem:runtime}. 

Recall our algorithm for the computation of the dual potentials from Appendix \ref{a:dualproof}. Suppose that $(x_1,z_1),\dots,(x_p,z_p)$ are ordered subpairs of pair $(x_0,z_0)$ in the optimal assignment $\pi$. Let $c_{ij}:=c(x_i,z_j)$. Then the system of inequalities, where for all $1\leq n<m\leq p$: 
    \begin{align}
        \label{eq:ineqs2} \max(c_{00}+c_{nn}-c_{0n}-c_{m0},c_{nn}-c_{mn})\leq \sum_{k=n+1}^m\beta_k\leq \min(c_{0m}+c_{n0}&-c_{00}-c_{mm},c_{nm}-c_{mm})
    \end{align}
    admits a solution $(\beta_2,\dots,\beta_p)$.

    Since the dual solution can be solved via standard linear programming, the worst-case runtime for our algorithm is $O(N^4)$.\footnote{See, for example, \citet{Boyd:2004}.}  Our algorithm is much more efficient as more layers of arcs are introduced. This is because compared to the standard linear programming, our algorithm solves the problem in the order from bottom arcs to top, while at each step the values of $\phi$ in the hidden arcs need not be computed again, but only adjusted with constant factors. Typically, the number $p$ will not be as large as $N$. The following proposition provides a general upper bound of the number $p$, which is a consequence of the absence of intersecting pairs. Define the number of crossings of the measure of underqualification $H$ at level $\tau \in \R$ as
\begin{align*}
    C_H(\tau):=\sum_{1\leq k\leq N}\bone_{\{H(x_k)=\tau\}}+\sum_{1\leq k\leq N}\bone_{\{H(z_k)=\tau\}}.
\end{align*}

\begin{proposition}\label{prop:cross}
    Suppose that $(x_1,z_1),\dots,(x_p,z_p)$ are ordered subpairs of the pair $(x_0,z_0)$ in the optimal assignment $\pi$. Then there exists $\tau \in \mathbb{R}$ such that the measure of underqualification $H$ crosses the level $\tau$ for $p$ times, that is, $C_{H}(\tau)\geq p$.
\end{proposition}


\begin{proof}
    Consider numbers $t_i\in(\max(x_i,z_i),\min(x_{i+1},z_{i+1}))$ for $1\leq i<p$. By the property of no intersecting pairs, $H(t_i)$ is constant in $i$. On the other hand, $H$ cannot be constant on the interval $[t_i,t_{i+1}]$. The claim thus follows.
\end{proof}

\begin{proposition}\label{rem:runtime}
    Suppose that $F_N,G_N$ are independent empirical measures of the uniform distribution on $[0,1]$. Then the runtime of the algorithm is $O(N^{2.5}(\log\log N)^{3/2})$ almost surely. 


\end{proposition}

In order to prove Proposition \ref{rem:runtime}, consider the (random) empirical cumulative densities $F_N,G_N$, drawn from two independent sequences $\{X_i\}_{1\leq i\leq N}$ and $\{Z_i\}_{1\leq i\leq N}$ uniformly in $[0,1]$, i.e.,
\begin{align*}
    F_N(t)=\frac{1}{N}\sum_{k=1}^N\bone_{\{X_k\leq t\}}\qquad \text{ and }\qquad G_N(t)=\frac{1}{N}\sum_{k=1}^N\bone_{\{Z_k\leq t\}}.
\end{align*}
It is well known that the scaled measure of underqualification $\sqrt{N}(F_N-G_N)$ can be well approximated by a Brownian bridge, where we recall that a (standard) Brownian bridge $B=\{B(t)\}_{t\in[0,1]}$ is a centered Gaussian process with covariance $\E[B(s)B(t)]=\min(s,t)-st$. 
We denote the local time of a standard Brownian bridge  $B$ on $[0,1]$ at $x\in\R$ by $L_B(x)$. By definition, the local time process $\{L_B(x)\}_{x\in\R}$ is such that for any bounded Borel function $f$,
\begin{align*}
    \int_0^1 f(B(t))\,\d t=\int_\R f(x)L_B(x)\,\d x.
\end{align*}
The following Lemma is a special case of Theorem 5 of \citet{khoshnevisan1992level}.
\begin{lemma}[Theorem 5 of \citet{khoshnevisan1992level}]\label{lemma:localtime} There exists a suitable probability space carrying $F_N,G_N$, and a sequence of Brownian bridges $\{B_N\}$, such that 
    \begin{align*}
        \lim_{N\to\infty} \max_{k\in\mathbb{Z}}\left|N^{-1/2}C_{H_N}\left(\frac{k}{N}\right)-\sqrt{2} L_{B_N}\left(\frac{k\sqrt{2}}{\sqrt{N}}\right)\right|=O(N^{-0.24})\ \text{ a.s.}
    \end{align*}
\end{lemma}
We also have the following Lemma on fluctuations of the local time for Brownian bridges. This is taken from Lemma 3.2 of  \citet{bass1995laws} applied with $n_k=k$ and $\ee_k=\sqrt{2/k}$ therein.
\begin{lemma}\label{lemma:BBlocaltime}
    Let $\{B_N\}$ be any sequence of Brownian bridges.  It holds that
    \begin{align*}
        \sup_{|x-y|<\sqrt{2/N}}|L_{B_N}(x)-L_{B_N}(y)|=O(N^{-0.24})\ \text{ a.s.}
    \end{align*}
\end{lemma}
With a Borel-Cantelli argument in \citet{csorgHo1999some} applied to the sequence of Brownian bridges $\{B_N\}$ (with the tail estimates supplied by Theorem 5.1 therein), the following lemma can be similarly established as Theorem 1.4 of \citet{csorgHo1999some}.
\begin{lemma}
    \label{lemma:BBintegral}Let $\{B_N\}$ be any sequence of Brownian bridges. There is a constant $C>0$ such that
\begin{align*}
    \bP\left(\int_{\R}L_{B_N}(x)^4\d x>y\right) \leq \exp\left(-\frac{y^{2/3}}{C}\right).
\end{align*}
    Moreover,
    \begin{align*}
        \int_{\R}L_{B_N}(x)^4\d x=O\left((\log\log N)^{3/2}\right)\ \text{ a.s.}
    \end{align*}
\end{lemma}
\begin{proof}
The first claim is Theorem 5.1 of \citet{csorgHo1999some} applied with $p=4$ therein. The second claim can be proved in a similar way to (3.7a) of \citet{bass1995laws}.
\end{proof}

\begin{proof}[Proof of Proposition \ref{rem:runtime}]
 Recall that solving \eqref{eq:ineqs2} has complexity $O(p^4)$. In view of Proposition \ref{prop:cross}, the runtime of our algorithm has the upper bound
\begin{align*}\sum_{|k|\leq N}C_{F-G}\left(\frac{k}{N}\right)^4.
\end{align*}
Using Lemma \ref{lemma:localtime}, we get that almost surely,
\begin{align*}
    C_{F-G}\left(\frac{k}{N}\right)=\sqrt{2N}L_{B_N}\left(\frac{\sqrt{2}k}{\sqrt{N}}\right)+O(N^{0.26}).
\end{align*}
    Therefore, using the elementary inequality $(A+B)^4\leq 16(A^4+B^4)$ we have almost surely,
    \begin{align*}
        \sum_{|k|\leq N}C_{F-G}\left(\frac{k}{N}\right)^4&\leq 64 \sum_{|k|\leq N}\left(N^2L_{B_N}\left(\frac{\sqrt{2}k}{\sqrt{N}}\right)^4+O(N^{1.04})\right)\\
        &\leq O(N^{2.04})+1024N^{2.5}\left(\int_{-\sqrt{N}}^{\sqrt{N}+1/\sqrt{N}} L_{B_N}(x)^4\d x+\sum_{|k|\leq N}\frac{1}{\sqrt{N}}O(N^{-0.96})\right)\\
        &\leq O(N^{2.04})+1024N^{2.5}\int_{\R}L_{B_N}(x)^4\d x,
    \end{align*}where we applied Lemma \ref{lemma:BBlocaltime} in the second inequality. Applying Lemma \ref{lemma:BBintegral}
 concludes the proof.    
\end{proof}


\subsection{Further Details of Quantitative Results} \label{s:intquant}

In this appendix we present the parameters of the quantitative model and intuition for the quantitative results in Section \ref{sec:quantitative}.

\subsubsection{Model Paramaters}

\Cref{t:modelparam} summarizes the model calibration. The first column displays the model parameter, while the second and third columns display the parameter values for 1980 and 2005.

\begin{table}[!t]
\global\long\def\arraystretch{1.35}%
\begin{centering}
\caption{Model Parameters}
\label{t:modelparam} %
\begin{tabular}{l|cc}
\hline \hline 
Parameter & \multicolumn{2}{c}{Values} \\
& 1980 & 2005 \\
\hline 
Mean of job difficulties $\mu_{1}$ & 0.38 & 0.42  \\
Variance of job difficulties $\sigma^2_{1}$  \hspace{1.75 cm}  & \hspace{1.75 cm} 0.06 \hspace{1.75 cm}  & \hspace{1.75 cm} 0.03  \hspace{1.75 cm} \\
Mean of job difficulties $\mu_{2}$ & 0.00 & $-0.12$ \hspace{0.2 cm} \\
Variance of job difficulties $\sigma^2_{2}$ & 0.75 & 0.51 \\
Mixing weight $p$ & 0.36 & 0.38 \\ 
Variance of worker skills $\sigma_x^2$ & 0.20 & 0.36 \\
\hline \hline 
\end{tabular}
\par\end{centering}
{\scriptsize \vspace{0.5cm}
\Cref{t:modelparam} summarizes the model calibration. The first column displays the model parameter, while the second and third columns display the parameter values for 1980 and 2005.}
\end{table}

\subsubsection{Intuition for Quantitative Results}

We provide intuition for the quantitative results by analyzing the equilibrium for 1980 in more detail. In \Cref{f:workerjobdist}, we plot the distributions of workers and jobs in 1980 implied by the parameters in \Cref{t:modelparam}. As shown in \Cref{s:principles}, the optimal assignment is constructed by first forming perfect pairs. Workers and jobs that are perfectly paired are indicated by the shaded area in the left panel of  \Cref{f:workerjobdist}. Workers in this region are sorted positively into an occupation with a complexity level that perfectly matches their skill. All high-skill workers, workers with skill levels above 0.2, are perfectly positively paired to the most complex jobs. Similarly, all lowest-skill workers, the workers with skill levels below $-0.9$, are positively paired with the least complex jobs. Overall, about two-thirds of workers and jobs are perfectly paired. The remaining workers and jobs are mismatched.

\begin{figure}[!t]
\begin{centering}
\includegraphics[width=0.45\textwidth,height=0.26\textheight ]{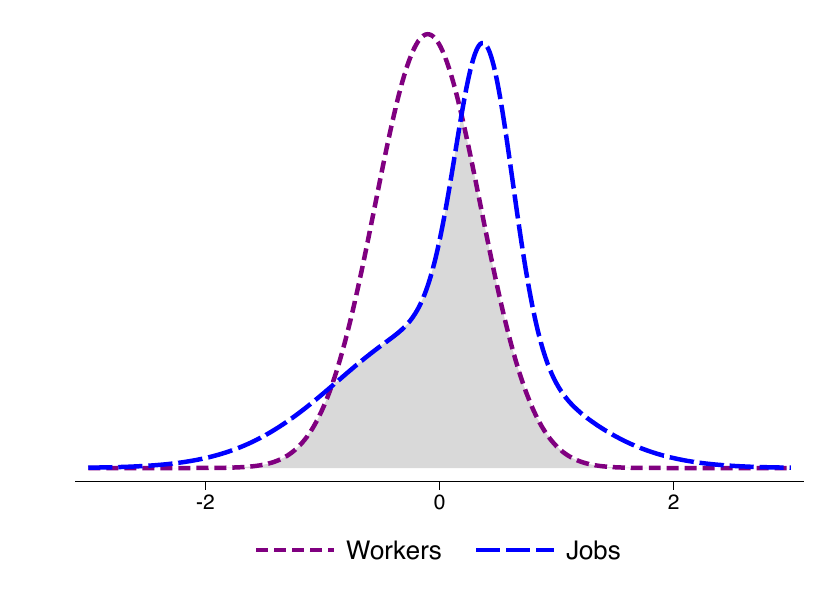}
\includegraphics[width=0.48\textwidth,height=0.26\textheight ]{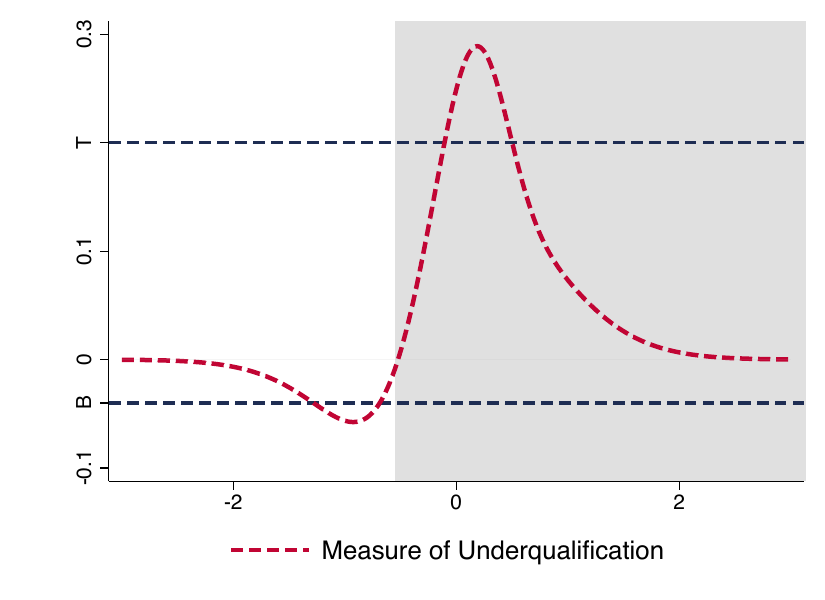}
\par\end{centering}
\vspace{-0.4cm}

\caption{Distributions of Workers and Jobs and the Measure of Underqualification}
\label{f:workerjobdist} 

{\scriptsize{}{}\vspace{0.2cm}
The left panel of \Cref{f:workerjobdist} shows the distribution of workers (in purple) and the distribution of jobs (in blue) for our quantitative analysis. The right panel shows the implied measure of underqualification. An increasing measure of underqualification indicates mismatched workers; a decreasing measure of underqualification indicates mismatched jobs.} 
\end{figure}

We next describe mismatched workers and jobs. The right panel of \Cref{f:workerjobdist} plots the measure of underqualification induced by the worker and job distributions. An increasing measure of underqualification indicates mismatched workers; a decreasing measure of underqualification indicates mismatched jobs. The figure indicates that mismatched workers are in the middle of the skill distribution, with log skill levels between $-0.9$ and 0.2. Mismatched jobs are jobs in occupations with low complexity, below $-0.9$, or high complexity, above 0.2. The shaded area indicates the skills for which the measure of underqualification is positive, while the non-shaded part indicates the skill ranges where it is negative.

In order to characterize the optimal assignment of mismatched workers and jobs, we decompose the measure of underqualification into layers illustrated by the blue dashed lines in the right panel of \Cref{f:workerjobdist}. By \Cref{t:numberpair}, it follows that each layer contains at most two pairs. Sorting in layers corresponding to the non-shaded area, such as layer $B$, is simple as it contains a single worker and job, which necessarily are paired. The worker with skill level $-0.7$ is paired to the job with complexity $-1.2$. Workers in the bottom layers are paired with jobs for which they are overqualified. The top layers, such as layer $T$, also contain one worker and one job. In each top layer in \Cref{f:workerjobdist}, where the measure of underqualification is positive, the worker is paired with jobs for which they are underqualified.\footnote{The optimal assignment is identical for all $\zeta_{p}, \zeta_{u} \in (0,1)$ because the measure of underqualification contains a single worker and a single job within each layer of the measure of underqualification.}

The optimal assignment of mismatched workers and jobs features significant variation in mismatch. Mismatch ranges from small, for workers at skills slightly below 0.1 and jobs at skill slightly above $0.1$, to large, for example, between the worker at skill $-0.4$ and the job with difficulty 1.2.

In order to characterize the nature of mismatch between workers and jobs, consider first the region of underqualification captured by the shaded area in the right panel of \Cref{f:workerjobdist}. Workers with skills between $-0.5$ and 0.2 are sorted negatively to high complexity jobs. For example, the worker in the top layer marked by $T$ with skill $-0.1$ works in occupation 0.5, while worker 0.1 works in occupation 0.3. Negative sorting in this region implies that more complex jobs feature larger investments. Workers with lower skills are paired with jobs with higher complexity meaning that the skill gaps $z-x$ between the worker and the job is larger. Since the technology choice increases in mismatch, investments are larger. To ensure that the value of the complex job is not significantly diminished, a larger investment is made. In the region of overqualification, indicated by the non-shaded area in the right panel of \Cref{f:workerjobdist}, mismatched workers are more qualified than the jobs require and firms provide amenities for more skilled employees to reduce their utility cost of mismatch.

The equilibrium features composite sorting. First, distinct worker types work in the same occupation. For example, both a worker with skill 0.1 and a worker with skill 0.3 are assigned to occupation 0.3. The perfectly positively sorted worker has skill 0.3 and there is no skill gap for this worker. The worker with skill 0.1 is assigned to occupation 0.3 through the top dashed layer in the right panel of \Cref{f:workerjobdist} and there is mismatch for this lower-skill worker. Since distinct workers work in the same occupation, the equilibrium features wage dispersion within occupations.\footnote{In line with the predictions of our framework, \citet{Bayer:2023} argue empirically that differences in job execution in terms of responsibility and autonomy within the same occupation can account for a sizable portion of observed wage differences.}  Second, the same occupation is part of both positive and negative sorting. For example, occupation 0.3 is a part of positive sorting (with worker 0.3 as a part of positive sorting with perfect pairs) and negative sorting (with worker 0.1 as a part of negative sorting of medium-skilled workers with high complexity jobs). 


\subsection{Preservation of Concealed Pairs}

We show another feature of optimal sorting, which is the preservation of concealed pairs. A pair $(x,z)$ within an assignment is labeled concealed when the interval $(x,z)$ is strictly contained within an interval $(x',z')$ corresponding to some other pair $(x',z')$ within the same assignment. The pair $(x_{2},z_{1})$ is concealed since the interval $(z_{1},x_{2})$ is contained within the interval corresponding to the pairing $(x_{1},z_{2})$.

The next principle establishes that within each layer, every concealed pair is preserved, a term which we define precisely in the formulation of Lemma \ref{l:hidden} following \citet{Delon:2012}. We provide a simple proof of this result, which we extend to importantly allow for asymmetric costs of skill gaps.

\noindent 

\begin{lemma}{\textit{Preservation of Concealed Pairs}.}\label{l:hidden}
Consider any interval $\mathcal{I}$ that has a balanced number of
workers and jobs in a layer $\ell$. If, in an optimal assignment
between $F_{\ell}$ and $G_{\ell}$ restricted to the interval $\mathcal{I}$,
a pair $(x_{i},z_{j})$ is concealed then it is optimal in the full
assignment between $F_{\ell}$ and $G_{\ell}$.
\end{lemma}

\begin{proof}
Consider the interval $\mathcal{I}=[x_{I},z_{J}]$, the measures of workers $F_1 = F_\ell|_{\mathcal{I}}$ and jobs $G_1= G_\ell|_{\mathcal{I}}$, and, additionally the measures of workers $F_2 = F_\ell + F_1$ and jobs $G_2=G_\ell+G_1$, and let an optimal assignment between workers $F_i$ and jobs $G_i$ be given by $\pi_i$, for $i\in\{1,2\}$. By Lemma \ref{l:layer}, the optimal assignment between the measure of workers $F_2$ and the measure of jobs $G_2$ is the sum of the optimal assignments for each layer, or $\pi_2 = \pi_1 + \pi_\ell$.

We prove the result by contradiction. Suppose there is a concealed pair in the assignment $\pi_1$ that is not preserved in the assignment $\pi_\ell$. This means there is at least one pair $(x_i,z_j)$ in $\pi_1$ such that some skill level $s_l \in (x_i,z_j)$ is connected to some skill level $s_k$ outside this interval in the assignment $\pi_\ell$. This is represented by the blue arrow in Figure \ref{f:hiddencontra}. Otherwise, by replacing the assignment $\pi_\ell$ on $[x_i,z_j]$ by the assignment $\pi_1$ on $[x_i,z_j]$ decreases the total cost. This means that the corresponding intervals $[x_i,z_j]$ and $[s_l,s_k]$ intersect in $\pi_2$, which violates the property of no intersecting pairs applied to $\pi_2$.\end{proof}


\begin{figure}[!t]
    \begin{center}\resizebox{17cm}{2.8cm}{
        \begin{tikzpicture}
\draw[darkgray, thick] (-9,1) -- (-8,1) ;
\draw[darkgray, thick] (-9,1) -- (-9,0)node[below]{$x_1$};
\draw[darkgray, thick] (-8,1) -- (-8,0)node[below]{$z_1$} (-7.5,0.5)node[right]{$\dots$};

\draw[darkgray, thick] (-6,1) -- (-5,1);
\draw[darkgray, thick] (-6,1) -- (-6,0);
\draw[darkgray, thick] (-5,1) -- (-5,0) node[below]{$z_{I-1}$};

\draw[darkgray, thick] (-4,2) -- (-3,2);
\draw[darkgray, thick] (-4,2) -- (-4,0) node[below]{$x_{I}$};
\draw[darkgray, thick] (-3,0) -- (-3,2)(-0.5,0.5)node[right]{$\dots$};

\draw[darkgray, thick] (-2,2) -- (-1,2);
\draw[darkgray, thick] (-2,2) -- (-2,0) node[below]{$x_i$};
\draw[darkgray, thick] (-1,0) -- (-1,2);
\draw[darkgray, thick] (2,2) -- (1,2);
\draw[darkgray, thick] (2,2) -- (2,0)node[below]{$z_j$} ;
\draw[darkgray, thick] (1,0) -- (1,2);

\draw[darkgray, thick] (9,1) -- (8,1);
\draw[darkgray, thick] (9,1) -- (9,0) node[below]{$z_{n}$};
\draw[darkgray, thick] (8,1) -- (8,0) node[below]{$x_{n}$};

\draw[darkgray, thick] (6,1) -- (5,1);
\draw[darkgray, thick] (6,1) -- (6,0);
\draw[darkgray, thick] (5,1) -- (5,0) node[below]{$x_{J+1}$}(6.5,0.5)node[right]{$\dots$};

\draw[darkgray, thick] (4,2) -- (3,2);
\draw[darkgray, thick] (4,2) -- (4,0) node[below]{$z_{J}$};
\draw[darkgray, thick] (3,0) -- (3,2);

\draw[darkgray, thick,dashed] (-10,0) -- (10,0);
\draw[blue, thick,dashed] (-10,1) -- (10,1) (-0.5,0)node[below]{$s_l$} (7,0)node[below]{$s_k$};
\path[<-,every node/.style={font=\sffamily\small}] (2,2) edge[bend right=30] node [left] {} (-2,2);
\path[<-,every node/.style={font=\sffamily\small}, blue] (7,1) edge[bend right=15] node [left] {} (-0.5,1);
        \end{tikzpicture}}
    \end{center}

\vspace{-0.5 cm}

    \caption{Proof for Preservation of Concealed Pairs}
    {\scriptsize \vspace{.2 cm}  The black arrow indicates the exposed pair $(x_i,z_j)$ in the assignment $\pi_1$; the blue arrow indicates some pair $(s_l,s_k)$ in the assignment $\pi$ where $s_l\in(x_i,x_j)$. Since the arcs corresponding to these pairings cross, we obtain a contradiction.}
    
    \label{f:hiddencontra}
\end{figure}
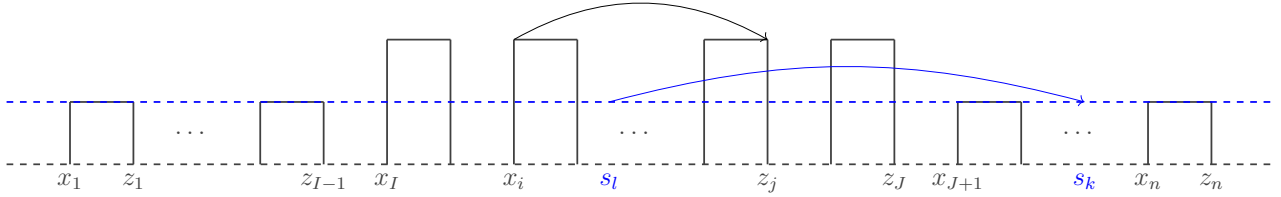

\subsection{Comparative Statics of Sorting} \label{s:andersonsmith}

In their work on the comparative statics of sorting, \citet{Anderson:2022} provides sufficient conditions for sorting to be more positive as the output function changes. However, their sufficient conditions do not apply in our setting. To show this, we introduce some of their definitions, formulate their sufficient conditions, and provide a counterexample.

A central object in their work is the difference between output under positive and negative sorting, which they call synergy. By definition, a rectangle $r$ is a combination of two workers $(x_1,x_2)$ where $x_1 < x_2$ and two jobs $(z_1,z_2)$ where $z_1 < z_2$, so their representation in $\mathbb{R}^2$ is a rectangle. Rectangular synergy $S(r;\zeta)$ is the synergy inside the rectangle given by $y(x_1,z_1;\zeta) + y(x_2,z_2;\zeta) - y(x_1,z_2;\zeta) - y(x_2,z_1;\zeta) $, where $\zeta$ emphasizes the dependence of output on the concavity of the mismatch costs function. Summed rectangular synergy sums synergies on any finite set of disjoint rectangles. One of the assumptions that \citet{Anderson:2022} requires is that the summed rectangular synergy is up-crossing in $\zeta$, where a function $\Upsilon$ is up-crossing if $\Upsilon(\zeta) \geq 0$ implies that $\Upsilon(\zeta') \geq 0$ for all $\zeta' \geq \zeta$.

\begin{figure}[!t]
\begin{centering}
\resizebox{11cm}{5.5cm}{    \begin{tikzpicture}
\node[circle,draw, minimum size=0.1pt,scale=0.6,label=below:{$x_1$}] (A2) at  (8,0){};
\node[circle,draw, minimum size=0.1pt,scale=0.6,label=below:{$x_2$}] (A) at  (10.5,0){};

\node[circle,fill=black,draw, minimum size=0.1cm,scale=0.6,label=below:{$z_1$}] (D) at  (11.5,0) {} ;

\node[circle,fill=black,draw, minimum size=0.1cm,scale=0.6,label=below:{$z_2$}] (D2) at  (14,0) {} ;
\node[circle,draw, minimum size=0.1pt,scale=0.6,label=below:{$x_3$}] (E3) at  (13,0){};
\node[circle,fill=black,draw, minimum size=0.1cm,scale=0.6,label=below:{$z_3$}] (72) at  (19,0) {} ;

\draw[dotted] (A2)--(A);\draw[dotted] (A) -- (E3);
\draw[dotted] (E3) -- (72);

\path[dashed,blue,every node/.style={font=\sffamily\small}] (D2) edge[bend left=50] node [below] {\color{black}{}} (E3);
\path[dashed,blue,every node/.style={font=\sffamily\small}] (A2) edge[bend right=50] node [below] {\color{black}{$39.9$}} (72);
\path[-,blue,every node/.style={font=\sffamily\small}] (72) edge[bend right=50] node [below] {\color{black}{$20$}} (E3);
\path[-,blue,every node/.style={font=\sffamily\small}] (A2) edge[bend left=50] node [below] {\color{black}{$20$}} (D2);

\path[-,orange,every node/.style={font=\sffamily\small}] (E3) edge[bend left=50] node [below] {\color{black}{$0.1$} } (D2);
\path[-,orange,every node/.style={font=\sffamily\small}] (A) edge[bend left=50] node [above] {\color{black}{$0.1$}} (D);

\path[dashed,orange,every node/.style={font=\sffamily\small}] (D) edge[bend right=30] node [below] {\color{black}{$1.8$}} (E3);
\path[dashed,orange,every node/.style={font=\sffamily\small}] (D2) edge[bend left=80] node [below] {\color{black}{$2$}} (A);
       
   \end{tikzpicture}    }\hspace{0.5cm}
\resizebox{0.28\textwidth}{0.28\textwidth}{\begin{tikzpicture}
    \begin{axis}[axis lines=middle,
            enlargelimits,
            ytick=\empty,
            xtick=\empty,]
\addplot[name path=3,domain={2:4}] {1} ;
\addplot[name path=32,domain={0:4}] {0} ;
\addplot[name path=5,domain={1:4}] {4};
\addplot[name path=6,domain={1:4}] {6};

 \addplot[ fill=orange,opacity=0.6]fill between[of=3 and 5, soft clip={domain=2:4}];
  \addplot[ fill=blue,opacity=0.6]fill between[of=5 and 6, soft clip={domain=1:4}];

\addplot[thick, samples=50, smooth,domain=0:6, name path=three] coordinates {(2,4)(2,1)};
\addplot[thick, samples=50, smooth,domain=0:6, name path=three] coordinates {(4,6)(4,1)};
\addplot[thick, samples=50, smooth,domain=0:6, name path=three] coordinates {(1,4)(1,6)};

\node at (axis cs:-0.2,1){\large $
z_1$};
\node at (axis cs:-0.2,4){\large $
z_2$};
\node at (axis cs:-0.2,6){\large $
z_3$};
\node at (axis cs:1,-0.3){\large $
x_1$};
\node at (axis cs:2,-0.3){\large $
x_2$};
\node at (axis cs:4,-0.3){\large $
x_3$};

\end{axis}
    \end{tikzpicture}
   
        } 
\par\end{centering}
\caption{Summed Rectangular Synergy is not One-Crossing with Concave Costs of Skill Gaps}
\label{f:anderson} {\scriptsize{}{}\vspace{0.2cm}
Figure \ref{f:anderson} shows that the condition of summed rectangular synergy is not satisfied in our setting with concave costs of skill gaps. We construct two rectangles to show that summed rectangular synergy is not one-crossing. Each of the rectangles contains two workers and two jobs. The blue rectangle consists of workers $(x_1,x_3)$ and jobs $(z_2,z_3)$ while the orange rectangle consists of workers $(x_2,x_3)$ and jobs $(z_1,z_2)$. In the left panel, positive sorting for these rectangles is represented by solid arcs above the dotted line, while negative sorting is captured by blue dashed arcs below the dotted line. Distances between the paired workers and jobs are shown by the numbers on the arc (not in scale). Summed rectangular synergy equals approximately  $\{ 0.09, -0.19, 0.30 \}$ for $\zeta = \{ 0.2,0.5,0.8 \}$, showing summed rectangular synergy is not one-crossing as a function of the concavity of the output function $\zeta$.}
\end{figure}

To show that their assumption does not hold in our environment, we give an example that summed rectangular synergy is neither up-crossing nor down-crossing ($-\Upsilon$ is up-crossing) in $\zeta$. Consider the case where the cost function (\ref{eq:cxz}) is symmetric in terms of the concavity of the mismatch function $\zeta = \zeta_p = \zeta_u$. Consider the example in Figure \ref{f:anderson} where we have a subset of three workers $(x_1,x_2,x_3)$ and three jobs $(z_1,z_2,z_3)$ of the alternating assignment problem within the layer, which are ordered such that $x_1 < x_2 < z_1 < x_3 < z_2 < z_3$. The distances $|x-z|$ between the workers and the jobs are indicated by the numbers on the arcs. 

We consider two distinct rectangles. The blue rectangle consists of workers $(x_1,x_3)$ and jobs $(z_2,z_3)$. Positive sorting for this rectangle is represented by the blue solid arcs above the dotted line, and negative sorting for this rectangle is captured by the blue dashed arcs below the dotted line. Similarly, the orange rectangle consists of workers $(x_2,x_3)$ and jobs $(z_1,z_2)$. Positive sorting for this rectangle is represented by the orange solid arcs above the dotted line, and negative sorting for this rectangle is captured by the orange dashed arcs below the dotted line. The rectangles are represented in the right panel of Figure \ref{f:anderson}. Rectangular synergy measures the difference in costs of skill gaps under positive sorting and the costs of skill gaps under negative sorting in the rectangle. The summed rectangular synergy over the blue and the orange rectangles is the sum of the synergies on the disjoint rectangles. 

A numerical example shows directly that summed rectangular synergies are neither up-crossing nor down-crossing in our setting. Specifically, synergy for the blue rectangle is given by $39.9^\zeta + 0.1^\zeta - 20^\zeta - 20^\zeta$ while synergy for the orange rectangle is $1.8^\zeta + 2^\zeta - 0.1^\zeta - 0.1^\zeta$. The resulting summed synergies equal approximately $\{ 0.09, -0.19, 0.30 \}$ for $\zeta = \{ 0.2,0.5,0.8 \}$ respectively, meaning that the summed rectangular synergy is neither up-crossing nor down-crossing as a function of the concavity of the output function $\zeta$.

\end{document}